\documentclass[a4paper, 11pt, oneside]{Thesis}  
\graphicspath{Figures/}  

\usepackage[square, numbers, comma, sort&compress]{natbib}  
\usepackage{verbatim}  
\usepackage{vector}  
\hypersetup{urlcolor=blue, colorlinks=true}  

\usepackage{amsthm}  

\usepackage{graphics} 
\usepackage{epsfig} 
\usepackage{amsmath,amsfonts,amssymb,amscd}

\DeclareMathAlphabet{\mathpzc}{OT1}{pzc}{m}{it}

\usepackage{algorithm,algorithmic}
\usepackage{float}
\usepackage{epsfig}
\usepackage{graphicx}
\usepackage{graphics}
\usepackage{url}
\usepackage{verbatim}
\usepackage{xspace}
\usepackage[usenames,dvipsnames]{color}
\usepackage{setspace}
\usepackage{multicol}

\floatstyle{ruled}
\newfloat{model}{H}{mod}
\floatname{model}{\footnotesize Model}
\newfloat{notatio}{H}{not}
\floatname{notatio}{\footnotesize Notation}

\newenvironment{varalgorithm}[1]
  {\algorithm\renewcommand{\thealgorithm}{#1}}
  {\endalgorithm}

\usepackage{url,changebar,bm,xspace,dsfont}
\let\mathbb=\mathds 

\newtheorem{defn}{\bfseries Definition}
\newtheorem{prop}{\bfseries Proposition}

\begin{document}
\frontmatter      

\title  {DISTRIBUTED WEIGHT BALANCING \\ IN DIRECTED TOPOLOGIES}
\authors  {\texorpdfstring
			{\href{https://scholar.google.com/citations?user=DGh25F0AAAAJ&hl=en&authuser=1}
			{\bf APOSTOLOS I. RIKOS}}
            {APOSTOLOS I. RIKOS}
            }
\addresses  {\groupname\\\deptname\\\univname}  
\date       {\textbf{MAY 2018}}
\subject    {}

\maketitle

\setstretch{1.3}  

\fancyhead{}  
\rhead{\thepage}  
\lhead{}  

\pagestyle{fancy}  

\committee{
\addtocontents{toc}{\vspace{1em}}  

\vspace{1cm}

\textbf{Doctoral Candidate:} \textbf{Apostolos I. Rikos} 

\vspace{0.5cm}

\textbf{Doctoral Thesis Title:} \textbf{Distributed Weight Balancing in Directed Topologies} 

\vspace{0.7cm}

The present Doctoral Dissertation was submitted in partial fulfillment of the requirements for the Degree of Doctor of Philosophy at the \textbf{Department of Electrical and Computer Engineering} and was approved on the $23^{rd}$ of April, $2018$ by the members of the \textbf{Examination Committee}. 

\vspace{0.5cm}

\textbf{Examination Committee:}

\vspace{0.4cm}

\textbf{Research Supervisor:}
\rule[0em]{27em}{0.5pt}  

\vspace{-0.6cm}
\begin{flushright}
(Dr. Christoforos N. Hadjicostis, Professor)
\end{flushright}

\vspace{0.2cm}

\textbf{Committee Member:}
\rule[0em]{27em}{0.5pt}  

\vspace{-0.6cm}
\begin{flushright}
(Dr. Georgios Ellinas, Professor)
\end{flushright}

\vspace{0.2cm}

\textbf{Committee Member:}
\rule[0em]{27em}{0.5pt}  

\vspace{-0.6cm}
\begin{flushright}
(Dr. Ioannis Krikidis, Associate Professor)
\end{flushright}

\vspace{0.2cm}

\textbf{Committee Member:}
\rule[0em]{27em}{0.5pt}  

\vspace{-0.6cm}
\begin{flushright}
(Dr. Themistoklis Charalambous, Assistant Professor)
\end{flushright}

\vspace{0.2cm}

\textbf{Committee Member:}
\rule[0em]{27em}{0.5pt}  

\vspace{-0.6cm}
\begin{flushright}
(Dr. Gabriele Oliva, Assistant Professor)
\end{flushright}


}
\clearpage  

\Declaration{

\addtocontents{toc}{\vspace{1em}}  

\vspace{0.5cm}

The present doctoral dissertation was submitted in partial fulfillment of the requirements for the degree of Doctor of Philosophy of the University of Cyprus. It is a product of original work of my own, unless otherwise mentioned through references, notes, or any other statements.

\vspace{2cm}

..................................................................[Full Name of Doctoral Candidate]

\vspace{0.6cm}

..................................................................[Signature]

%
%
}
\clearpage  

\addtotoc{Abstract in English}  
\abstract{
\addtocontents{toc}{\vspace{1em}}  

A distributed system or network can be viewed as a set of subsystems that can share information via interconnection links, which form a generally directed communication topology. Distributed systems prove to be of vital importance for the effectiveness of performing various tasks in the areas of cooperative control, distributed coordination, and control of multicomponent systems. This doctoral thesis concerns novel distributed algorithms for weight balancing over directed (communication) topologies. 
A directed topology (digraph) with nonnegative (or positive) weights assigned on each edge is weight-balanced if, for each node, the sum of the weights of in-coming edges equals the sum of the weights of out-going edges.
The novel algorithms introduced in this thesis can facilitate the development of strategies for generating weight balanced digraphs, in a distributed manner, and find numerous applications in coordination and control of multi-component systems.

In the first part of this thesis, we address the problem of integer weight balancing in a multi-component system. 
We introduce a novel distributed algorithm that operates over a static topology and solves the weight balancing problem when the weights are restricted to be nonnegative integers. 
The proposed algorithm is shown to converge to a weight balanced digraph after a finite number of iterations that we explicitly bound. 
This algorithm can also be viewed as a distributed method for obtaining a set of integer flows that balance a flow network. 

In the second part of the thesis, we investigate the problem of integer weight balancing in a multi-component system under a directed (static) interconnection topology in the presence of bounded or unbounded delays (packet drops) in the communication links. 
Specifically, we present a novel distributed algorithm which solves the integer weight balancing problem in the presence of arbitrary (time-varying and inhomogeneous) delays that might affect the transmission at a particular link at a particular time. 
Then, we present an event-based version of the proposed protocol in which each node autonomously decides when communication and control updates should occur. 
In the presence of packet drops over the communication links, the algorithm can be modified to converge to a set of weights that form a balanced graph after a finite number of iterations (with probability one). 
In all the above cases, the resulting weight balanced digraph is shown to be unique and independent on how delays or packet drops manifest themselves during the execution of the algorithm. 

In the third part of this thesis, we investigate the problem of integer weight balancing in a multi-component system under a static directed interconnection topology in the presence of lower and upper limit constraints on the edge weights. 
We present a novel distributed algorithm for obtaining admissible and balanced integer weights for the case when there are lower and upper weight constraints on the communication links. 
Compared with the distributed algorithms mentioned earlier, the additional constraint here is that each edge weight has to lie within a given interval, whereas communication exchanges (but not necessarily the assignment of weights) between neighboring nodes are assumed to be bidirectional. 

In the fourth part of this thesis we investigate the problem of integer weight balancing in a multi-component system over a directed (static) interconnection topology, under lower and upper limit constraints on the edge weights, in the presence of bounded or unbounded delays (packet drops) in the communication links.
Specifically, we present a novel distributed algorithm which solves the integer weight balancing problem under lower and upper weight constraints over the communication links for the case where  arbitrary (time-varying and inhomogeneous) time delays affect the transmission at a particular link at a particular time. 
Furthermore, we present an event-based version of the proposed protocol in which each node autonomously decides when communication and control updates should occur so that the resulting network executions still result in a weight balanced digraph and all nodes eventually stop performing transmissions. 
Then, we extend the applicability of the proposed algorithm to the case where possible packet drops affect the communication links. 
}

\clearpage  

\acknowledgements{
\addtocontents{toc}{\vspace{1em}}  

\vspace{1cm}

Firstly, I would like to express my sincere gratitude to my advisor Professor Christoforos N. Hadjicostis for the continuous support of my Ph.D. study and related research, for his patience, motivation, and immense knowledge. His guidance helped me in while researching and writing this thesis. I could not have imagined having a better advisor and mentor for my Ph.D. study.

Besides my advisor, I would like to thank my thesis committee, Professor Georgios Ellinas, Professor Ioannis Krikidis, Professor Themistoklis Charalambous, and Professor Gabriele Oliva, for their insightful comments and encouragement, but also for the suggestions and questions, which prompted me to widen my research from various perspectives.

My sincere thanks also goes to Professor Andrea Gasparri for the excellent collaboration we had. 

I would also like to thank my fellow Ph.D. students, Christoforos Keroglou and Nicolas Manitara, for the stimulating discussions, the sleepless nights we were working together, and all the fun we have had during the last years. 

Special thanks goes to my friends in Cyprus, Greece and others scattered around the world. 
Thank you for your thoughts, well-wishes, phone calls, e-mails, texts, visits, advice, and being there whenever I needed a friend.

Last but not the least, I would like to thank my family: my parents, Ioannis Rikos and Tamara Papakonstantinou, and my brother, Filippos Rikos. Words cannot express how grateful I am for all of the sacrifices that you’ve made on my behalf, for supporting me spiritually throughout my Ph.D. studies and, in general, in my life. 

}
\clearpage  

\setstretch{1.3}  


\pagestyle{fancy}  

\lhead{\emph{Contents}}  
\tableofcontents  

\lhead{\emph{List of Figures}}  
\listoffigures  

\mainmatter	  
\pagestyle{fancy}  


\lhead{\emph{Introduction}}  

\chapter{Introduction}
\label{introduction}

The successful operation of a distributed system or network depends on a number of basic protocols to circulate and process data between its components. 
In distributed systems whose functionality does not simply consist of transmitting data, but also involves control and decision tasks (e.g., workload balancing across available computing resources or leader election), traditional routing protocols may be inadequate or insufficient. 
For this reason, the design of algorithms and protocols for distributed computation has attracted significant attention by the communication, control and computer science communities over the past few decades (e.g., \cite{1996:Lynch, 2007:olfati-saber_consensus, 2003:Koetter, 2004:Rabbat, 2005:Giridhar, 2005:Hromkovic}, and references therein).

A distributed system or network consists of a set of components (nodes) that can share information with neighboring components via connection links (edges), forming a generally directed interconnection topology (digraph). 
The digraphs that describe the communication and/or physical topology typically prove to be of vital importance for the effectiveness of distributed strategies in performing various tasks \cite{2003:jadbabaie_coordination, 2007:olfati-saber_consensus, 2008:RenBeard}. 
In many applications, the assignment of weights to the edges of this graph in a way that forms a balanced digraph (i.e., for each node, the sum of the weights on its incoming edges equals the sum of the weights on its outgoing edges) is key to enabling the desired functionality. 
For example, applications where balance plays a key role include network adaptation strategies based on the use of continuous second order models \cite{2010:DeLellis}, and distributed adaptive strategies to tune the coupling weights of a network based on local information of node dynamics \cite{2012:YuDeLellis}. 
Weight balancing is also closely related to weights that form a doubly stochastic digraph, which find applications in multi-component systems (such as sensor networks) where one is interested in distributively averaging measurements at each component. 
Doubly stochastic digraphs play a key role in networked control problems, including distributed averaging \cite{BulloCortes, 2004:Murray, 2008:RenBeard, 2004:XiaoBoyd} and distributed convex optimization \cite{2009:Johansson, Nedic, ZhuMart}. 
In particular, weight-balance is important in the well-studied case where a set of components (nodes) want to distributively average their individual measurements (in this scenario, each node provides a local measurement of global quantity). 
One approach towards average consensus is to follow a linear iteration, where (instead of routing the value of each node to all other nodes) each node repeatedly updates its value to be a weighted linear combination of its own value and the values of its neighbouring nodes. 
The choice of the weights has a relevant effect on how the interconnection behaves and whether average consensus is reached. 
For example, it has been shown that nonnegative weights that form a primitive doubly stochastic matrix (and thus also balance the graph) is a sufficient condition for asymptotic average consensus \cite{1970:HooiTong, 2010:Cortes}.

Because of the numerous algorithms available in the literature that use of weight assignments that are balanced or even form a doubly stochastic matrix (possibly with self weights for each node), an important research question is to characterize when a digraph can be given such edge weight assignments. 
In this thesis we focus on nonzero weight assignments. 
In addition to its theoretical interest, the consideration of nonzero weight assignments is also relevant from a practical perspective, as the use of the maximum number of edges generally leads to higher algebraic connectivity \cite{Wu}, which in turn affects positively the rate of convergence \cite{2005:Boyd, LiuMorse, CaoWu, GhoshBoyd} of the algorithms that are to be executed over doubly stochastic digraphs.

Once a characterization of weight-balanceable (doubly stochasticable) digraphs is available, the next natural question is the design of distributed strategies that allow the components to find the appropriate weight assignments so that the overall interaction digraph is weight balanced or doubly stochastic.
Therefore, we present novel distributed control algorithms to address these challenges.

\section{Literature Review} 
\label{IntroLiterature}

Previous work on weight balancing consists of the following:

\begin{itemize}

\item[\tiny{$\blacksquare$}] In \cite{2010:Cortes, 2009:Cortes} the authors introduce a synchronized distributed strategy on a directed communication network in which each agent provably balances its incoming and outgoing edge weights in finite time. 
In this algorithm, each individual agent can send a message to one of its out-neighbours and receive a message from its in-neighbours. 
Furthermore, as explained in \cite{2010:Cortes, 2009:Cortes}, once weight-balance is achieved, the nodes can easily obtain weights that form a doubly stochastic matrix in a distributed manner (in order to use them, for instance, to asymptotically reach average consensus). 
[The proposed algorithm achieves this weight assignment under the assumption that individual agents can add self-weights to the structure of the digraph.]

\item[\tiny{$\blacksquare$}] In \cite{2013:Priolo} the authors introduce a distributed algorithm in which each agent is assumed to be able to distinguish the information coming from the other agents according to the identifier of the sender. Also, a global stopping time is set at which the iterations stop to perform  weight-balancing.

\item[\tiny{$\blacksquare$}] In \cite{2013:ThemisHadj} the authors introduce a synchronized distributed strategy on a directed communication network in which each agent provably balances its incoming and outgoing edge weights in an asymptotic fashion. In this algorithm each individual agent can send a message to all of its out-neighbours and receive messages from all its in-neighbours. The authors show that the proposed distributed algorithm guarantees geometric convergence rate. The simplicity of the algorithm has allowed its extension to asynchronous operation, which is a valuable contribution given that in reality there are inevitable delays in the exchange of information) as well as its continuous-time analog, that guarantees average consensus without the need to obtain a doubly stochastic matrix.

\item[\tiny{$\blacksquare$}] In \cite{Zenios} the authors introduce a synchronized distributed strategy on a directed communication network in which each agent provably balances its incoming and outgoing edge weights in an asymptotic fashion. In this algorithm, each individual agent is assumed to be able to send a messages to all of its out-neighbours and receive messages from all its in-neighbours. Specifically, each agent calculates a fraction $\lambda$ which has as numerator the sum of incoming weights and as denominator the sum of outgoing weights, and then changes the weights of its incoming and outgoing links accordingly.
\end{itemize}

\section{Motivation and Applications} 
\label{IntroMotivation}

The study of weight-balanced graphs/matrices has proven to play an important role in the analysis and convergence of distributed coordination algorithms since they find numerous applications in distributed adaptive control or synchronization in complex networks. 
The main applications of weight-balanced graphs/matrices are shown below:

\textbf{Balancing of Physical Quantities:}
In \cite{1970:HooiTong} a traffic-flow problem is studied consisting of $n$ junctions and $m$ one-way one way streets. 
Such an application shows that the goal of ensuring a smooth traffic flow is associated with balanced weights on an appropriately defined digraph. 
Weight balanced digraphs appear also in the design of stable flocking algorithms for agents with significant inertial effects, where weight-balance allows for the decoupling of the centroid dynamics from the internal group formation \cite{1984:LeeSpong}. 
Additionally, examples of applications where balance plays a key role include network adaptation strategies based on the use of continuous second order models \cite{2010:DeLellis}, and distributed adaptive strategies to tune the coupling weights of a network based on local information of node dynamics \cite{2012:YuDeLellis}. 
Weight balancing can also be associated with the matrix balancing problem in network optimization which is, in turn, associated with numerous applications, such as predicting the distribution matrix of telephone traffic  \cite{1998:Bertsekas}.
Furthermore, the weight balancing problem is related to matrix scaling problems which have been addressed in the context of nonnegative matrices \cite{Zenios}. 
One of the early motivations for the matrix scaling problem was the desire to start from the stochastic matrix of a Markov chain and obtain a scaled version of it that is doubly stochastic and adheres to the sparsity structure of the original one. 
Many applications of matrix scaling can also be found in economy or accounting models (where it is important to balance the flow-of-funds), urban planning, statistics, and demography.

\textbf{Doubly Stochastic Weights:}
In all of the above applications, weights are associated with the physical interactions in a distributed control system, and are assigned to edges of the physical digraph. 
There are also many applications where weight balancing plays a significant role in the cyber digraph of a given distributed control system. 
In particular, weight balancing is closely related to weights that form a doubly stochastic digraph, which find applications in multi-component systems (such as sensor networks) where one is interested in distributively averaging measurements at each component. 
Doubly stochastic digraphs play a key role in networked control problems, including distributed averaging \cite{2018:BOOK_Hadj, BulloCortes, 2004:Murray, 2008:RenBeard, 2004:XiaoBoyd} and distributed convex optimization \cite{2009:Johansson, Nedic, ZhuMart}. 
Convergence in gossip algorithms also relies on the structure of doubly stochastic digraphs, see \cite{2005:Boyd, LiuMorse}.
In particular, weight-balance is important in the well-studied case where a set of components (nodes) want to distributively average their individual measurements (in this scenario, each node provides a local measurement of global quantity). 
One approach towards average consensus is to follow a linear iteration, where (instead of routing the value of each node to all other nodes) each node repeatedly updates its value to be a weighted linear combination of its own value and the values of its neighbouring nodes. 
Asymptotic average consensus is then guaranteed (i.e., the nodes asymptotically reach consensus to the average of their initial values \cite{2007:olfati-saber_consensus, 2008:RenBeard, 1984:Tsitsiklis, 2008:Cortes}) if the weights used in the linear iteration form a doubly stochastic matrix (which correspond to a balanced digraph) \cite{1984:Tsitsiklis}.
The choice of the weights is important in how the interconnection behaves and whether average consensus is reached. 
For example, it has been shown that nonnegative weights that form a primitive doubly stochastic matrix (and thus also balance the graph) is a sufficient condition for asymptotic average consensus as long as the digraph is strongly connected \cite{1970:HooiTong, 2010:Cortes}.
Average consensus is a special case of the consensus problem which has received significant attention from the computer science community \cite{1996:Lynch} and the control community (see \cite{2007:olfati-saber_consensus, 2004:XiaoBoyd, 2003:jadbabaie_coordination}), due to its applicability to diverse areas, including multi-component systems, cooperative control \cite{2007:OlfatiReza}, modeling of flocking behavior in biological and physical systems (e.g., \cite{2007:olfati-saber_consensus, 2008:RenBeard, 2003:jadbabaie_coordination}) and estimation and tracking \cite{2008:CarliChiusoSchenatoZampieri}.

\textbf{Flow Balancing:}
A weighted digraph that has a real or integer value (called the edge weight) associated with each edge is also similar to a flow network where each edge receives a flow that typically cannot exceed a given capacity (or, more generally, has to lie within a given interval). 
Flows  must satisfy the restriction that the amount of flow into a node equals the amount of flow out of it, unless the node is a source, which has only outgoing flow, or a sink, which has only incoming flow.
Thus, the weight-balancing problem we deal with in this thesis can also be viewed as the problem of producing a feasible circulation in a directed graph with upper and lower flow constraints \cite{1986:Papadimitriou}. [In such settings, a circulation in a directed graph is an assignment of nonnegative weights to the cycles of the graph and is called feasible if the flow in each edge (i.e., the sum of the weights of the cycles containing this edge) lies between the corresponding upper and lower flow constraints.]
Additionally, the problem we deal with can be also viewed as a particular case of the standard network flow problem (see, e.g., \cite{2010:Fulkerson}), where there is a cost associated to the flow on each link, and the objective is to minimize the total cost subject to constraints on the flows.
Moreover, flow algorithms find further applications in a variety of other problems, like the maximum flow problem \cite{1988:GoldTarj}, auction algorithms \cite{1992:Bertsekas}, and energy minimization \cite{2004:Kolmogorov}.

\textbf{Balancing with Integer Weights:}
Digraphs that are balanced with integer weights find numerous applications in a variety of problems like swarm guidance \cite{2014:Bandyopadhyay}, fractional packing \cite{1995:Plotkin, 2007:Gang}, matching in bipartite graphs \cite{1973:Hopcroft}, and edge-disjoint paths \cite{2003:Guruswami}.

\section{Main Contributions} 
\label{IntroContributions}

As stated previously, this thesis focuses on the development of distributed novel algorithms that facilitate the development of strategies for generating weight balanced digraphs. 
The main contributions of this thesis are as follows:

\begin{itemize}


\item[\tiny{$\blacksquare$}] In Section~\ref{distralg} we introduce a novel distributed algorithm which achieves integer weight balancing in a multi-component system. We present its formal description along with an illustrative example. Then, we show that the proposed distributed algorithm converges to a weight balanced digraph after a finite number of iterations, for which we obtain explicit bounds. Finally, we present examples and simulations for the distributed algorithm.

\item[\tiny{$\blacksquare$}] In Section~\ref{algDel} we introduce a novel distributed algorithm which achieves integer weight balancing in a multi-component system, in the presence of time delays over the communication links. We present its formal description, along with an illustrative example, and we show that the proposed distributed algorithm converges to a weight balanced digraph after a finite number of iterations in the presence of bounded time delays over the communication links. 
Then, in Section~\ref{triggeredalgDel}, we discuss an event-triggered version of the proposed distributed algorithm and we show that it results in a weight balanced digraph after a finite number of iterations in the presence of arbitrary (time-varying, inhomogeneous) but bounded time delays over the communication links.

\item[\tiny{$\blacksquare$}] In Section~\ref{packetalgDel} we show that the proposed distributed algorithm presented in Section~\ref{algDel} is also able to converge (with probability one) to a weight balanced digraph in the presence of unbounded delays (packet drops).


\item[\tiny{$\blacksquare$}] In Section~\ref{upperloweralgorithm} we introduce a novel distributed algorithm which achieves integer weight balancing in a multi-component system, in the presence of specified lower and upper limit constraints on the edge weights. We present its formal description along with an illustrative example. Then, we show that as long as the conditions hold, then the proposed distributed algorithm converges to a weight balanced digraph after a finite number of iterations. Finally, we present examples and simulations for the proposed distributed algorithm.

\item[\tiny{$\blacksquare$}] In Section~\ref{upperloweralgorithm_delays} we introduce a novel distributed algorithm which achieves integer weight balancing in a multi-component system under specified lower and upper limit constraints on the edge weights, in the presence of time delays over the communication links. 
We show that as long as the conditions hold, then the proposed distributed algorithm converges to a weight balanced digraph after a finite number of iterations.
In Section~\ref{triggeredalgDel_delays}, we discuss an event-triggered extension regarding the operation of the aforementioned distributed algorithm and we show that it results in a weight balanced digraph after a finite number of iterations.

\item[\tiny{$\blacksquare$}] In Section~\ref{packetalgDel_delays} we introduce a novel distributed algorithm which achieves integer weight balancing in a multi-component system under specified lower and upper limit constraints on the edge weights, in the presence of unbounded delays (packet drops) over the communication links. 
We show that as long as the conditions hold, then the proposed distributed algorithm converges to a weight balanced digraph after a finite number of iterations.

\end{itemize}

\section{Thesis Organization}  
\label{IntroOrganization}

This thesis is organized as follows. 
In Chapter~\ref{notationproblem}, we present some basic notions and notation needed for our development. 
Then we present the problem formulation and we discuss a possible solution in a centralized fashion. 
In Chapter~\ref{centrvsdistr}, we introduce the distributed algorithm which achieves balance with integer weights after a finite number of iterations. 
In Chapter~\ref{distrdelpacket}, we present the distributed algorithm which achieves integer weight balancing in the presence of bounded delays after a finite number of iterations. We also analyze the case of unbounded delays (packet drops) in the communication links and discuss an event-triggered version of the algorithm (that can be used to avoid unnecessary transmissions).
In Chapter~\ref{constraintsbalancing}, we present the conditions for the existence of a set of integer weights (within the allowable intervals) that balance a weighted digraph. Then, we present the distributed algorithm which achieves integer weight balancing in a multi-component system, in the presence of specified lower and upper limit constraints on the edge weights. 
In Chapter~\ref{constraintsbalancing_delays}, we present the distributed algorithm, which achieves integer weight balancing under lower and upper limit constraints on the edge weights in the presence of bounded delays after a finite number of iterations. We also analyze the case of unbounded delays (packet drops) in the communication links and discuss an event-triggered version of the algorithm (that can be used to avoid unnecessary transmissions).
Finally, in Chapter~\ref{conclusions} we conclude this thesis with a brief summary and remarks about future work.

\clearpage

\lhead{\emph{Preliminaries, Problem Statement and Centralized Approach}}

\chapter{Preliminaries, Problem Statement and Centralized Approach}
\label{notationproblem}

In this chapter, we first introduce some key notions and notation in Section~\ref{GraphNotation}. 
Then, we state and discuss the problem under consideration in Section \ref{ProbStatement} and present a centralized approach of the problem in Section~\ref{centralg}.

\section{Graph-Theoretic Notions}
\label{GraphNotation}

The sets of real, integer and natural numbers are denoted by $\mathbb{R}$, $\mathbb{Z}$ and $\mathbb{N}$, respectively. The symbol $\mathbb{N}_0$ denotes the set of nonnegative integers while the positive part of $\mathbb{Z}$ is denoted by the subscript $+$ (e.g. $\mathbb{Z}_+$). 
Vectors are denoted by small letters whereas matrices are denoted by capital letters. 
A matrix with nonnegative elements is called nonnegative matrix and is denoted by $A\geq0$ while a matrix with positive elements is called positive matrix and is denoted by $A>0$. 

A distributed system whose components can exchange masses of certain quantities of interest (weights or flows) via (possibly directed) links, can conveniently be captured by a digraph (directed graph). A digraph of order $n$ ($n \geq 2$), is defined as $\mathcal{G}_d = (\mathcal{V}, \mathcal{E})$, where $\mathcal{V} =  \{v_1, v_2, \dots, v_n\}$ is the set of nodes and $\mathcal{E} \subseteq \mathcal{V} \times \mathcal{V} - \{ (v_j,v_j)$ $|$ $v_j \in \mathcal{V} \}$ is the set of edges. A directed edge from node $v_i$ to node $v_j$ is denoted by $(v_j, v_i) \in \mathcal{E}$, and indicates that $v_j$ can receive information or physical quantities from $v_i$. 
We will refer to the digraph $\mathcal{G}_d$ as the {\em physical topology}.


A digraph is called \textit{strongly connected} if for each pair of vertices $v_j, v_i \in \mathcal{V}$, $v_j \neq v_i$, there exists a directed \textit{path} from $v_i$ to $v_j$, i.e., we can find a sequence of vertices $v_i \equiv v_{l_0},v_{l_1}, \dots, v_{l_t} \equiv v_j$ such that $(v_{l_{\tau+1}},v_{l_{\tau}}) \in \mathcal{E}$ for $ \tau = 0, 1, \dots , t-1$. All nodes that have edges to node $v_j$ are said to be in-neighbors of node $v_j$ and belong to the set $\mathcal{N}_j^- = \{v_i \in V \; | \; (v_j, v_i) \in \mathcal{E} \}$. The cardinality of $\mathcal{N}_j^-$ is called the \textit{in-degree} of node $v_j$ and is denoted by $\mathcal{D}_j^-$. The nodes that have edges from node $v_j$ comprise its out-neighbors and are denoted by $\mathcal{N}_j^+ = \{v_l \in \mathcal{V} \; | \; (v_l, v_j) \in \mathcal{E} \}$. The cardinality of $\mathcal{N}_j^+$ is called the \textit{out-degree} of $v_j$ and is denoted by $\mathcal{D}_j^+$. We also let $\mathcal{N}_j = \mathcal{N}^+_j \cup  \mathcal{N}^-_j$ denote the {\em neighbors} of node $v_j$, and $\mathcal{D}_j = \mathcal{D}^+_j + \mathcal{D}^-_j$ denote the {\em total degree} of node $v_j$. 
A weighted digraph $\mathcal{G}_d = (\mathcal{V}, \mathcal{E}, \mathcal{F})$ is a digraph in which each edge $(v_j, v_i) \in \mathcal{E}$ is associated with a real or integer value $f_{ji}$ called the edge weight; matrix $\mathcal{F} = [f_{ji}]$ with value $f_{ji}$ at its $j$th row, $i$th column position (where $f_{ji}=0$ if $(v_j,v_i)$ does not belong in $\mathcal{E}$).



\begin{defn}\label{DEFnodeinoutweight}
Given a weighted digraph $\mathcal{G}_d=(\mathcal{V},\mathcal{E},\mathcal{F})$ of order $n$, the total \textit{in-weight} of node $v_j$ is denoted by $\mathcal{S}_j^-$ and is defined as $\mathcal{S}_j^- = \sum_{v_i \in \mathcal{N}_j^-} f_{ji}$, whereas the total \textit{out-weight} of node $v_j$ is denoted by $\mathcal{S}_j^+$ and is defined as $\mathcal{S}_j^+ = \sum_{v_l \in \mathcal{N}_j^+} f_{lj}$.
\end{defn}
\begin{defn}\label{DEFnodebalance}
Given a weighted digraph $\mathcal{G}_d=(\mathcal{V},\mathcal{E},\mathcal{F})$ of order $n$, the weight \textit{imbalance} of node $v_j$ is denoted by $x_j$ and is defined as $x_j = \mathcal{S}_j^- - \mathcal{S}_j^+$.
\end{defn}
\begin{defn}\label{defn:totalim}
Given a weighted digraph $\mathcal{G}_d=(\mathcal{V},\mathcal{E},\mathcal{F})$ of order $n$, the \textit{total imbalance} (or {\em absolute imbalance}) of digraph $\mathcal{G}_d$ is denoted by $\varepsilon$ and is defined as $\varepsilon = \sum_{j=1}^{n} \vert x_j \vert$.
\end{defn}
\begin{defn}
A weighted digraph $\mathcal{G}_d=(\mathcal{V},\mathcal{E},\mathcal{F})$ is called weight balanced if its \textit{total imbalance} (or {\em absolute imbalance}) is equal to $0$, i.e., $\varepsilon \buildrel\triangle\over = \sum_{j=1}^{n} \vert x_j \vert = 0$. 
\end{defn}

\section{Problem Statement for Weight Balancing}
\label{ProbStatement}

We are given a strongly connected digraph $\mathcal{G}_d = (\mathcal{V}, \mathcal{E})$, with a set of nodes $\mathcal{V} =  \{v_1, v_2, \dots, v_n\}$ and a set of edges $\mathcal{E} \subseteq \mathcal{V} \times \mathcal{V} - \{ (v_j,v_j)$ $|$ $v_j \in \mathcal{V} \}$. 
We want to develop a distributed algorithm that allows the nodes to iteratively adjust the weights on their edges so that they eventually obtain a set of integer weights $\{ f_{ji} \; | \; (v_j, v_i) \in \mathcal{E} \}$ that satisfy the following:
\begin{enumerate}
\item $f_{ji} \in \mathbb{N}$ for every edge $(v_j,v_i) \in \mathcal{E}$;
\item $f_{ji} = 0$ if $(v_j,v_i) \notin \mathcal{E}$;
\item $\mathcal{S}_j^+ = \mathcal{S}_j^-$ for every $v_j \in \mathcal{V}$.
\end{enumerate}

The algorithms we develop in this thesis are iterative, and we use $k$ to denote the iteration. 
For example, $\mathcal{S}_j^+[k]$ will denote the value of the total out-weight of node $v_j$ at time instant $k$, where $k \in \mathbb{N}_0$.

In Figure~\ref{probstatebalanced} we can see an example of a digraph which satisfies the conditions presented above. As we can see, for every node $v_1, v_2, v_3, v_4$ and $v_5$ the total \textit{in-weight} (equal to $4, 4, 6, 4$ and $4$ respectively) is equal to the total \textit{out-weight}. 
Thus, the digraph in Figure~\ref{probstatebalanced} is weight balanced.

\begin{figure} [ht]
\centering
\includegraphics[width=0.40\textwidth]{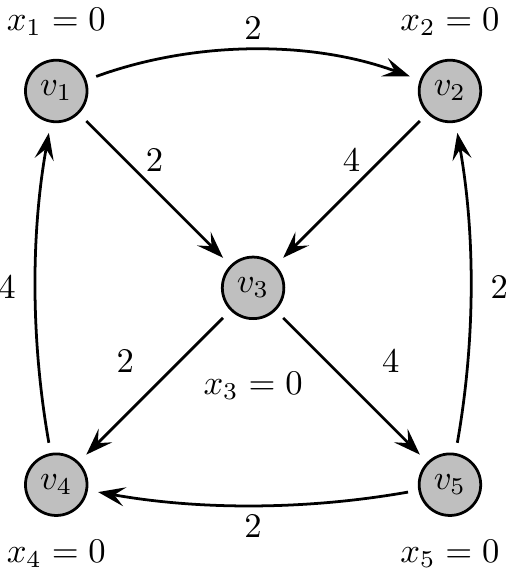}
\caption{Example of weight balanced digraph.}
\label{probstatebalanced}
\end{figure}

\section{Centralized Algorithm for Weight Balancing}
\label{centralg}

We now introduce an algorithm which solves the integer weight balancing problem over a multi-component system in a centralized fashion. 

The centralized algorithm takes as input a strongly connected digraph $\mathcal{G}_d = (\mathcal{V}, \mathcal{E})$. 
It initializes the weights of all edges to unity and then iteratively performs the following steps until the graph is balanced:

\begin{enumerate}
\item Computes the weight imbalance of each node.
\item Pick one node with positive imbalance and one with negative imbalance; if there is more than one pair of nodes satisfying such a condition, any pair can be chosen (something like that).
\item Find a path in the digraph from the node with positive imbalance to the node with negative imbalance (this is always possible as long as the graph is strongly connected).
\item Increase the weights of all the edges in the path by the value of the weight imbalance of the positively imbalanced node.
\end{enumerate}

We discuss why the algorithm results in a weight balanced graph (and how many steps it takes to do so), after we describe the algorithm more formally.

\subsection{Formal Description of Centralized Algorithm}
\label{formalcentralg}

A formal description of the proposed centralized algorithm is presented in Algorithm~\ref{algorithm:1}.

\begin{varalgorithm}{1}
\caption{Centralized balancing with integer weights}
\textbf{Input} \\ A strongly connected digraph $\mathcal{G}_d=(\mathcal{V},\mathcal{E})$ with $n=|\mathcal{V}|$ nodes and $m=|\mathcal{E}|$ edges.\\
\textbf{Initialization} \\ Set $k=0$; each node $v_j \in \mathcal{V}$ sets its out-going edge weights as
\vspace{-0.1cm}
\begin{align*}
f_{lj}[0] = \left\{ \begin{array}{ll}
         0, & \mbox{if $v_l \notin \mathcal{N}_j^+$,}\\
         1, & \mbox{if $v_l \in \mathcal{N}_j^+$.}\end{array} \right. 
\end{align*}

\textbf{Iteration} \\ For $k=0,1,2,\dots$, each node $v_j \in \mathcal{V}$ does the following:
\vspace{-0.1cm}
\begin{enumerate}
\item It computes its weight imbalance $x_j[k] = \mathcal{S}_j^-[k] - \mathcal{S}_j^+[k]$.
\item It selects one node $v^+$ with positive imbalance $br^+$ and one node $v^-$ with negative imbalance $br^-$ (e.g., it selects the node with the largest positive imbalance and the node with the smallest negative imbalance, respectively).
\item It finds a non-cyclic path $v^+=v_{j_0},v_{j_1},\cdots,v_{j_t}=v^-$ from $v^+$ to $v^-$.
\item It increases the weight on each edge on the path, which connects $v^+$ to $v^-$, by $br^+$, i.e., 
\[ f_{j_{e+1},j_e}[k+1]=f_{j_{e+1},j_e}[k]+br^+ \]
for $e=0,1,\cdots,t-1$. (All the other weights are left unchanged.)
\item It repeats (increases $k$ to $k+1$ and goes back to Step~1).
\end{enumerate}
\label{algorithm:1}
\end{varalgorithm}

\subsection{Illustrative Example of Centralized Algorithm}
\label{examplecentralg}

We first illustrate the centralized algorithm. We then explain why it results in a weight balanced digraph after a finite number of iterations (bounded by $n - 1 =|\mathcal{V}| - 1$ in the worst-case).

Consider the digraph $\mathcal{G}_d = (\mathcal{V}, \mathcal{E})$ in Figure~\ref{initial-centralized}, where $\mathcal{V} = \{ v_1,v_2,\cdots,v_7 \}$, $E = \{ e_1,e_2,\cdots,e_{13} \}$, $\mathcal{E} \subseteq \mathcal{V} \times \mathcal{V} - \{ (v_j,v_j)$ $|$ $v_j \in \mathcal{V} \}$. The weight on each edge is initialized to $f_{ji}[0]=1$ for $(v_j,v_i) \in E$ (otherwise $f_{ji}[0]=0$). As a first step, we compute the weight imbalance $x_j[0]=S_j^-[0] - S_j^+[0]$ for each node (this is shown in Figure \ref{initial-centralized}).

\begin{figure} [ht]
\centering
\includegraphics[width=0.60\textwidth]{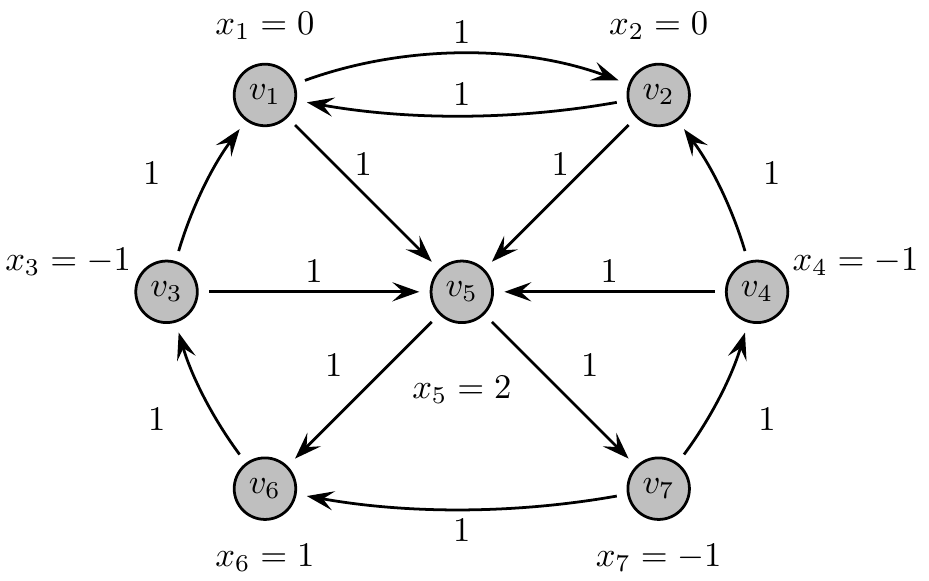}
\caption{Weighted digraph with initial weights and initial imbalance for each node.}
\label{initial-centralized}
\end{figure}

Once we compute the imbalance of each node, the centralized algorithm selects (randomly or otherwise) one node with positive imbalance, say $v_5$, and one node with negative imbalance, say $v_7$. A path from node $v_5$ to node $v_7$ is selected (e.g., the path $v_5,v_7$) and the weights of all the edges in the path are increased by the value of the weight imbalance of the positively imbalanced node $v_5$ (namely, by the value of the weight imbalance $x_5=2$), as shown in Figure \ref{iter-centralized}.

\begin{figure} [ht]
\centering
\includegraphics[width=0.60\textwidth]{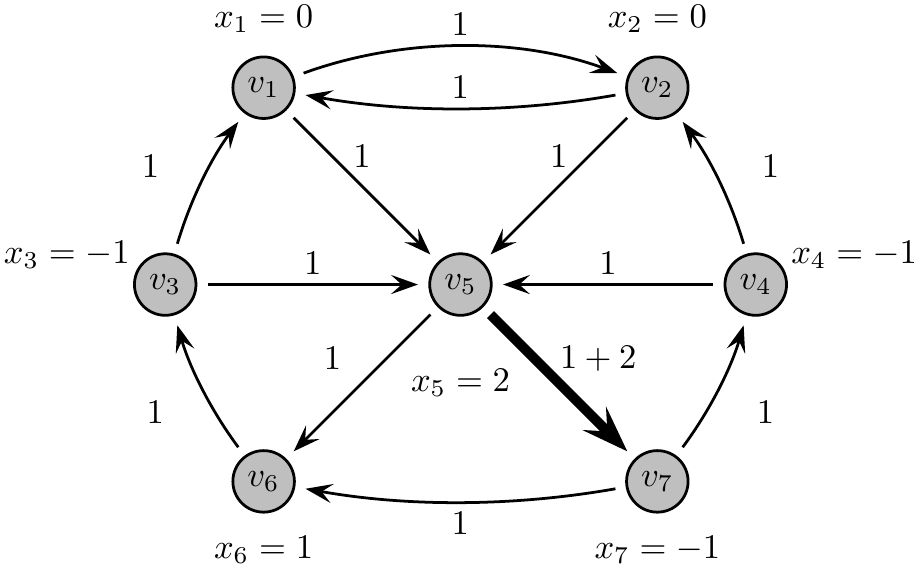}
\caption{Selection of a path between a node with positive imbalance and a node with negative imbalance.}
\label{iter-centralized}
\end{figure}

At the next iteration, after the increase of the weights of all the edges of the path $v_5,v_7$, we recalculate the imbalance for each node $v_j$ as $x_j[1]=S_j^-[1] - S_j^+[1]$. This can be seen in Figure \ref{calculate-centralized}.

\begin{figure} [ht]
\centering
\includegraphics[width=0.60\textwidth]{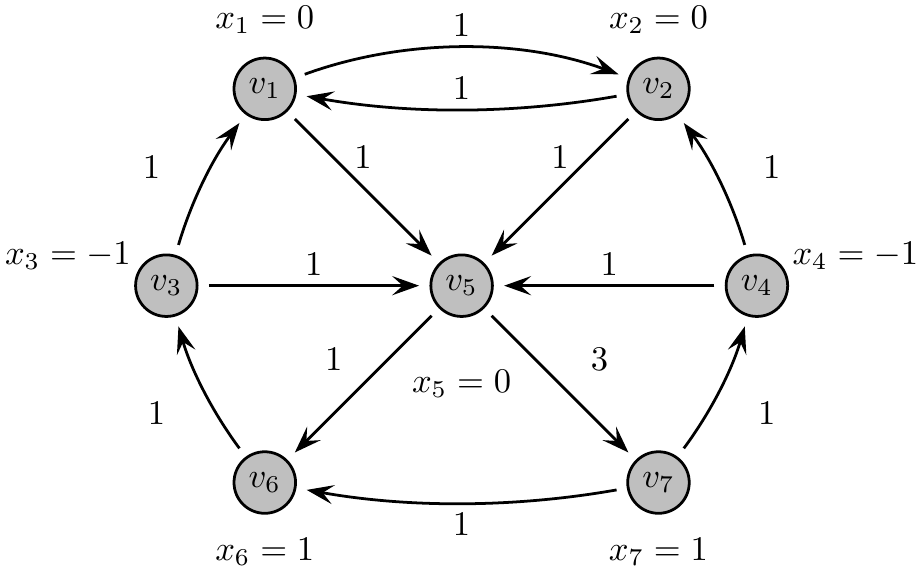}
\caption{Calculation of imbalance for each node.}
\label{calculate-centralized}
\end{figure}

At the next iteration, one node with positive imbalance, say $v_7$, and one node with negative imbalance, say $v_4$ are selected. A path from node $v_7$ to node $v_4$ is created (e.g., the path $v_7,v_4$) and the weights of all the edges in the path are increased by the value of the weight imbalance of the positively imbalanced node $v_7$ (namely, by the value of the weight imbalance $x_4=1$), as shown in Figure \ref{sec_iter_centr}.

\begin{figure} [ht]
\centering
\includegraphics[width=0.60\textwidth]{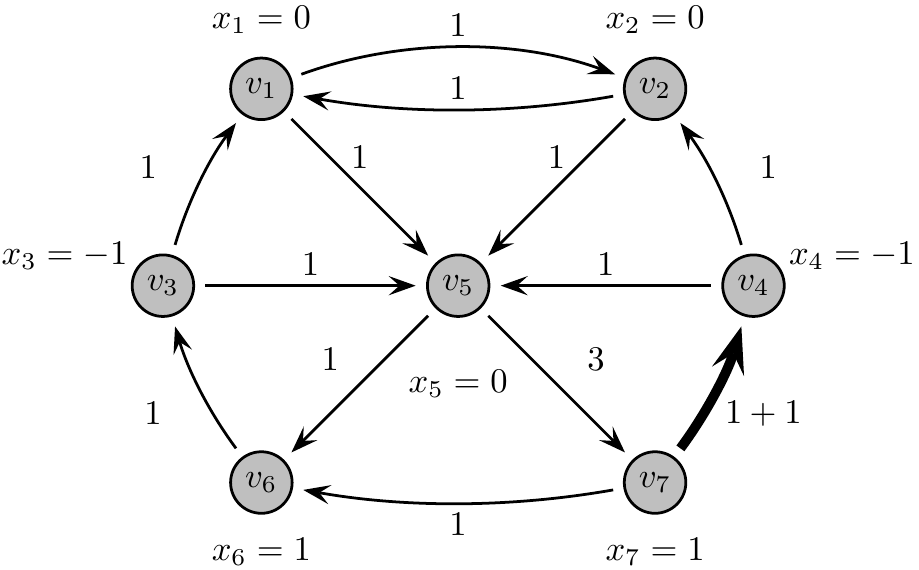}
\caption{Selection of a path between a node with positive imbalance and a node with negative imbalance.}
\label{sec_iter_centr}
\end{figure}

The process is repeated until the graph becomes weight balanced. In this particular example, this occurs after four iterations and the final weights are shown in Figure~\ref{balanc-centralized}.

\begin{figure} [ht]
\centering
\includegraphics[width=0.60\textwidth]{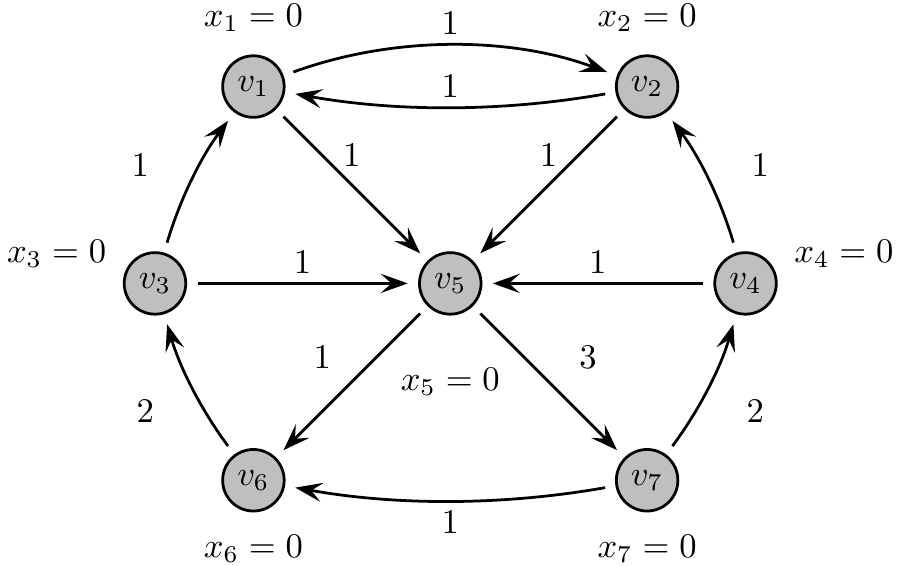}
\caption{Resulting weight balanced digraph.}
\label{balanc-centralized}
\end{figure}

We will next explain why the proposed centralized algorithm results in a weight balanced digraph after a finite number of iterations (bounded by $n - 1$ in the worst-case).

\subsection{Bound on Number of Iterations of Centralized Algorithm}
\label{analysiscentralg}

Notice that each edge appears as the incoming edge of exactly one node and as the outgoing edge of exactly one (other) node. Thus, we immediately have that
$$
\sum_{j=1}^{n}  x_j[k] = 0 \; , \text{ for all } k \; .
$$
This means that if there is a node with positive imbalance at Step~2 of the iteration, then there has to be at least one node with negative imbalance (and vice-versa). Also, notice that the adjustment of weights in Step~3 of the iteration, achieves the following:
\begin{enumerate}
\item It balances the node $v^+$ that had positive imbalance.
\item It does not change the imbalance of all other intermediate nodes $v_{j_1},\cdots,v_{j_{t-1}}$ in the path $v^+=v_{j_0},v_{j_1},\cdots,v_{j_t}=v^-$ from $v^+$ to $v^-$.
\item It increases the imbalance of the node $v^-$ that had negative imbalance.
\end{enumerate}
Note that if a node starts balanced or becomes balanced during any iteration, it remains balanced for the remainder of the algorithm (because Steps~1,~2, and~3 do not affect the imbalance of intermediate nodes in the path that is selected, and only these nodes can be balanced). Furthermore, at each iteration, at least one node becomes balanced (namely, the node with positive imbalance that is picked at Step~1). Note that it is possible for two nodes to become balanced at each iteration (if the node with negative imbalance that is picked happens to also become balanced; in fact, this is the case at the last iteration). Thus, it is easy to see that the algorithm takes at most $n-1$ iterations to reach a set of weights that forms a weight balanced graph. 

Another easily obtainable bound on the number of iterations is the following: if we think of the {\em absolute balance} of the graph at iteration $k$ as $\varepsilon[k] = \sum_{j=1}^n |x_j[k]|$, then each iteration decreases this imbalance by at least $2$ (i.e., $\sum_{j=1}^n |x_j[k+1]| \leq \sum_{j=1}^n |x_j[k]| - 2$) unless the graph is balanced. [Note that $\sum_{j=1}^n |x_j[0]|$ is necessarily an even number (because the sum of positive balances is equal to the negative of the sum of the negative imbalances, and both of them are integer numbers). Also, each iteration decreases the sum of the positive imbalances by at least 1; thus, the absolute sum of the negative imbalances also has to decrease by at least 1 as well.]

The above discussion implies the proposition below.

\begin{prop}
The number of iterations $T$ required by the proposed centralized algorithm to balance a digraph $\mathcal{G}_d = (\mathcal{V}, \mathcal{E})$ satisfies $T \leq \min(n-1, \frac{1}{2}\sum_{j=1}^n |x_j[0]|)$.
\end{prop}

\subsection{Simulation Study}
\label{resultscentralg}

In this section, we present simulation results for random graphs of size $n = 20$ and $50$ nodes.

Figure~\ref{centralized20} shows the case of a random digraph of $n = 20$ nodes. Here we can see that the proposed centralized algorithm converges to a weight balanced digraph after a finite number of iterations. 

\begin{figure}[ht] 
\centering    
\includegraphics[width=0.55\textwidth]{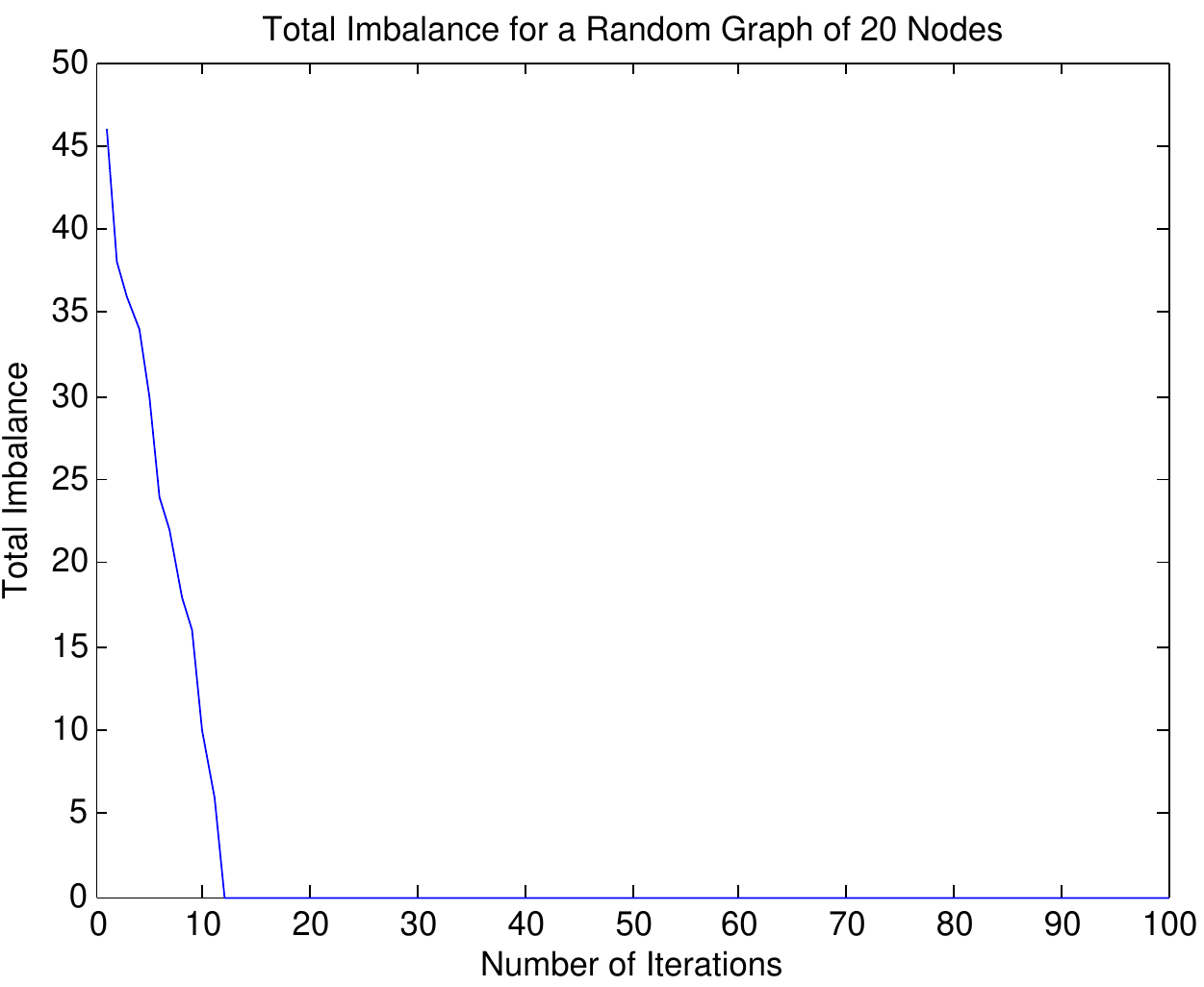}
\caption{Total imbalance plotted against the number of iterations for a random digraph of $20$ nodes when the proposed centralized algorithm is executed.}
\label{centralized20}
\end{figure}

Figure~\ref{centralized50} shows the same case as Figure~\ref{centralized20} with the difference that the network consists of $50$ nodes. 
The increase in network size does not cause any major changes in performance and the conclusions are the same as in Figure~\ref{centralized20}.

\begin{figure}[ht] 
\centering    
\includegraphics[width=0.55\textwidth]{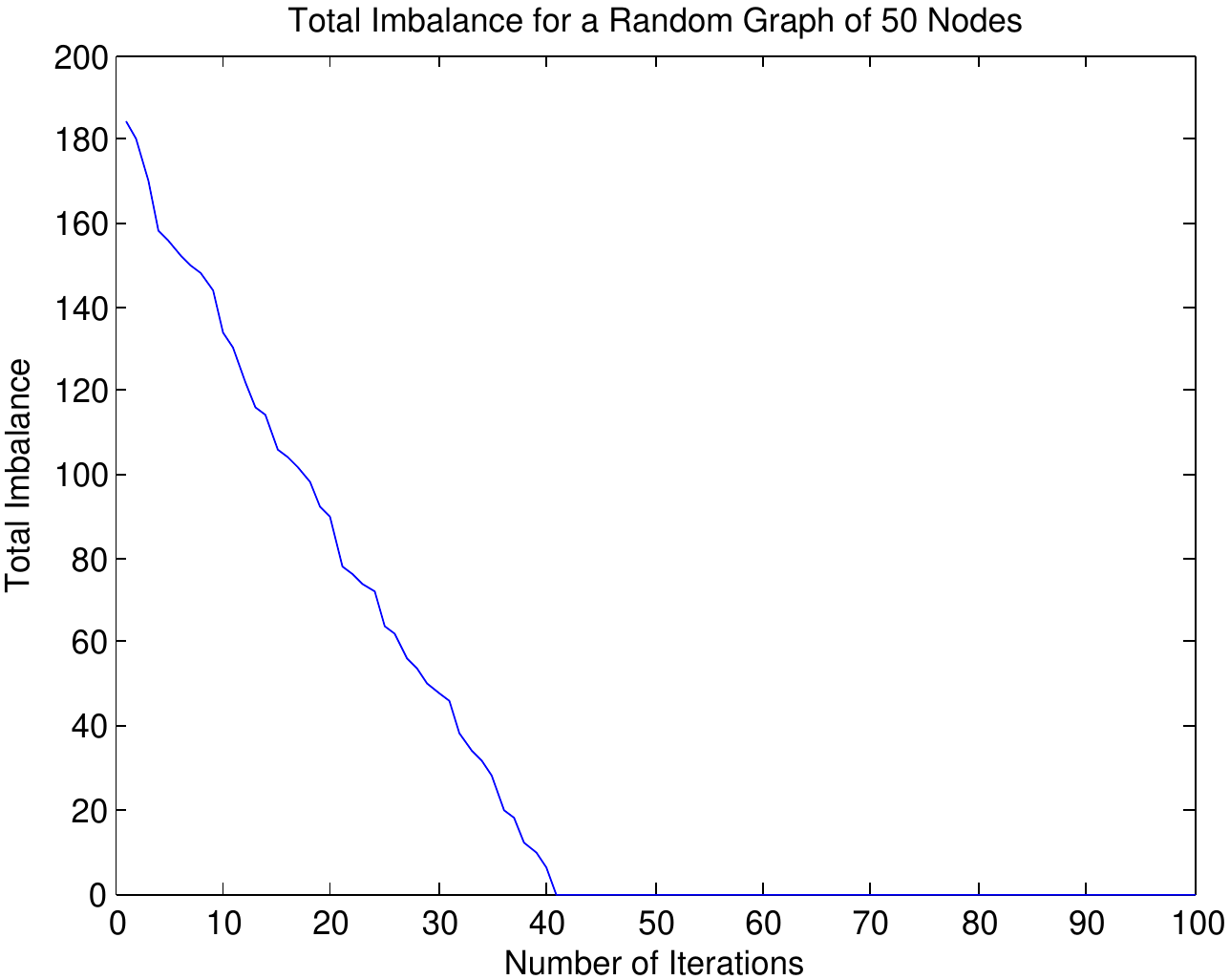}
\caption{Total imbalance plotted against the number of iterations for a random digraph of $50$ nodes when the proposed centralized algorithm is executed.}
\label{centralized50}
\end{figure}

\clearpage

\lhead{\emph{Distributed Weight Balancing}}

\chapter{Distributed Weight Balancing}
\label{centrvsdistr}

In this chapter, we present a novel distributed algorithm which deals with the problem of balancing a weighted digraph. 

This chapter is organized as follows. In Section~\ref{distralg} we introduce a novel distributed algorithm which achieves integer weight balancing in a multi-component system. In Sections~\ref{formaldistralg} - \ref{resultsdistralg} we present a formal description of the proposed distributed algorithm and demonstrate its performance via an illustrative example. 
Then, in Section~\ref{analysisdistralg}, we show that the proposed distributed algorithm converges to a weight balanced digraph after a finite number of iterations and we calculate an explicit bound on the number of iterations required. 
In Section~\ref{resultsdistralg} we present simulation results and comparisons for the proposed distributed algorithm.
The chapter is concluded in Section~\ref{sumarycentrdistralg}.

\section{Distributed Algorithm for Weight Balancing}
\label{distralg}

In this section we present a distributed algorithm (Algorithm~\ref{algorithm:2}) in which the nodes iteratively adjust the positive integer weights of their outgoing edges, such that the digraph becomes weight balanced after a finite number of iterations. 
We assume that each node observes but cannot set the weights of its incoming edges, and based on these weights it adjusts the weights on its outgoing edges. 
It is required that the weights on the outgoing edges of each node can be adjusted {\em differently} if necessary. 
(Note that this requirement is not present in \cite{2012:Rikos} where each node sets equal weights to all of its outgoing edges; when restricting ourselves to integer weights, however, this requirement becomes necessary for balancing to be possible, see, for example, \cite{2012:CortesJournal}).

Given a strongly connected digraph $\mathcal{G}_d = (\mathcal{V},\mathcal{E})$, Algorithm~\ref{algorithm:2} has each node $v_j \in \mathcal{V}$ initialize the weights of all of its outgoing edges to $n$ (or some constant greater than or equal to unity\footnote{It will become evident from our analysis that this constant does not have to be the same for all nodes (and its value does not affect the termination of the algorithm), but for simplicity we take it to be $n$ in this analysis. Note that if all edge weights were initialized to $1$, the execution of Algorithm~\ref{algorithm:2} is identical to the execution of the algorithm in \cite{2013:RikosHadj} because nodes with negative imbalance never take any action.}). 
Then, it enters an iterative stage where each node $v_j$  performs the following steps (a formal description of Algorithm \ref{algorithm:2} appears later):
\vspace{-0.1cm}
\begin{enumerate}
\item The node computes its weight imbalance.
\item If the node has positive imbalance, it increases the integer weights of its outgoing edges so that it becomes weight balanced (assuming no further changes by its in-neighbors on its incoming edges). Specifically, the outgoing edges are assigned, if possible, equal integer weights; otherwise, if this is not possible, they are assigned integer weights such that the maximum difference among them is equal to unity. 
This means that some of the outgoing edges of each node might get larger weights (by unity) than others, and we assume that each node selects {\em a priori} a fixed (possibly randomly selected) ordering of its out-neighbors that determines the precedence with which outgoing edges get higher weight.\footnote{The exact ordering is not critical and, in fact, other strategies are possible as long as they keep some balance among weights. For example, the algorithm proposed in \cite{2012:CortesJournal}, by having each node with positive imbalance increase the weight of the outgoing edge with minimum weight, also imposes some sort of balance among the weights on its outgoing edges. This algorithm has been shown to complete in finite time, but a explicit bound on the number of steps required has not been obtained.}
\item If the node has negative imbalance, it decreases (if possible) the integer weights of its outgoing edges so that i) they have value greater or equal to unity, and ii) its  weight imbalance becomes equal to $-1$ (assuming no further changes by its in-neighbors on its incoming edges). 
As in Step~$2$ above, the outgoing edges are assigned, if possible, equal integer weights; otherwise, if this is not possible, they are assigned integer weights (greater or equal to unity) such that the maximum difference among them is equal to unity (again, we assume each node determines which of its outgoing edges get higher weight based on some {\em a priori} fixed ordering of its out-neighbors).
\end{enumerate}

For simplicity, we assume that during the execution of the distributed algorithm, the nodes update the weights on their outgoing edges in a synchronous manner, but it should be evident from our proof that the algorithm can also be extended to asynchronous settings where, during each iteration, a node is selected (randomly or otherwise) to update the weights on its outgoing edges based on its imbalance at that point.\footnote{A key requirement in that case is that, as long as the graph is not balanced, no node is completely excluded from selection.}

\subsection{Formal Description of Distributed Algorithm}
\label{formaldistralg}

A formal description of the proposed distributed algorithm is presented in Algorithm~\ref{algorithm:2}.

\begin{varalgorithm}{2}
\caption{Distributed balancing with integer weights}
\textbf{Input} \\ A strongly connected digraph $\mathcal{G}_d=(\mathcal{V},\mathcal{E})$ with $n=|\mathcal{V}|$ nodes and $m=|\mathcal{E}|$ edges.\\
\textbf{Initialization} \\ Set $k=0$; each node $v_j \in \mathcal{V}$ sets its outgoing edge weights as
\vspace{-0.1cm}
\begin{align*}
f_{lj}[0] = \left\{ \begin{array}{ll}
         0, & \mbox{if $v_l \notin \mathcal{N}_j^+$,}\\
         n, & \mbox{if $v_l \in \mathcal{N}_j^+$.}\end{array} \right. 
\end{align*}
Node $v_j$ also orders its out-neighbors in the set $\mathcal{N}_j^+$ in some random (but fixed) order.

\textbf{Iteration} \\ For $k=0,1,2,\dots$, each node $v_j \in \mathcal{V}$ does the following:
\vspace{-0.1cm}
\begin{enumerate}
\item It computes its weight imbalance $x_j[k] = \mathcal{S}_j^-[k] - \mathcal{S}_j^+[k]$.
\item If  $x_j[k] > 0$, it sets 
the values of the weights on its outgoing edges as $f_{lj}[k+1] = \left \lfloor \frac{\mathcal{S}_j^-[k]}{\mathcal{D}_j^+} \right \rfloor$, $ \forall v_l \in \mathcal{N}_j^+$. Then, it chooses the first $\mathcal{S}_j^-[k] - \mathcal{D}_j^+ \left \lfloor \dfrac{\mathcal{S}_j^-[k]}{\mathcal{D}_j^+} \right \rfloor$ of its outgoing edges (according to the ordering of its out-neighbors chosen during initialization), and increases their value by $1$ so that $\vert f_{lj} - f_{hj} \vert \leq 1, \forall v_l,v_h \in \mathcal{N}_j^+ $. 
\item If  $x_j[k] < -1$, it does the following:
\\ (i) If $\left \lfloor \frac{\mathcal{S}_j^-[k]}{\mathcal{D}_j^+} \right \rfloor \geq 1$, then node $v_j$ sets 
the values of the weights on its outgoing edges as $f_{lj}[k+1] = \left \lfloor \frac{\mathcal{S}_j^-[k]}{\mathcal{D}_j^+} \right \rfloor$, $ \forall v_l \in \mathcal{N}_j^+$. Then, it chooses the first $1+ \mathcal{S}_j^-[k] - \mathcal{D}_j^+ \left \lfloor \dfrac{\mathcal{S}_j^-[k]}{\mathcal{D}_j^+} \right \rfloor$ of its outgoing edges (according to the ordering of its out-neighbors chosen during initialization), and increases their weight by $1$ so that $\vert f_{lj} - f_{hj} \vert \leq 1, \forall v_l,v_h \in \mathcal{N}_j^+ $. 
\\ (ii) If $\left \lfloor \frac{\mathcal{S}_j^-[k]}{\mathcal{D}_j^+} \right \rfloor = 0$, then node $v_j$ sets 
the values of the weights on its outgoing edges as $f_{lj}[k+1] = 1$. 
\item It repeats (increases $k$ to $k+1$ and goes back to Step~1).
\end{enumerate}
\label{algorithm:2}
\end{varalgorithm}

\subsection{Illustrative Example of Distributed Algorithm}
\label{exampledistralg}

Consider the digraph $\mathcal{G}_d=(\mathcal{V},\mathcal{E})$ in Figure~\ref{ex2-initial}, where $\mathcal{V}= \{ v_1,v_2,\dots,v_7 \}$, $\mathcal{E}= \{ e_1,e_2,\dots,e_{13} \}$, $\mathcal{E} \subseteq \mathcal{V} \times \mathcal{V} - \{ (v_j,v_j)$ $|$ $v_j \in \mathcal{V} \}$. The weight on each edge is initialized to $f_{lj}[0]=7$ for $(v_l,v_j) \in \mathcal{E}$ (otherwise $f_{lj}[0]=0$).

\begin{figure}[ht] 
\centering    
\includegraphics[width=0.60\textwidth]{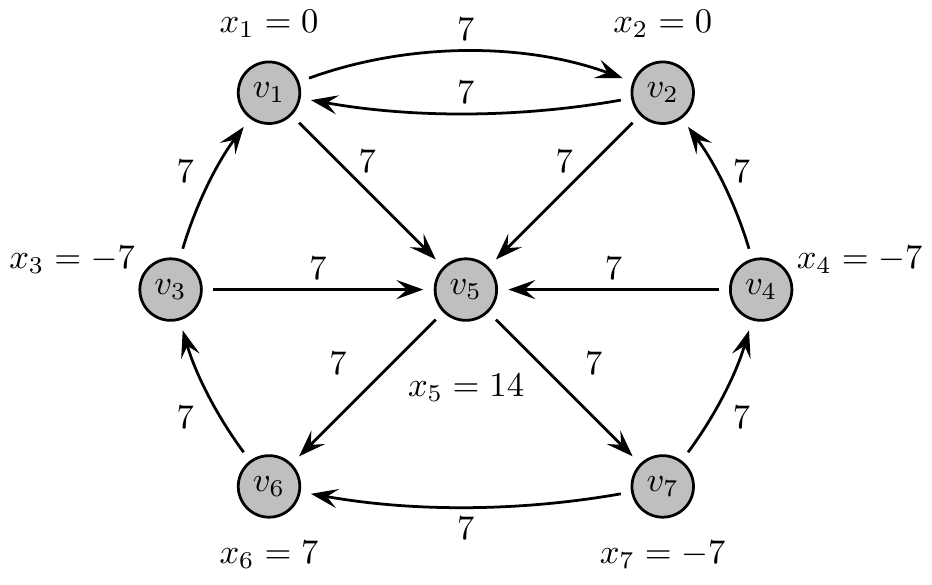}
\caption{Weighted digraph with initial weights and initial imbalance for each node.}
\label{ex2-initial}
\end{figure}

\noindent As a first step, each node computes its weight imbalance $x_j[0]=\mathcal{S}_j^-[0] - \mathcal{S}_j^+[0]$ (the corresponding imbalances are shown in Figure~\ref{ex2-initial}). Algorithm~\ref{algorithm:2} requires each node with positive imbalance to increase the value of the weights on its outgoing edges by equal integer amounts (or with maximum difference between them equal to unity), so that the total increase makes the node balanced. This ensures that weights remain strictly positive and $\mathcal{S}_j^+[k+1] = \mathcal{S}_j^-[k]$. In particular, the balance of node $v_j$ will become zero, unless the weights of its incoming edges are changed by its in-neighbors. In this case, the nodes that have positive imbalance (equal to $14$ and $7$, respectively) are nodes $5$ and $6$, which distribute their imbalance to their outgoing edges as shown in Figure~\ref{ex2-distr}. For example, node 5 has imbalance $x_5[0] = 14$ and sets $f_{65}[1]=7+7=14$ and $f_{75}[1]=7+7=14$.

\begin{figure}[ht] 
\centering    
\includegraphics[width=0.60\textwidth]{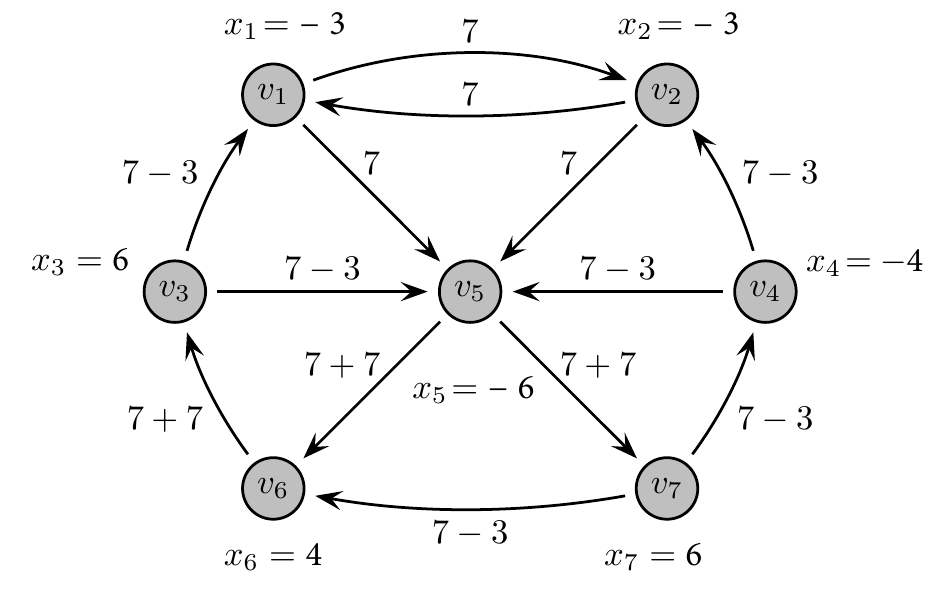}
\caption{Distribution of weights to outgoing edges by nodes with positive imbalance.}
\label{ex2-distr}
\end{figure}

The distributed algorithm also requires each node with negative imbalance to decrease the value of the weights on its outgoing edges by equal integer amounts (or with maximum difference between them equal to unity), so that i) the weights are at least $1$ and ii) the total decrease makes the nodes imbalance equal to $-1$. This ensures that weights remain strictly positive and $\mathcal{S}_j^+[k+1] = \mathcal{S}_j^-[k] +1$. The balance of node $v_j$ remains negative at $-1$, unless the weights of its incoming edges are changed by its in-neighbors. In this case, the nodes that have negative imbalance (all equal to $-7$) are nodes $3$, $4$ and $7$, which distribute their imbalance to their outgoing edges as shown in Figure~\ref{ex2-distr}. For example, node $7$ has imbalance $x_7[0] = -7$ and sets $f_{67}[1]=7-3=4$ and $f_{47}[1]=7-3=4$.

In the next iteration, after the integer weight update on the outgoing edges of each node with positive (or negative) imbalance at $k=0$, the nodes recalculate their imbalance as $x_j[1]=\mathcal{S}_j^-[1] - \mathcal{S}_j^+[1]$, and the process is repeated. After a finite number of iterations, which we explicitly bound in the next section, we reach the weighted digraph with the integer weights shown in Figure~\ref{ex2-final}. In this example, this occurs after $17$ iterations.

\begin{figure}[ht] 
\centering    
\includegraphics[width=0.60\textwidth]{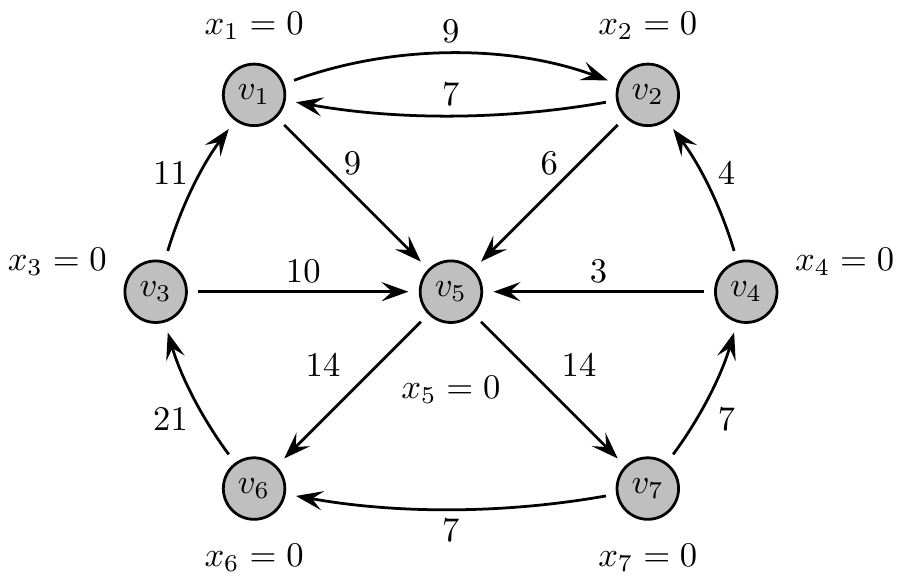}
\caption{Weight balanced digraph after $17$ iterations.}
\label{ex2-final}
\end{figure}

\begin{remark}
Note that Algorithm~\ref{algorithm:2} resembles \cite{2014:RikosHadj} in many ways: when a node has positive imbalance, it increases the weights on its outgoing edges; whereas when it has negative imbalance, it decreases the weights on its outgoing edges. The restriction to integers, however, creates some ``anomalies" that need to be addressed. In particular, one has to worry about the following cases:
\\ (i) When the imbalance is positive but not divisible by the out-degree of the node, it is not possible to increase the weights of the outgoing edges by equal integer amounts. In order to keep the weights on its outgoing edges integer-valued, the node makes the needed adjustments while allowing the weights to differ by at most~$1$.
\\ (ii) The same applies when the imbalance is negative (in which case the weights have to decrease). 
An obvious additional constraint in this case is the fact that the weights have to be kept strictly positive. 
\\ (iii) An interesting feature of Algorithm~\ref{algorithm:2} is that when the imbalance of a node is negative, the node adjusts the weights on its outgoing edges so that it achieves an imbalance of $-1$ (not zero) assuming no changes in their incoming weights. The reasons for this choice become clear in the proof of termination later on.
\end{remark}

\begin{remark}
It is worth pointing out that Algorithm~\ref{algorithm:2} is similar to the weight balancing algorithm presented in \cite{2013:RikosHadj}, but allows nodes to adjust the weights of their outgoing edges regardless of whether they have positive or negative imbalance; this means that Algorithm \ref{algorithm:2} will typically converge faster than the algorithm presented in \cite{2013:RikosHadj} (though in the worst case they will take a similar number of iterations to converge). 
\end{remark}

\section{Execution Time Analysis of Distributed Algorithm}
\label{analysisdistralg}

Herein, we calculate an explicit bound on the number of iterations required, \textit{in the worst-case}, for the given digraph to become balanced. Our bound is $O(n^3 m^2)$ iterations and, since $m < n^2$, we can also bound the number of iterations as $O(n^7)$. Note that the bound we obtain also applies for the imbalance-correcting algorithm in \cite{2012:CortesJournal} (note that \cite{2012:CortesJournal} established convergence in finite time for the imbalance-correcting algorithm, but did not provide a bound on the number of iterations). Thus, the importance of the result in this section is that the number of iterations (for both our proposed algorithm and for the imbalance-correcting algorithm in \cite{2012:CortesJournal}) is polynomial in the size of the digraph and not exponential.

Let the total imbalance of the digraph at iteration $k$ be
\vspace{-0.15cm}
$$
\varepsilon[k] = \sum_{j=1}^n | x_j[k] | \;, 
$$
where $x_j[k] = \mathcal{S}_j^-[k] - \mathcal{S}_j^+[k]$ with $\mathcal{S}_j^-[k] = \sum_{v_i \in \mathcal{N}_j^-} f_{ji}[k]$ and $\mathcal{S}_j^+[k]= \sum_{v_l \in \mathcal{N}_j^+} f_{lj}[k]$. Clearly, the total imbalance is a nonnegative quantity that is zero if and only if the digraph is balanced. Also, since we have $\sum_j x_j[k] = 0$ (because each weight $f_{ji}[k]$ appears twice, once with a positive sign and once with a negative sign), we see that: (i) the total imbalance of the digraph at any given $k$ is an even number, and (ii) if there is a node with positive imbalance, then there is also (at least one) node with negative imbalance. For convenience, in the remainder of this section, we will sometimes refer to nodes with positive (negative) imbalance as {\em positive} ({\em negative}) nodes, and to nodes with zero imbalance as {\em balanced} nodes.

\begin{prop}
\label{prop1}
During the execution of Algorithm~\ref{algorithm:2}, we have 

\vspace{-0.35cm}
$$
\varepsilon[k+1] \leq \varepsilon[k] \leq \varepsilon[0] \leq n^2 (n-2) \; .
$$
\end{prop}

\vspace{-0.15cm}
\begin{proof}
Since the digraph is assumed to be strongly connected, each node has at least one incoming edge and at least one outgoing edge. Also, each node has at most $n-1$ incoming edges and at most $n-1$ outgoing edges. Since at initialization all edges have weight $n$, we have $|x_j[0]| \leq n(n-2)$, which means that $\varepsilon[0] \leq n^2(n-2)$.

To gain some insight, assume (for now) that at iteration $k$, only node $v_j$ changes the weights on its outgoing edges  (with the other nodes {\em not} making any changes regardless of their imbalance). We have the following three cases to consider:

\noindent 
(i) {\em Case 1: Node $v_j$ has positive imbalance $x_j[k]$.} In such case, node $v_j$ uniformly\footnote{It is possible that some weights are not increased, but, due to the fixed ordering of the out-neighbors of node $v_j$ (with which weights are increased by one over the value of other weights), we have $f_{lj}[k+1] \geq f_{lj}[k]$ for all $v_l \in \mathcal{N}_j^+$.} increases the weights on its outgoing edges in such a way so that $x_j[k+1]=0$ (because $\mathcal{S}_j^+[k+1] = \mathcal{S}_j^-[k]$ and $\mathcal{S}_j^-[k+1] = \mathcal{S}_j^-[k]$ ---since no other node updates the weights on its outgoing edges).
In order to see how the total imbalance changes, we look at 
\vspace{-0.05cm}
$$
\begin{array}{rcl}
\varepsilon[k+1]-\varepsilon[k] & = & \sum_{l = 1}^n | x_l[k+1] | - | x_l[k] | \\
& \stackrel{(a)}{=} & -x_j[k] + \sum_{v_l \in \mathcal{N}_j^+} | f_{lj}[k+1]-f_{lj}[k] + x_l[k] | - |x_l[k]| \\
& \stackrel{(b)}{\leq} & -x_j[k] + \sum_{v_l \in \mathcal{N}_j^+}  | f_{lj}[k+1]-f_{lj}[k] | \\
& \stackrel{(c)}{=} & -x_j[k] + \sum_{v_l \in \mathcal{N}_j^+}  (f_{lj}[k+1]-f_{lj}[k]) \\ & = & 0 \; .
\end{array}
$$
(a) follows because $x_j[k+1]=0$, $x_j[k]$ is positive, and only nodes in $\mathcal{N}_j^+$ see changes in their weights from incoming edges; in particular, node $v_l$ sees its incoming weight from node $v_j$ increase by $f_{lj}[k+1]-f_{lj}[k]$. (b) stems from the triangle inequality, and (c) from Step~2 of Algorithm~\ref{algorithm:2} and the fact that $f_{lj}[k+1] \geq f_{lj}[k]$ for $v_l \in \mathcal{N}_j^+$. Note that equality holds in Case~1 if {\em all} nodes in $\mathcal{N}_j^+$ have positive or zero balance.


(ii) {\em Case 2: Node $v_j$ with negative imbalance $x_j[k] \leq -2$.} In this case, node $v_j$ decreases the weights on its outgoing edges so that $x_j[k+1]=-a_j$ for some integer $a_j \geq 1$. The aim is for $a_j$ to be unity (in Step~3.1 of Algorithm~\ref{algorithm:2}) but this may not be possible as the weights on the outgoing edges of $v_j$ are constrained to remain above unity (in Step~3.2 of Algorithm~\ref{algorithm:2}).
Using similar arguments as in the previous case (but keeping in mind the difference in the signs of the various quantities), we have
$$
\begin{array}{rcl}
\varepsilon[k+1]-\varepsilon[k] & = & a_j-|x_j[k]| + \sum_{l \neq j} | x_l[k+1] | - | x_l[k] | \\
& = & a_j + x_j[k] + \sum_{v_l \in \mathcal{N}_j^+} | f_{lj}[k+1]-f_{lj}[k] + x_l[k] | - |x_l[k]| \\
& \leq & a_j + x_j[k] + \sum_{v_l \in \mathcal{N}_j^+}  | f_{lj}[k+1]-f_{lj}[k] | \\
& \stackrel{(d)}{=} &  a_j + x_j[k] - \sum_{v_l \in \mathcal{N}_j^+}  (f_{lj}[k+1]-f_{lj}[k]) \\ & = & 0 \; , 
\end{array}
$$
where (d) follows from the fact that the total decrease in the weights satisfies
$$
\begin{array}{rcl}
\sum_{v_l \in \mathcal{N}_j^+}  (f_{lj}[k+1]-f_{lj}[k]) & = & \mathcal{S}_j^+[k+1] - \mathcal{S}_j^+[k] \\
 & = & \mathcal{S}_j^-[k]+a_j - \mathcal{S}_j^+[k] \\
 & = & a_j + x_j[k]  \; .
\end{array}
$$
In Case~2, equality holds if all nodes in $\mathcal{N}_j^+$ have negative or zero balance.

(iii) {\em Case 3: Node $v_j$ with $x_j[k] = -1$ or $x_j[k]=0$.} In this case, node $v_j$ does not do anything so we easily conclude that $\varepsilon[k+1]-\varepsilon[k] = 0$.


Clearly, the arguments in the above three cases establish that $\varepsilon[k+1] \leq \varepsilon[k]$ in an asynchronous setting where, at each iteration $k$, a single node $v_j$ is selected (randomly or otherwise) to update the weights on its outgoing edges based on its imbalance $x_j[k]$ at that point. In fact, with a little bit of book-keeping, we can extend the above argument in the synchronous setting of Algorithm~\ref{algorithm:2}, where nodes adjust the weights on their outgoing edges simultaneously. The key difference is that now the weights on the incoming edges of each node $v_j$ may change: $\mathcal{S}_j^-[k+1]$ is no longer identical to $\mathcal{S}_j^-[k]$ and has to be explicitly accounted for. Letting $\Delta f_{ji} = f_{ji}[k+1] - f_{ji}[k]$, we have
$$
\begin{array}{rcl}
\varepsilon[k+1]-\varepsilon[k] & = & \sum_{j=1}^n |x_j[k+1]| - | x_j[k] | \\
& = & \sum_{v_j \in \mathcal{V}} |\mathcal{S}_j^-[k+1] - \mathcal{S}_j^+[k+1] | - | x_j[k] | \\
& = & \sum_{v_j \in \mathcal{V}} | \sum_{v_i \in \mathcal{N}_j^-} \Delta f_{ji} - \sum_{v_l \in \mathcal{N}_j^+} \Delta f_{lj} + x_j[k] | - | x_j[k] | .
\end{array}
$$

Consider now the following partition of $\mathcal{V}$: set $\mathcal{A} = \{ v_j \; | \; x_j[k] > 0 \} $, set $\mathcal{B} = \{ v_j \: | \: x_j[k] \leq -2 \}$, and set $\mathcal{C} = \{ v_j \; | \; x_j[k] = -1 \text{ or } x_j[k]=0 \}$. Using these partitions, 
\vspace{-0.15cm}
$$
\begin{array}{rcl}
\varepsilon[k+1]-\varepsilon[k] & \stackrel{(e)}{=} & \sum_{v_j \in \mathcal{A}} | \sum_{v_i \in \mathcal{N}_j^-} \Delta f_{ji} | - x_j[k] \\
& \stackrel{(f)}{+} & \sum_{v_j \in \mathcal{B}} | \sum_{v_i \in \mathcal{N}_j^-} \Delta f_{ji} -a_j | + x_j[k] \\
& \stackrel{(g)}{+} & \sum_{v_j \in \mathcal{C}} | \sum_{v_i \in \mathcal{N}_j^-} \Delta f_{ji} + x_j[k] | + x_j[k] \; , 
\end{array}
$$
where (e) follows from Case~1, (f) follows from Case~2, and (g) from the fact that $x_j[k] \leq 0$ for $v_j \in \mathcal{C}$. Using the triangle inequality on each line we have 
\vspace{-0.25cm}
$$
\begin{array}{rcl}
\varepsilon[k+1]-\varepsilon[k] & \leq & \sum_{v_j \in \mathcal{A}} \sum_{v_i \in \mathcal{N}_j^-} | \Delta f_{ji} | -  x_j[k] + \sum_{v_j \in \mathcal{B}} \sum_{v_i \in \mathcal{N}_j^-} | \Delta f_{ji} | \\  & + & a_j + x_j[k] + \sum_{v_j \in \mathcal{C}} \sum_{v_i \in \mathcal{N}_j^-} | \Delta f_{ji} | \\
& = & \sum_{v_j \in \mathcal{A}} \sum_{v_l \in \mathcal{N}_j^+} |\Delta f_{lj}| - x_j[k] + \sum_{v_j \in \mathcal{B}} \sum_{v_l \in \mathcal{N}_j^+} |\Delta f_{lj}| \\ & + & a_j + x_j[k] + \sum_{v_j \in \mathcal{C}} \sum_{v_l \in \mathcal{N}_j^+} |\Delta f_{lj}| \\
& = & 0 ,
\end{array}
$$
where the key was to re-arrange the summation of the $|\Delta f_{ji}|$ and to take advantage of the inequalities we proved earlier, in Cases~1, 2~and~3.
%
\end{proof}

Note that Proposition~\ref{prop1} essentially gives us a way to analyze the execution time of the proposed algorithm (and also the algorithm in \cite{2012:CortesJournal}). The basic idea is to bound the number of steps $K$ it takes for $\varepsilon[k+K]$ to become strictly smaller than $\varepsilon[k]$. In Proposition~\ref{propbound} below, we argue that $K \leq m^2$, where $m$ is the number of edges of the digraph; this implies that the execution time of the algorithm can be bounded by 
\vspace{-0.15cm}
$$
\text{Execution Time} \leq m^2 \frac{\varepsilon[0]}{2} \leq [n(n-1)]^2 \frac{n^2(n-2)}{2} \leq \alpha n^7 \; \vspace{-0.15cm}
, 
$$
where $\alpha$ is a natural number (i.e.,  $\alpha \in \mathbb{N}$). 
[Note that the number of edges $m$ in a digraph satisfies $m \leq n(n-1)$, the initial total imbalance satisfies\footnote{Note that the total initial imbalance depends on the initial weights that each node assigns to its outgoing edges. For the case where each node initializes the weights of its outgoing edges to be equal to unity then the execution time of the proposed distributed algorithm can be bounded by $O(n^6)$ time steps.} $\varepsilon[0] \leq n^2(n-2)$, and each decrease that takes place is by at least two (since $\varepsilon[k]$ is always an even number).]

\begin{prop}
\label{propbound}
During the execution of Algorithm~\ref{algorithm:2}, we have 
$$
\varepsilon[k+K] < \varepsilon[k], \; \; \; k=0, 1, 2, ...
$$
when $\varepsilon[k] > 0$ and $K > m^2$ ($m = | \mathcal{E} |$ is the number of edges in the given digraph $\mathcal{G}_d$).
\end{prop}

Before providing the proof of Proposition~\ref{propbound}, we discuss an example of a ``difficult'' digraph in order to provide intuition about the problem.

\begin{example}
\label{EXAMPLE2}
\textup{
Consider the digraph $\mathcal{G}_d=(\mathcal{V},\mathcal{E})$ in Figure~\ref{example_converg}, where $\mathcal{V}= \{ v_1,v_2,\dots,v_8 \}$, $\mathcal{E}= \{ e_1,e_2,\dots,e_{15} \}$, $\mathcal{E} \subseteq \mathcal{V} \times \mathcal{V} - \{ (v_j,v_j)$ $|$ $v_j \in \mathcal{V} \}$. Edges $\{ e_1,e_2,\dots,e_{15} \}$ are not denoted in the figure to avoid cluttering the diagram. The weight on each edge is $f_{lj}[0]=1$ for $(v_l,v_j) \in \mathcal{E}$ (otherwise, $f_{lj}[0]=0$). $\mathcal{G}_d$ involves of 4 cycles $C_1, C_2, C_3, C_4$, which comprise of the following edges: \\
\begin{tabular}{ l l }
  $C_1:$ & $<(v_2, v_1),(v_3, v_2),(v_1, v_3)> $, \\
  $C_2:$ & $<(v_4, v_2),(v_5, v_4),(v_3, v_5),(v_2, v_3)> $,  \\
  $C_3:$ & $<(v_6, v_4),(v_7, v_6),(v_5, v_7),(v_4, v_5)> $,  \\
  $C_4:$ & $<(v_8, v_6),(v_7, v_8),(v_6, v_7)> $.  \\
\end{tabular} }

\begin{figure}[ht] 
\centering    
\includegraphics[width=0.60\textwidth]{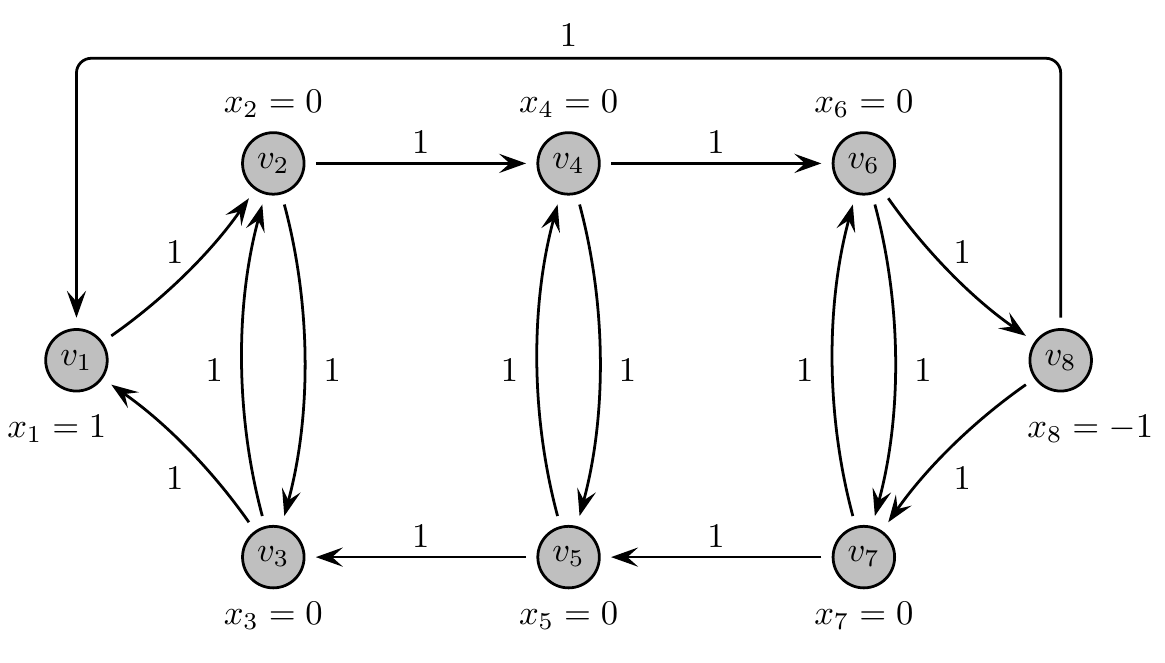}
\caption{Weighted digraph with initial weights and initial imbalances.}
\label{example_converg}
\end{figure}

\textup{
Initially, each node computes its weight imbalance $x_j[0]=\mathcal{S}_j^-[0] - \mathcal{S}_j^+[0]$ (this initial imbalance is indicated in Figure~\ref{example_converg}). 
The only node that will take action is node $v_1$, which has positive imbalance equal to $1$ (all other nodes have imbalance zero or $-1$). 
Note that node $v_1$ will increase the weight on edge $(v_2, v_1)$, causing an imbalance to node $v_2$ who will then be forced to take action. Depending on whether $v_2$ increases the weight on edge $(v_3,v_2)$ or $(v_4,v_2)$, node $v_3$ or node $v_4$ will be forced to take action. Thus, there are different possibilities (executions of the algorithm) that depend on the fixed ordering of out-neighbors, and below we consider a particular such execution. [Note that at each iteration, only one node takes action because all other nodes have imbalance zero or $-1$: initially node $v_1$, then node $v_2$, etc.; thus, we use the term ``transferring of the imbalance" to indicate the fact that the imbalanced node forces an out-neighbor to be imbalanced.]
\begin{itemize}
\item In the first 3 iterations, the imbalance gets transferred to nodes $v_2$, $v_3$ and then back to node $v_1$. 
Note that the choices we made here were for node $v_2$ to transfer the imbalance to node $v_3$, and for node $v_3$ to transfer the imbalance back to node $v_1$.
\item In the next 7 iterations, the imbalance gets transferred to nodes $v_2$, $v_4$, $v_5$, $v_3$, $v_2$, $v_3$ and back to $v_1$. Again, at each iteration, only one node has positive imbalance: first node $v_1$, then node $v_2$, then $v_4$, then $v_5$, then $v_3$, then $v_2$, then $v_3$, and finally $v_1$. Note that, given the previous choices, the only choice we had here was whether at node $v_4$ we increase the weight at edge $(v_6, v_4)$ or edge $(v_5, v_4)$; we assumed the latter. 
\item In the next 11 iterations, the imbalance gets transferred to nodes $v_2$, $v_4$, $v_6$, $v_7$, $v_5$, $v_4$, $v_5$, $v_3$, $v_2$, $v_3$, and back to $v_1$, respectively. 
\item In the last 4 iterations, the imbalance gets transferred to nodes $v_2$, $v_4$, $v_6$, and $v_8$ respectively.
\item The resulting balanced digraph is shown in Figure~\ref{example_converg_bal} and is reached after 25 iterations.
\end{itemize}
Summarizing, we have that cycle $C_1$ was crossed four times, cycle $C_2$ was crossed three times, cycle $C_3$ was crossed two times, and cycle $C_4$ was crossed one time. As a result, the number of iterations required, in order for digraph $\mathcal{G}_d$ to reach weight balance, is $4|C_1| + 3|C_2| + 2|C_3| +| C_4|$, where $|C_i|$ is the length of cycle $C_i$. }

\begin{figure}[ht] 
\centering    
\includegraphics[width=0.60\textwidth]{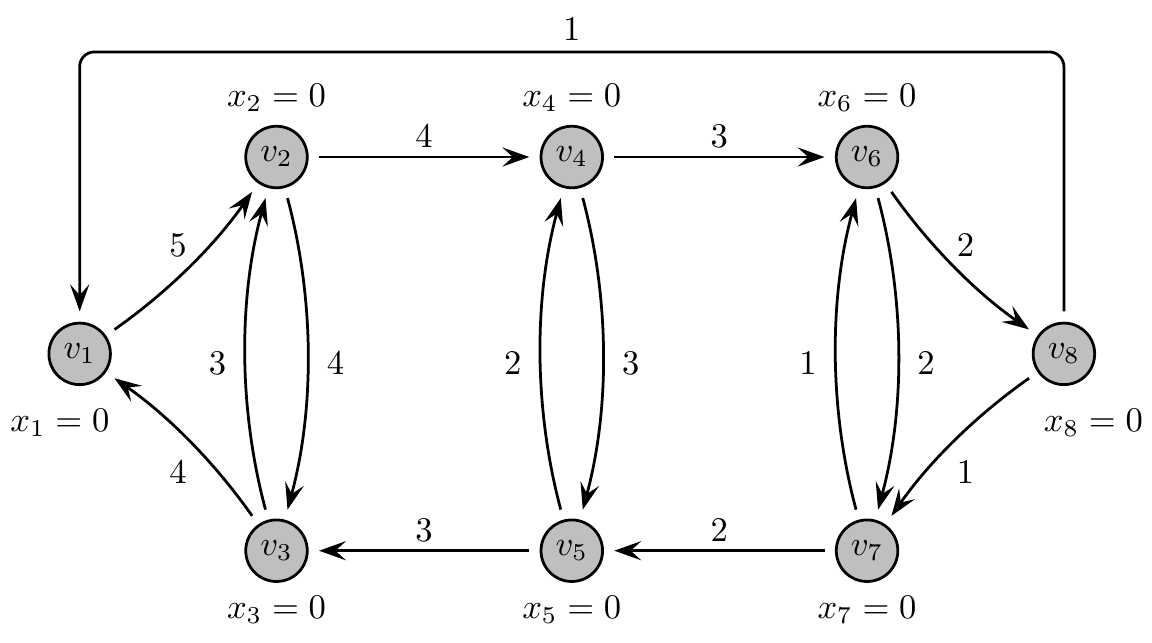}
\caption{Weight balanced digraph after 25 iterations.}
\label{example_converg_bal}
\end{figure}
\end{example}

\noindent We are now ready to proceed with the proof of Proposition~\ref{propbound}.

\begin{proof}
Since $\varepsilon[k] > 0$, we have at least one node with positive imbalance, say node $v_1$, and at least one node with negative imbalance, say node $v_n$. 

Note that, at each iteration of Algorithm~\ref{algorithm:2}, node $v_n$ (in fact, any node with negative imbalance) will retain its negative imbalance unless at least one of its in-neighbors $v_{n_i}$, $(v_n, v_{n_i}) \in \mathcal{E}$, has {\em positive} imbalance and increases the weight on the edge $(v_n, v_{n_i})$. The reason is that the changes that $v_n$ initiates on the weights on its outgoing edges can only make its imbalance $-1$; thus, for the imbalance to become zero or positive, it has to be that one or more of its incoming weights are increased. This can only happen if one or more of its in-neighbors have positive imbalance, at which point it follows from the proof of Proposition~\ref{prop1} (strict inequality for Case~1) that the total imbalance will decrease (by at least two).

In order to determine a bound on the number of steps $K$ required for the total imbalance to decrease, we can assume without loss of generality that negative nodes remain negative (because at the moment any negative node becomes balanced or positive, we also have a decrease in the total imbalance). Consider the (worst\footnote{It will become evident that this is the worst case scenario at the completion of the argument.}) case where $v_1$ has imbalance $b$ for some positive integer $b$, $v_n$ has imbalance $-b$, and the remaining nodes $v_2$, $v_3$, ..., $v_{n-1}$ are all balanced (refer to Figure~\ref{FIGgeneralcase}, where $v_1$ is the node on the far left and $v_n$ is the node on the far right). At the first iteration, node $v_1$ sends its imbalance to at least one of its out-neighbors (by increasing the weight on at least one of its outgoing edges). This out-neighbor (resp. these out-neighbors) of node $v_1$ does (resp. do) the same at the next iteration, and this process is repeated. If, at any point, node $v_n$ is reached, the overall imbalance will decrease by at least two (i.e., $\varepsilon[k+1] \leq \varepsilon[k] - 2$).

\begin{figure}[ht] 
\centering    
\includegraphics[width=0.75\textwidth]{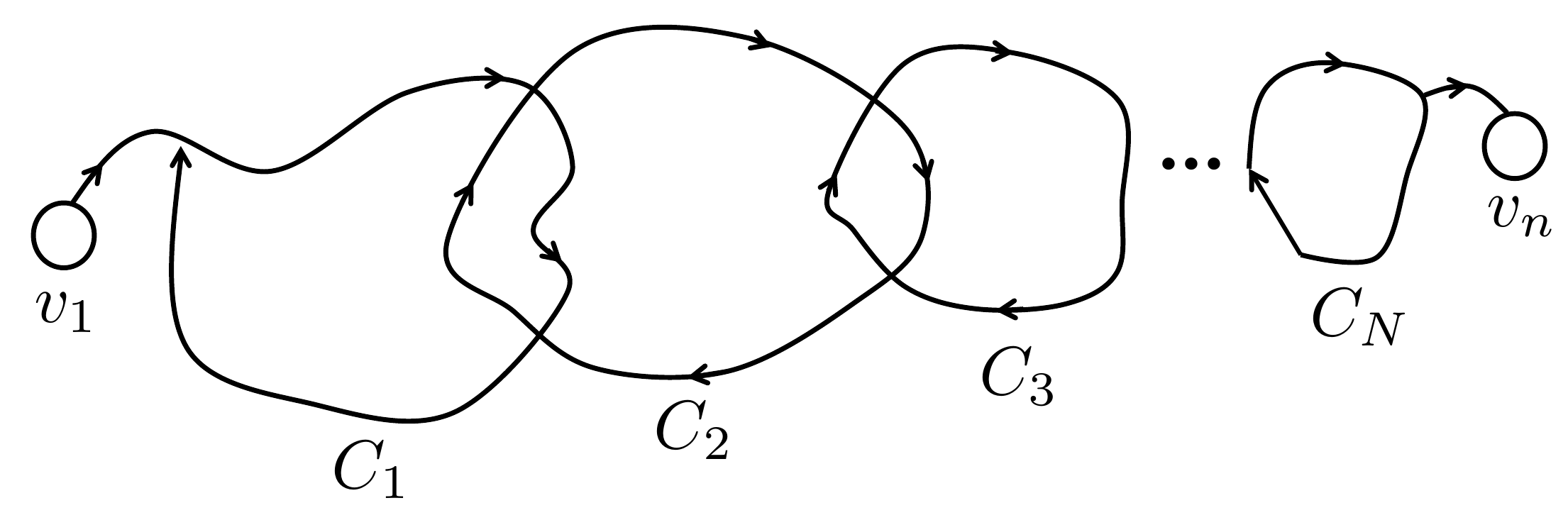}
\caption{Transfer of positive imbalance from node $v_1$ on the left to node $v_n$ on the right.}
\label{FIGgeneralcase}
\end{figure}

Let us now analyze the number of iterations it takes for node $v_n$ to be reached (i.e., for one of its in-neighbors to become positive). Observe the following:
\\ (i) At each iteration, the imbalance gets transferred from a node to one or more of that nodes's out-neighbors. If more than one out-neighbors are involved, then the decrease will occur faster; thus, we consider the case when the imbalance gets transferred to only one out-neighbor.
\\ (ii) As the iterations proceed, mark the first time an edge is traversed for the second time, and call the cyclic sequence of edges visited up to this point $C_1$ (refer to Figure~\ref{FIGgeneralcase}). Note that $C_1$ is a cyclic sequence of edges (not nodes, i.e., a certain node may be visited more than once while traversing cycle $C_1$). Also note that $C_1$ has at most $n(n-1)$ edges because that is the total number of edges of the digraph.
\\ (iii) While traversing $C_1$ for the second time, we will be forced at some point to traverse a new edge (that has not been traversed before (otherwise, the digraph is not strongly connected). The reason is the fact that the weights on the out-going edges from each node cannot differ by more than unity (this is where the notion of approximate balance among the weights on the out-going edges plays are role). Let $C_2$ denote the set of edges traversed until the time we stop traversing new edges (i.e., we are forced to traverse an edge that we have already visited). Let $C_2$ be the set of edges that are not in $C_1$ and have been visited so far.
\\ (iv) We can continue in this fashion (defining $C_i$ as shown in Figure~\ref{FIGgeneralcase}) until we reach node $v_n$. Note that the number of cycles $N$ satisfies $N \leq m$ (where $m=|\mathcal{E}|$ is the number of edges of the given digraph $\mathcal{G}_d$) because each cycle has at least one edge and the digraph has a total of $n(n-1)$ cycles. 

\noindent From the above discussion (see also Example~\ref{EXAMPLE2}), we have
\vspace{-0.2cm}
$$
\text{\#(iterations to reach } v_n) \leq \sum_{i=1}^N (N-i+1) | C_i | \leq N m \; \vspace{-0.2cm},
$$
since $\sum_{i=1}^N |C_i| \leq m$. Finally, each cycle $C_i$ has at least one edge, which means we can have at most $m$ cycles. This allows us to conclude that 
$
\text{\#(iterations to reach } v_n) \leq m^2 \; ,
$
which completes the proof of the proposition.
\end{proof}

\begin{remark}
The bound we obtain is actually $m^2 \frac{\varepsilon[0]}{2}$ and can easily be improved (if one looks at the last equation in Section IV and realizes that it is in our best interest to have cycle 1 have size $|C_1| = m-(N-1)$ and all other cycles have size equal to $1$) to $m^2 \frac{\varepsilon[0]}{2}$. In terms of the example we provide, this bound suggests $m^2/2$ iterations which is $113$. Nevertheless, the main motivation in obtaining our bound was to establish that the number of iterations is bounded by a number that is polynomial (and not exponential) in the size of the graph. 
\end{remark}

\section{Simulation Study}
\label{resultsdistralg}

We compare the proposed algorithms with the current state-of-the-art. 
Specifically, we run Algorithm~\ref{algorithm:2} in large digraphs (of size $n=20$ and $50$) and compare their performance against two other algorithms: (a) the weight balancing algorithm in \cite{2013:RikosHadj} in which each node $v_j$ with positive imbalance $x_j>0$, increases the weights of its outgoing edges by equal integer amounts (if possible) so that it becomes weight balanced, (b) the imbalance-correcting algorithm in \cite{2009:Cortes} in which every node $v_j$ with positive imbalance $x_j>0$ adds all of its weight imbalance $x_j$ to the outgoing node with the lowest weight $w$.

Figure~\ref{2_figures_20_nodes} shows the cases of: (i) 1000 averaged digraphs of $20$ nodes each, where every edge is initialized to $1$, and (ii) $1000$ averaged digraphs of $20$ nodes each, where every edge is initialized to $20$, respectively. 
Figure~\ref{2_figures_20_nodes} shows that for the first case Algorithm~\ref{algorithm:2} converges as fast as the one in \cite{2013:RikosHadj} (as expected due to the particular initialization used). 
For the second case we have that when the edge's initialization is greater than $1$ then Algorithm~\ref{algorithm:2} is the fastest among algorithms \cite{2013:RikosHadj, 2009:Cortes}.

\begin{figure}[ht] 
\centering    
\includegraphics[width=0.60\textwidth]{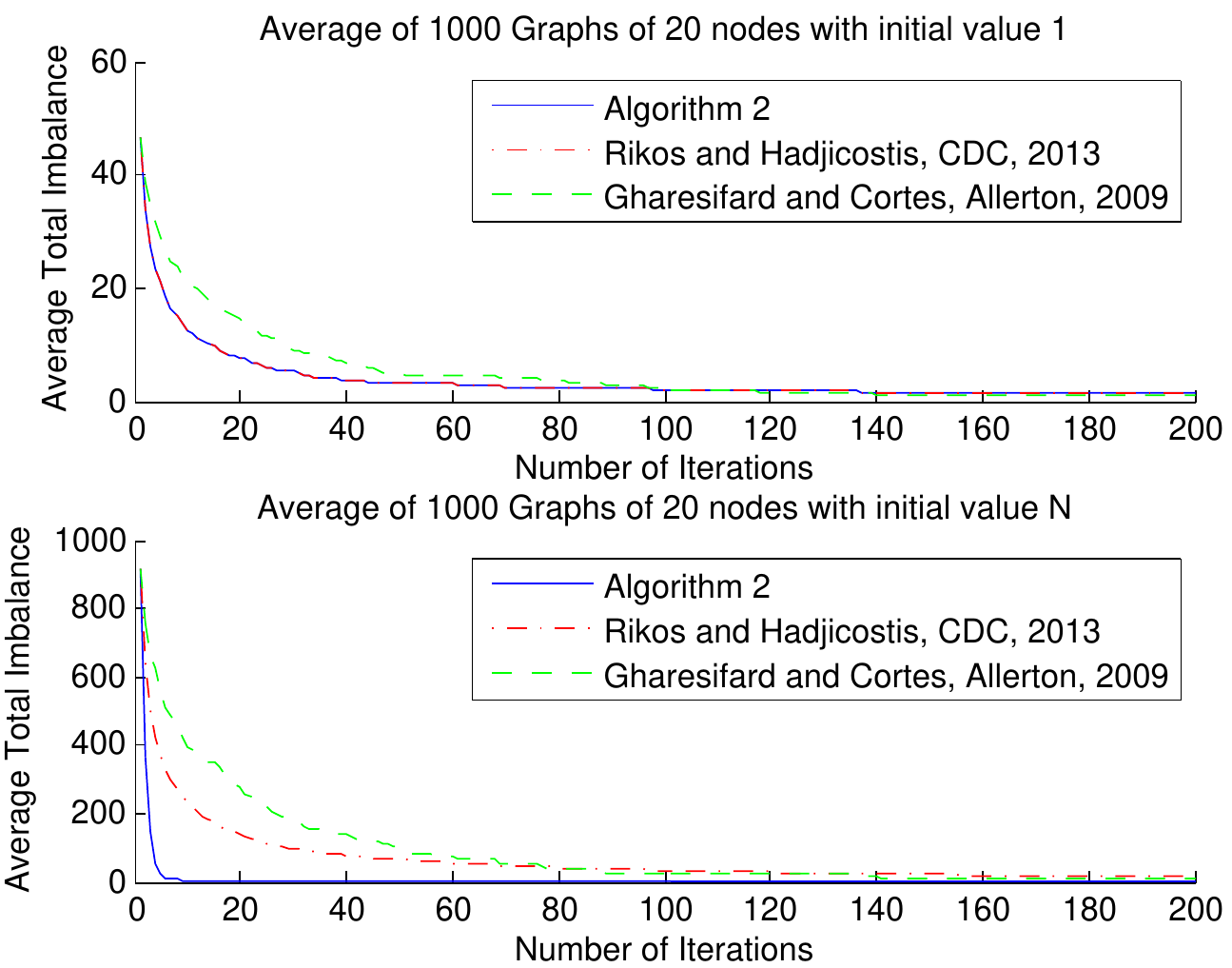}
\caption{Comparison between Algorithm~\ref{algorithm:2}, the weight balancing algorithm proposed in \cite{2013:RikosHadj} and the imbalance-correcting algorithm \cite{2009:Cortes}.
\emph{Top figure:} Average total imbalance plotted against the number of iterations for $1000$ random digraphs of $20$ nodes.
\emph{Bottom figure:} Average total imbalance plotted against the number of iterations for $1000$ random digraphs of $20$ nodes.}
\label{2_figures_20_nodes}
\end{figure}

Figure~\ref{2_figures_50_nodes} shows the same cases as Figure~\ref{2_figures_20_nodes}, with the difference that the network consists of $50$ nodes.
The performances do not change due to the network size and the conclusions are the same as in Figure~\ref{2_figures_20_nodes}.

\begin{figure}[ht] 
\centering    
\includegraphics[width=0.60\textwidth]{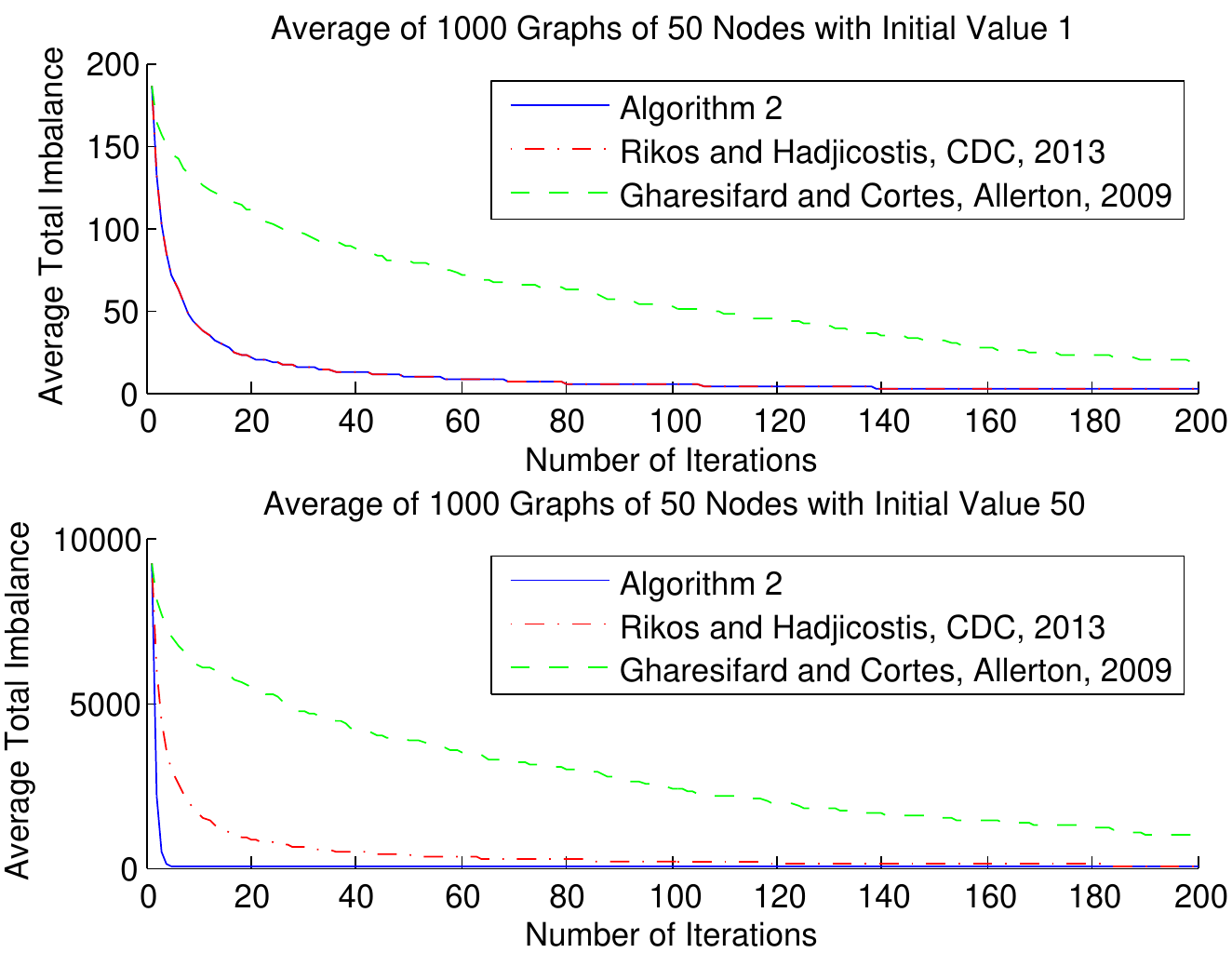}
\caption{Comparison between Algorithm~\ref{algorithm:2}, the weight balancing algorithm proposed in \cite{2013:RikosHadj} and the imbalance-correcting algorithm \cite{2009:Cortes}.
\emph{Top figure:} Average total imbalance plotted against the number of iterations for $1000$ random digraphs of $50$ nodes.
\emph{Bottom figure:} Average total imbalance plotted against the number of iterations for $1000$ random digraphs of $50$ nodes.}
\label{2_figures_50_nodes}
\end{figure}

\section{Chapter Summary}
\label{sumarycentrdistralg}

In this section, we described an iterative distributed algorithm and established that it converges to a weight balanced digraph after a finite number of steps. 
We have also bounded the execution time of the proposed algorithm as $O(n^7)$, where $n$ is the number of nodes in the digraph and we demonstrated the operation, performance, and advantages of the proposed algorithm via various simulations.

\clearpage

\lhead{\emph{Weight Balancing over Unreliable Communication}}

\chapter{Weight Balancing over \\ Unreliable Communication}
\label{distrdelpacket}

In this chapter, we present a novel distributed algorithm which deals with the problem of balancing a weighted digraph in the presence of time delays and packet drops over the communication links. 
The algorithm presented in this chapter has appeared in \cite{2015:RikosHadj, 2017:RikosHadj}. 

This chapter is organized as follows. In Section~\ref{ModelDelaysPacket} we present the additional notation needed in this chapter and the model used for the time delays and the way they manifest themselves, while in Section~\ref{ProbStatementDelays} we present the problem formulation. 
In Section~\ref{algDel} we introduce a novel distributed algorithm which achieves integer weight balancing in a multi-component system, in the presence of time delays over the communication links. 
We present a formal description of the proposed distributed algorithm, demonstrate its performance via an illustrative example and show that it converges to a weight balanced digraph after a finite number of iterations in the presence of bounded time delays over the communication links. 
In Section~\ref{triggeredalgDel} we discuss an event-triggered version of the proposed distributed algorithm and show that it results in a weight balanced digraph after a finite number of iterations in the presence of arbitrary (time-varying, inhomogeneous) but bounded time delays over the communication links.
In Section~\ref{packetalgDel} we show that the proposed distributed algorithm is also able to converge (with probability one) to a weight balanced digraph in the presence of unbounded delays (packet drops). 
In Section~\ref{resultsalgDel} we present simulation results and comparisons of the proposed distributed algorithm and the chapter is concluded in Section~\ref{summaryalgDel}.

\section{Modeling Time Delays and Packet Drops}
\label{ModelDelaysPacket}

In this chapter, we assume that node $v_j$ assigns a unique order in the set $\{0,1,..., \mathcal{D}^+_j-1\}$, to each of its outgoing edges. 
The order of link $(v_l,v_j)$ is denoted by $P_{lj}$ (such that $\{P_{lj} \; | \; v_l \in \mathcal{N}^+_j\} = \{0,1,..., \mathcal{D}^+_j-1\}$) and will be used later on as a way of allowing node $v_j$ to make changes to its outgoing edge weights in a predetermined \textit{round-robin}\footnote{Each node $v_j$ updates (increases) the weights of its out-neighbors by following a unique predetermined order. More specifically, following this predetermined order, node $v_j$ updates (increases by unity) its weights, one at a time, until it becomes balanced. The next time it needs to update the weights of its out-neighbors, it will continue from the outgoing edge it stopped the previous time and cycle through the edges in a round-robin fashion according to the unique predetermined ordering.} fashion.  

Furthremore, for the development of the results in this chapter, we assume that a transmission on the link from node $v_j$ to node $v_l$ initiated at time step $k$ undergoes an \textit{a priori unknown} delay $\tau_{lj}[k]$, where $\tau_{lj}[k]$ is an integer that satisfies $0\leq \tau_{lj}[k] \leq  \overline{\tau}_{lj} \leq \infty$ (i.e., delays are bounded\footnote{We later relax the assumption of bounded delays and consider packet dropping links.}). 
The maximum delay is denoted by $\overline{\tau} = \max_{(v_l,v_j) \in \mathcal{E}}\overline{\tau}_{lj}$. In the weight balancing setting we consider, node $v_j$ is in charge of assigning weights $f_{lj}[k]$ to each link $(v_l,v_j)$ and of sending to each out-neighbor $v_l$ the value $f_{lj}[k]$. 
Under the above delay model, the weight $f_{lj}[k]$ becomes available to node $v_l$ at time step $k + \tau_{lj}[k]$. 
From the perspective of node $v_j$, the delayed weight for link $(v_j,v_i)$, $\forall v_i \in \mathcal{N}_j^-$, at time step $k$ is given by 
\begin{align} \label{delay_equation}
\overline{f}_{ji}[k]= f_{ji}[k_{last}], \text{ where } k_{last}=\max\{s \ |\  s+\tau_{ji}[s]\leq k\},  
\end{align}
i.e., $\overline{f}_{ji}[k]$ is the most recently sent weight $f_{ji}[\cdot]$ seen at node $v_j$ by time step $k$.

The weights available on the incoming links of node $v_j$ at time step $k$ comprise the values received by its in-neighbors by that time, i.e., it is the set of values in the set $\{ f_{ji}[s]$ $|$ $0 \leq s \leq k,$ $s + \tau_{ji}[s] \leq k,$ $v_i \in \mathcal{N}_j^- \}$ where $f_{ji}[s]$ denotes the value sent by the in-neighbor $v_i$ to node $v_j$ at time step $s$, and is received at $v_j$ by time step $k$. 
The protocol we propose has each node $v_j$ update the information state at time $k$ by combining (in a linear fashion) its own outgoing-edge values $f_{lj}[k], \forall v_l \in \mathcal{N}_j^+$ (which are available without delay) and the possibly delayed information received \textit{by time step $k$} from its in-neighbors $\{f_{ji}[s]$ $|$ $0 \leq s \leq k,$ $s + \tau_{ji}[s] \leq k,$ $v_i \in \mathcal{N}_j^- \}$, i.e., the values that arrive at node $v_j$ by time $k$.

\begin{remark}
In our case, the above definition of $\overline{f}_{ji}[k]$ is equivalent to $\overline{f}_{ji}[k] = \max\{ {f}_{ji}[s] $ $ \; | \; s + \tau_{ji}[s] \leq k\}$ because, during the execution of the algorithm presented later, the weights $f_{ji}$ on each edge $(v_j,v_i)$ are assigned by node $v_i$ in a non-decreasing manner, i.e., for $v_j \in \mathcal{N}_i^+$, $f_{ji}[k] \leq f_{ji}[k+1]$, where $k \in \mathbb{N}_0$. 
\end{remark}

\begin{prop}\label{monotonic_prop}
When $f_{ji}[k]$ are non-decreasing and delays are bounded, the above delay model, implies that $\overline{f}_{ji}[k+\overline{\tau}] \geq f_{ji}[k]$. 
\end{prop}

\begin{proof}
The proof follows directly from the definition of $\overline{f}_{ji}[k]$. We have that $\overline{f}_{ji}[k+\overline{\tau}] = f_{ji}[k_{last}]$ where  $k_{last}=\max\{s \ |\  s+\tau_{ji}[s]\leq k+\overline{\tau} \}$. Obviously, $k_{last} \geq k$ because $\tau_{ji}[k] \leq \overline{\tau}$. Since the weights of each edge are non-decreasing we have that $\overline{f}_{ji}[k+\overline{\tau}] = f_{ji}[k_{last}] \geq f_{ji}[k]$. 
\end{proof}

\begin{defn}
Given a weighted digraph $\mathcal{G}_d=(\mathcal{V},\mathcal{E},\mathcal{F})$ of order $n$, the total \textit{in-weight} seen at time step $k$ by node $v_j$ is $\overline{\mathcal{S}}_j^-[k] = \sum_{v_i \in \mathcal{N}_j^-} \overline{f}_{ji}[k]$. Since, for every node, the weights of its outgoing edges are available without delay, the total \textit{out-weight} of node $v_j$ at time step $k$ is the same as in Definition \ref{DEFnodeinoutweight} (denoted by $\mathcal{S}_j^+[k]$).
\end{defn}

\begin{defn}
Given a weighted digraph $\mathcal{G}_d=(\mathcal{V},\mathcal{E},\mathcal{F})$ of order $n$, the delayed weight \textit{imbalance} of node $v_j$, calculated at time step $k$, is $\overline{x}_j[k] = \overline{\mathcal{S}}_j^-[k] - \mathcal{S}_j^+[k]$.
\end{defn}

\begin{defn}
Given a weighted digraph $\mathcal{G}_d=(\mathcal{V},\mathcal{E},\mathcal{F})$ of order $n$, the \textit{total delayed imbalance} of digraph $\mathcal{G}_d$, at a time step $k$, is $\overline{\varepsilon}[k] = \sum_{j=1}^{n} \vert \overline{x}_j[k] \vert$.
\end{defn}

Apart from bounded delays, \textit{unreliable} communication links in practical settings could also result in possible packet drops (i.e., unbounded delays) in the corresponding communication network. 
To model packet drops, we assume that, at each time step $k$, a packet that is sent from node $v_i$ to node $v_j$ on link $(v_j, v_i) \in \mathcal{E}$ is dropped with probability $q_{ji}$, where we have $q_{ji} <1$. 
For simplicity, we assume independence between packet drops at different time steps or different links. 
We establish that, in both cases, despite the presence of bounded delays or packet drops, the proposed distributed algorithm converges, after a finite number of iterations, to a weight balanced digraph that is identical to the one obtained under no packet drops (in the case of packet drops this convergence occurs with probability one).

\section{Problem Statement}
\label{ProbStatementDelays}

We are given a strongly connected digraph $\mathcal{G}_d = (\mathcal{V}, \mathcal{E})$, with a set of nodes $\mathcal{V} =  \{v_1, v_2, \dots, v_n\}$ and a set of edges $\mathcal{E} \subseteq \mathcal{V} \times \mathcal{V} - \{ (v_j,v_j)$ $|$ $v_j \in \mathcal{V} \}$. 
We want to develop a distributed algorithm that allows the nodes to iteratively adjust the weights on their edges so that they eventually obtain a set of integer weights $\{ f_{ji} \; | \; (v_j, v_i) \in \mathcal{E} \}$ that satisfy the following:
\begin{enumerate}
\item $f_{ji} \in \mathbb{N}$ for every edge $(v_j,v_i) \in \mathcal{E}$;
\item $f_{ji} = 0$ if $(v_j,v_i) \notin \mathcal{E}$;
\item $\mathcal{S}_j^+ = \mathcal{S}_j^- = \overline{\mathcal{S}}_j^-$ for every $v_j \in \mathcal{V}$.
\end{enumerate}

We introduce and analyze a distributed algorithm that allows each node to assign integer weights to its outgoing links, so that the resulting weight assignment is balanced. The proposed algorithm is able to handle arbitrary but bounded time-delays that may affect the information exchange between agents in the system. We also explicitly bound the number of steps required for convergence. Among other implications, this bound establishes that the proposed algorithm completes its operation in polynomial time, as long as the underlying digraph is strongly connected or is a collection of strongly connected digraphs.\footnote{As discussed in \cite{2009:Cortes}, this is a necessary and sufficient condition for weight-balancing to be possible.}

\section{Distributed Algorithm for Weight Balancing \\ in the Presence of Time Delays}
\label{algDel}

Given a strongly connected digraph $\mathcal{G}_d=(\mathcal{V},\mathcal{E})$, the algorithm has each node initialize the weights of all of its outgoing edges to unity. 
We consider for now that in digraph $\mathcal{G}_d$, each link transmission can undergo an arbitrary but bounded delay. In order to handle delays, we employ a strategy where the nodes run a weight balancing protocol and process weight information as soon as it arrives. According to this protocol, each node enters an iterative stage where  it performs the following steps:

\begin{enumerate}

\item It computes its delayed weight imbalance according to the latest received weight values from its in-neighbors.
\item If it has positive (delayed) imbalance, it increases by $1$ the integer weights of its outgoing edges one at a time, following the priority order until it becomes weight balanced. This means that the outgoing edges are assigned, if possible, equal integer weights; otherwise, if this is not possible, they are assigned integer weights such that the maximum difference among them is equal to one (it should be clear that for a given $\overline{\mathcal{S}}_{j}^-$ the order among the outgoing links of node $v_j$ make this assignment unique).

\end{enumerate}

We argue that the above distributed algorithm obtains integer weights that balance the digraph after a finite number of iterations (which we bound in terms of the number of nodes/edges of the given digraph). Using a path-based analysis of the algorithm, we prove that the resulting weight balanced digraph is unique and independent of the link-delays that may occur during the execution of the algorithm. For simplicity, we assume that during the execution of the distributed algorithm, the nodes update the weights on their outgoing edges in a synchronous\footnote{Even though we do not discuss this issue in the thesis, asynchronous operation is not a problem for the distributed algorithm.} manner based on the information available at each node at that particular instant. 
Note that if the delay between asynchronous changes in weights of different links can be bounded by some maximum delay then asynchronous updates can be captured by our synchronized delayed communication model if we allow an increase in ${\bar \tau}$.

\subsection{Formal Description of Distributed Algorithm}
\label{chapteralgDel}

A formal description of the proposed distributed algorithm is presented in Algorithm~\ref{algorithm:3}.

\begin{varalgorithm}{3}
\caption{Distributed balancing in the presence of time delays}
\textbf{Input} \\ A strongly connected digraph $\mathcal{G}_d=(\mathcal{V},\mathcal{E})$ with $n=|\mathcal{V}|$ nodes and $m=|\mathcal{E}|$ edges.\\
\textbf{Initialization} \\ Set $k=0$; each node $v_j \in \mathcal{V}$ does the following:
\vspace{-0.1cm}
\begin{enumerate}
\item It assigns a unique priority order to its outgoing edges as $P_{lj}$, for $v_l \in \mathcal{N}_j^+$ (such that $\{P_{lj} \ | \ v_l \in \mathcal{N}_j^+\} = \{ 0,1,...,\mathcal{D}_j^+ - 1 \}$). 
\item It sets its outgoing edge weights as
$$
f_{lj}[0] = \left\{ \begin{array}{ll}
         0, & \mbox{if $(v_l,v_j) \notin \mathcal{E}$,}\\
         1, & \mbox{if $(v_l,v_j) \in \mathcal{E}$.}\end{array} \right. 
$$
\item It sets its (perceived) incoming weights to be equal to unity, $\overline{f}_{ji}[0]=1$, $\forall v_i \in \mathcal{N}_j^-$.
\item It transmits the weights $f_{lj}[0]$ on each outgoing edge $(v_l,v_j)\in \mathcal{E}$ to each $v_l \in \mathcal{N}_j^+$. 
\end{enumerate}

\textbf{Iteration} \\ For $k=0,1,2,\dots$, each node $v_j \in \mathcal{V}$ does the following:
\vspace{-0.1cm}
\begin{enumerate}
\item It receives the weights on its incoming edges $\overline{f}_{ji}[k]$. More specifically, for each node $v_i \in \mathcal{N}_j^-$ let $\mathcal{F}_{ji} = \{f_{ji}[s+\tau_{ji}[s]]$ $|$ $s+\tau_{ji}[s] = k\}$ be the set of weights of link $(v_j,v_i) \in \mathcal{E}$ that arrive at node $v_j$ at time step $k$. We have that
$$
\overline{f}_{ji}[k+1] = \left\{ \begin{array}{ll}
         \overline{f}_{ji}[k], & \mbox{if $\mathcal{F}_{ji} = \emptyset$,}\\
         \max\{\overline{f}_{ji}[k] , \max_{\overline{f}_{ji} \in \mathcal{F}_{ji}}\{\overline{f}_{ji}\}\}, & \mbox{if $\mathcal{F}_{ji} \neq \emptyset$.}\end{array} \right. 
$$

\item It computes its weight imbalance according to the latest received weight values from its in-neighbors 
$$
\overline{x}_j[k + 1] = \overline{\mathcal{S}}_j^-[k + 1] - \mathcal{S}_j^+[k + 1]. 
$$
\item If  $\overline{x}_j[k + 1] = br_j^+ > 0$, it sets 
the values of the weights on its outgoing edges to $f_{lj}[k+1] = \left \lfloor \frac{\overline{S}_j^-[k + 1]}{\mathcal{D}_j^+} \right \rfloor$, $ \forall v_l \in \mathcal{N}_j^+$. 
Then, it chooses the set of the first (according to the priority order) $\overline{S}_j^-[k + 1] - \mathcal{D}_j^+ \left \lfloor \dfrac{\overline{S}_j^-[k + 1]}{\mathcal{D}_j^+} \right \rfloor$ outgoing edges, and increases their weight by 1 so that $\vert f_{lj}[k+1] - f_{hj}[k+1] \vert \leq 1, \forall v_l,v_h \in \mathcal{N}_j^+ $ and $\mathcal{S}_j^+[k+1] = \overline{\mathcal{S}}_j^-[k]$.
\item It transmits the new weights $f_{lj}[k+1]$ on each outgoing edge $(v_l,v_j)\in \mathcal{E}$ to the corresponding out-neighbor $v_l \in \mathcal{N}_j^+$. 
\item It repeats (increases $k$ to $k+1$ and goes back to Step~1).
\end{enumerate}
\label{algorithm:3}
\end{varalgorithm}

\noindent
The following lemma is useful in our analysis later on.

\begin{lemma}\label{weight_function}
Suppose that at iteration $k$ node $v_j$ has in-weights $\{\overline{f}_{ji}[k] \ | \ v_i \in \mathcal{N}_j^- \}$. Given a unique order $P_{lj}$, where $v_l \in \mathcal{N}_j^+$, on its outgoing links (such that $\{P_{lj} \ | \ v_l \in \mathcal{N}_j^+ \} = \{ 0,1,...,\mathcal{D}_j^+ - 1 \}$), we have that 
$$
f_{lj}[k+1] = \mathcal{F}_{lj}(\sum_{v_i \in \mathcal{N}_j^-} \overline{f}_{ji}[k]). 
$$
Moreover, $\mathcal{F}_{lj}$ is monotonically non-decreasing in its argument, i.e., $\mathcal{F}_{lj}(x) \geq \mathcal{F}_{lj}(y)$ if $x \geq y$. 
\end{lemma}

\begin{proof}
From the algorithm description we have that for integer $x$, we have that
$$
\mathcal{F}_{lj}(x) = \left \lfloor \frac{x}{\mathcal{D}_j^+} \right \rfloor + ind_l(x),
$$
where
$$
ind_l = \left\{ \begin{array}{ll}
         1, & \mbox{if $P_{lj} < x-\left \lfloor \frac{x}{\mathcal{D}_j^+} \right \rfloor \mathcal{D}_j^+$,}\\
         0, & \mbox{otherwise.}\end{array} \right.
$$
$\mathcal{F}_{lj}(x)$ is clearly monotonic in its argument. 
\end{proof}

We now illustrate the distributed algorithm via an example and then explain why it asymptotically results in a weight balanced digraph after a finite number of iterations. We also obtain bounds on its execution time.

\subsection{Illustrative Example of Distributed Algorithm}
\label{examplealgDel}

Consider the digraph $\mathcal{G}_d=(\mathcal{V},\mathcal{E})$ in Figure.~\ref{example-initial-distributed}, where $\mathcal{V}=(v_1,v_2,\dots,v_6)$, $\mathcal{E}=(m_{31},\dots,m_{46})$, $\mathcal{E} \subseteq \mathcal{V} \times \mathcal{V} - \{(v_i,v_i) \ | \ v_i \in \mathcal{V} \}$. The weight on each edge is initialized to $f_{lj}[0]=1$ for $(v_l,v_j) \in \mathcal{E}$ and each node assigns a unique priority order to each of its outgoing edges. For the purposes of this example, let us assume that the priority orders are as follows:
\begin{itemize}
\item $v_1 : P_{31}=1$,
\item $v_2 : P_{32}=1$,
\item $v_3 : P_{13}=1, P_{23}=2, P_{43}=3$,
\item $v_4 : P_{54}=1, P_{64}=2$,
\item $v_5 : P_{15}=1, P_{35}=2, P_{45}=3$,
\item $v_6 : P_{26}=1, P_{46}=2$.
\end{itemize}
\noindent (For example, node $v_4$ will first increase $f_{54}$ and then $f_{64}$). As a first step, each node computes its weight imbalance $\overline{x}_j[0] = \overline{\mathcal{S}}_j^-[0] - \mathcal{S}_j^+[0]$, $\forall v_j \in \mathcal{V}$ (these values are shown in Figure.~\ref{example-initial-distributed}).

\begin{figure} [ht]
\centering
\includegraphics[width=0.47\textwidth]{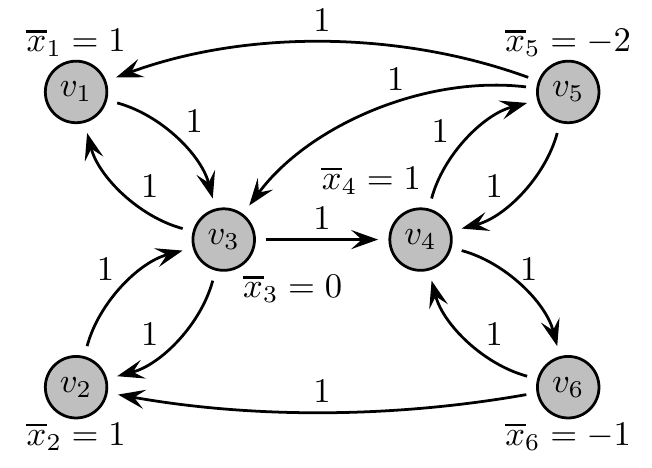}
\caption{Weighted digraph with initial weights and initial imbalances for each node.}
\label{example-initial-distributed}
\end{figure}

Once each node computes its imbalance, the distributed algorithm requires each node with positive imbalance to increase the value of the weights on its outgoing edges by equal integer amounts (or with maximum difference between them equal to one) according to the predetermined priority order that each node assigned to its outgoing edges, so that the total increase makes the node balanced. In this case, the nodes that have positive imbalance are nodes $v_1, v_2$ and $v_4$ (equal to $\overline{x}_1[0]=1, \overline{x}_2[0]=1$ and $\overline{x}_4[0]=1$) respectively), and they increase their outgoing edges as shown in Figure~\ref{example-distr-distributed}.

\begin{figure} [ht]
\centering
\includegraphics[width=0.45\textwidth]{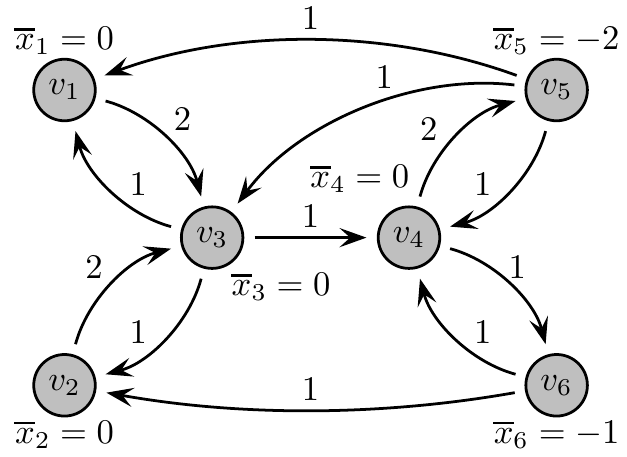}
\caption{Distribution of imbalance from positively imbalanced nodes.}
\label{example-distr-distributed}
\end{figure}

In Figure~\ref{example-distr-distributed} we can see that the edges $m_{31}, m_{32}, m_{54}$ have now values equal to $2$. Note here that the nodes $v_1, v_2$ and $v_4$ increased the edge-weights $f_{31}, f_{32}$ and $f_{54}$ respectively, since the corresponding nodes had the highest order (as chosen by the nodes during the initialization step). Nodes $v_3$ and $v_5$ will receive the new weights of their incoming edges after a number of iterations equal to the corresponding time-delay for each edge i.e., $v_3$ and $v_5$ will receive them after $\tau_{31}[0]$, $\tau_{32}[0]$ and $\tau_{54}[0]$, respectively. For example, let us consider that the time delays are equal to $\tau_{31}[0]=6$, $\tau_{32}[0]=3$ and $\tau_{54}[0]=7$. This means that node $v_3$ will receive the new weight of $m_{31}$ at $k=6$ and the new weight of $m_{32}$ at $k=3$, while $v_5$ will receive the new weight of $m_{54}$ at $k=7$. In Figure~\ref{example-distr-delay-distributed}, we can see the digraph at time step $k=5$. Here, node $v_3$ has received the new weight of edge $m_{32}$ and has increased its outgoing edge $m_{13}$ by $1$ (because it has the highest priority order) while it maintains the value of its outgoing edge $m_{23}$ (which has the second priority order) the same (and equal to $1$) because the new weight of the edge $m_{31}$ has not yet arrived.

\begin{figure} [ht]
\centering
\includegraphics[width=0.45\textwidth]{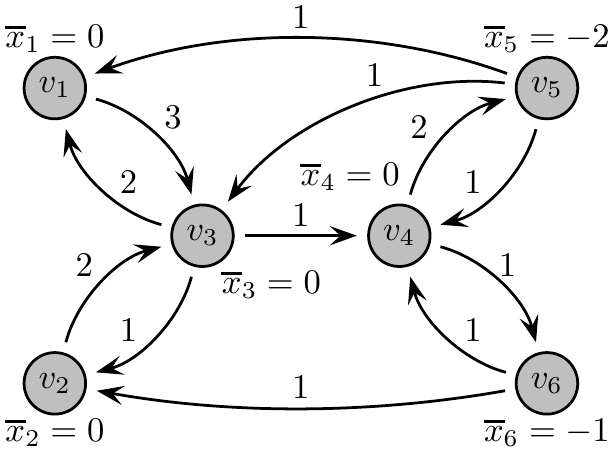}
\caption{Distribution of delayed imbalance from positively imbalanced nodes.}
\label{example-distr-delay-distributed}
\end{figure}

Note here that all the nodes in the digraph continue to send the same values on their outgoing edges at every iteration until they receive updated weights on their incoming edges. This means that the time delays $\tau_{31}[0]$, $\tau_{32}[0]$ and $\tau_{54}[0]$ are not necessarily the time-steps after which the nodes $v_3$ and $v_5$ will be informed for the new weights of their incoming edges. For example, if $\tau_{32}[0]=3$ then $v_3$ will receive the new weight $f_{32}[1]=2$ at iteration $k=3$; however, at iteration $k=1$ node $v_2$ re-sends its outgoing weights to its out-neighbors; thus, if $\tau_{32}[1]=1$ then $v_3$ will receive the new weight $f_{32}[1]=2$ at iteration $k=2$ (it will also receive it at $k=3$) and it will act accordingly (it will essentially ignore it).

After the integer weight update on the outgoing edges of each node with positive imbalance at $k=0$, the nodes check for updated incoming edge weights $f_{ji}[1]$, $\forall (v_j,v_i) \in \mathcal{E}$ (assuming that $\overline{f}_{ji}[1] = \overline{f}_{ji}[0] = f_{ji}[0]$ if no weight is received). Then they recalculate their imbalances $\overline{x}_j[1]$, $\forall v_j \in \mathcal{V}$, and the process is repeated. After a finite number of iterations, which we explicitly bound in the next sections, we reach the weighted digraph with integer weights shown in Figure~\ref{example-final-distributed}. As we will argue later on the thesis, this weighted digraph is the same irrespective of how time-delays manifest themselves.

\begin{figure} [ht]
\centering
\includegraphics[width=0.45\textwidth]{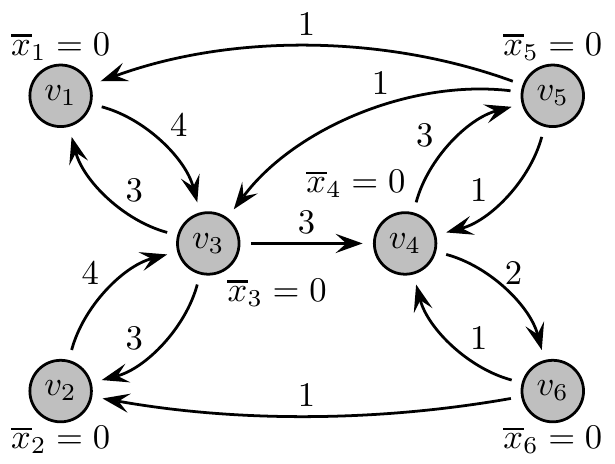}
\caption{Final weight balanced digraph.}
\label{example-final-distributed}
\end{figure}

\subsection{Execution Time Analysis of Distributed Algorithm}
\label{analysisalgDel}

In this section we analyze the functionality of the distributed algorithm and we prove that it solves the weight balancing problem in the presence of arbitrary (time-varying, inhomogeneous) but bounded time delays that may appear during the information exchange between agents in the system. 
Specifically, we prove that the proposed distributed algorithm results in a set of weights that form a weight balanced matrix after $O(n^6\overline{\tau})$ iterations, where $n$ is the number of nodes of the given digraph and $\overline{\tau}$ is the maximum delay in the digraph; also we show that the resulting weight balanced digraph is unique (irrespective of how delays manifest themselves) and identical to the one we obtain when transmissions between nodes happen instantaneously (no delays). 
We begin the analysis with the following theorem.

\textit{Setup:} Consider an arbitrary strongly connected digraph $\mathcal{G}_d = (\mathcal{V}, \mathcal{E})$, where $\mathcal{V}=(v_1,v_2,\dots,v_n)$ is the set of nodes, and $\mathcal{E} \subseteq \mathcal{V} \times \mathcal{V} - \{ (v_j,v_j)$ $|$ $v_j \in \mathcal{V} \}$ is the set of edges. Consider an execution of the proposed distributed algorithm where there are no delays $(\overline{\tau}=0)$ and denote the resulting weights on the edges as 
$$
f_{lj}^*[0]=1, f_{lj}^*[1],..., f_{lj}^*[k], ... \ \forall (v_l,v_j) \in \mathcal{E}.
$$
Consider another execution of the proposed distributed balancing algorithm where there are arbitrary but bounded delays $(0 < \overline{\tau} < \infty)$ and denote the resulting set of weights as
$$
Transmitted: f_{lj}[0]=1, f_{lj}[1],..., f_{lj}[k], ... \ \forall (v_l,v_j) \in \mathcal{E},
$$
$$
Received: \overline{f}_{lj}[0]=1, \overline{f}_{lj}[1],..., \overline{f}_{lj}[k], ... \ \forall (v_l,v_j) \in \mathcal{E}.
$$

\begin{theorem}\label{monotonicity_theorem}
Under the above \textit{setup}, we have for all $(v_l,v_j) \in \mathcal{E}$ and all $k \geq 0$
\begin{enumerate}
\item $f_{lj}^*[k+1] \geq f_{lj}^*[k]$,
\item $f_{lj}[k+1] \geq f_{lj}[k]$,
\item $\overline{f}_{lj}[k+1] \geq \overline{f}_{lj}[k]$.
\end{enumerate}
\end{theorem}

\begin{proof}
Consider the case when $\overline{\tau}=0$. At Step~$2$ of the proposed algorithm (at an arbitrary iteration $k$), if node $v_j$ has positive imbalance $x_j[k]>0$ then it increases the weights on its outgoing edges so that it becomes weight balanced (i.e., $f_{lj}^*[k+1] \geq f_{lj}^*[k]$, $\forall v_l \in \mathcal{N}_j^+$). If node $v_i$ has negative (or zero) imbalance $x_j[k] \leq 0$, it leaves the weights of its outgoing edges unchanged (i.e., $f_{lj}^*[k+1]=f_{lj}^*[k]$, $\forall v_l \in \mathcal{N}_j^+$). As a result, we have 
$$
f_{lj}^*[k+1] \geq f_{lj}^*[k], \forall (v_l,v_j) \in \mathcal{E}. 
$$

\noindent
Consider now the case when arbitrary but bounded time-delays ($\overline{\tau}>0$) affect link transmissions. Using a similar argument as above we easily establish that $f_{lj}[k+1] \geq f_{lj}[k], \forall (v_l,v_j) \in \mathcal{E}$. By the definition of $\overline{f}_{lj}[k+1]$ we have that $\overline{f}_{lj}[k+1] = f_{lj}[k_{last}]$ where $k_{last} = \max \{s \ | \ s+\tau_{lj}[s] \leq k+1 \} $. Similarly, $\overline{f}_{lj}[k] = f_{lj}[k'_{last}]$ where $k'_{last} = \max \{s \ | \ s+\tau_{lj}[s] \leq k \} $. Clearly, $k'_{last} \leq k_{last}$ and since $f_{lk}[k+1] \geq f_{lj}[k]$ and $k_{last} \geq k'_{last}$ we have that 
$$
\overline{f}_{lj}[k+1] \geq \overline{f}_{lj}[k], \forall v_l \in \mathcal{N}_j^+,
$$
which completes the proof.
\end{proof}

After establishing monotonicity for the weights of the outgoing edges for every node in the digraph, we continue with the following theorem.

\begin{theorem}
Under the above \textit{setup}, it holds that for every $k$, 
\begin{align} \label{convergence_equation}
\overline{f}_{lj}[k] \ \leq \ f_{lj}^*[k] \ \leq \ \overline{f}_{lj}[(k+1)(\overline{\tau} + 1)], \ \forall (v_l,v_j) \in \mathcal{E}. 
\end{align}

\end{theorem}

\begin{proof}
The proof is by induction. For $k=0$, we have at initialization $\overline{f}_{lj}[0] \ = \ f_{lj}^*[0] \ \leq \ \overline{f}_{lj}[\overline{\tau} + 1]$, where $\overline{f}_{lj}[0] \ = \ f_{lj}^*[0] \ = \ 1$, and (\ref{convergence_equation}) holds (since $f_{lj}[k]$ and $f_{lj}^*[k]$ are non decreasing and $\min\{f_{lj}[k]\} = 1$, for every $k$). 
Let us assume that for every $(v_l,v_j) \in \mathcal{E}$ we have 
$$
\overline{f}_{lj}[k] \ \leq \ f_{lj}^*[k] \ \leq \ \overline{f}_{lj}[(k+1)(\overline{\tau} + 1)],
$$
by the induction hypothesis.  
We would like to show that 
$$
\overline{f}_{lj}[k+1] \ \leq \ f_{lj}^*[k+1] \ \leq \ \overline{f}_{lj}[(k+2)(\overline{\tau} + 1)], \ \forall (v_l,v_j) \in \mathcal{E}.
$$
The fact that $\overline{f}_{lj}[k+1] \leq f_{lj}^*[k+1]$ is a consequence of Theorem~\ref{monotonicity_theorem}; we have that $\overline{f}_{lj}[k+1] = f_{lj}[k_{last}]$ where $k_{last}=\max\{s \ |\  s+\tau_{ji}[s]\leq k + 1\}$. Clearly, $k_{last} \leq k+1$ and (from Theorem~\ref{monotonicity_theorem}) $f_{lj}[k_{last}] \leq f_{lj}[k+1]$. As a result, we have that $\overline{f}_{lj}[k+1] \leq f_{lj}[k+1]$. 
To show that $\overline{f}_{lj}[k+1] \leq f_{lj}^*[k+1]$ we observe that 
\begin{enumerate}
\item[i)] $\overline{f}_{lj}[k] \leq f_{lj}^*[k]$ (by induction hypothesis),
\item[ii)] It follows from Lemma~\ref{weight_function} that $f_{lj}[k+1] = \mathcal{F}_{lj}(\sum_{v_i \in \mathcal{N}_j^-} \overline{f}_{ji}[k]) \leq  f_{lj}^*[k+1] = \mathcal{F}_{lj}(\sum_{v_i \in \mathcal{N}_j^-} f_{ji}^*[k])$, $\forall (v_l,v_j) \in \mathcal{E}$.
\end{enumerate}

\noindent
As a result, since $\overline{f}_{lj}[k+1] \leq f_{lj}[k+1]$ and $f_{lj}[k+1] \leq f_{lj}^*[k+1]$ we have that $\overline{f}_{lj}[k+1] \leq f_{lj}^*[k+1]$.

We are left with showing that $f_{lj}^*[k+1] \ \leq \ \overline{f}_{lj}[(k+2)(\overline{\tau} + 1)]$. 
From the induction hypothesis, we have that $f_{lj}^*[k] \ \leq \ \overline{f}_{lj}[(k+1)(\overline{\tau} + 1)]$. 
We observe that
\begin{enumerate}
\item[i)] $f_{lj}^*[k+1] = \mathcal{F}_{lj}(\sum_{v_i \in \mathcal{N}_j^-} f_{ji}^*[k]) \leq \mathcal{F}_{lj}(\sum_{v_i \in \mathcal{N}_j^-} \overline{f}_{ji}[(k+1)(\overline{\tau} + 1)]) = f_{lj}[(k+1)(\overline{\tau} + 1) + 1]$,
\item[ii)] $f_{lj}[(k+1)(\overline{\tau}+1) + 1] \leq \overline{f}_{lj}[(k+2)(\overline{\tau}+1)]$ (from Proposition~\ref{monotonic_prop}).
\end{enumerate}

\noindent
As a result, we have that $f_{lj}^*[k+1] \leq \overline{f}_{lj}[(k+2)(\overline{\tau} + 1)]$ and (\ref{convergence_equation}) holds. 
\end{proof}

Now, we can proceed with the final theorem where we establish that the proposed balancing algorithm converges to a set of weights that form a weight balanced digraph, which is unique and independent of the occurring delays. 

\begin{theorem}\label{theoremNoDrops}
Under the above \textit{setup}, the proposed balancing algorithm under no delays ($\overline{\tau}=0$) converges to a set of weights $f^*_{lj}$ that form a weight balanced digraph after a finite number of steps bounded by $O(n^6)$ while the proposed balancing algorithm in the presence of nonzero delays ($\overline{\tau}>0$) converges to a set of weights $f_{lj}=f_{lj}^*, \forall (v_l,v_j) \in \mathcal{E}$, after $O(n^6\overline{\tau})$ iterations.
\end{theorem}

\begin{proof}
As shown in \cite{2013:RikosHadj}, for the case where $\overline{\tau}=0$, the proposed distributed algorithm reaches a set of weights that forms a weight balanced digraph $F^*$ after a finite number of steps bounded by $O(n^6)$, where $n$ is the number of nodes in the digraph. 
This means that for every $(v_l,v_j) \in \mathcal{E}$, $\exists k_0 \in \mathbb{N}_0$ for which $f^*_{lj}[k+1] = f^*_{lj}[k], \forall k \geq k_0$. 
From (\ref{convergence_equation}) we have $\overline{f}_{lj}[(k_0+1)(\overline{\tau} + 1)] \geq f^*_{lj}[k_0]$ and $\overline{f}_{lj}[(k_0+1)(\overline{\tau} + 1)] \leq f^*_{lj}[(k_0+1)(\overline{\tau} + 1)]$; however, since $f^*_{lj}[k_0] = f^*_{lj}[(k_0+1)(\overline{\tau} + 1)]$ we have $\overline{f}_{lj}[(k_0+1)(\overline{\tau} + 1)] = f^*_{lj}[k_0]$, which means that the proposed algorithm reaches a set of weights $f_{lj} = \overline{f}_{lj} = f^*_{lj}, \forall (v_l,v_j) \in \mathcal{E}$ after $(k_0+1)(\overline{\tau} + 1)$ time steps. 
As a result, since the $\overline{\tau}=0$ case completes its operation after $O(n^6)$ steps (from \cite{2013:RikosHadj}), then the proposed distributed algorithm completes its operation after $O(n^6\overline{\tau})$ steps where $n$ and $\overline{\tau}$ are the number of nodes and the maximum delay in the given digraph, respectively. 
Furthermore, since $\overline{f}_{lj}[(k_0+1)(\overline{\tau} + 1)] = f^*_{lj}[k_0]$, $\forall (v_l,v_j) \in \mathcal{E}$, then the resulting edge weights are equal to the resulting edge weights of the case where no delays affect link transmissions.
\end{proof}

\section{Extension to Event-Triggered Operation}
\label{triggeredalgDel}

Motivated by the need to reduce energy consumption, communication bandwidth, network congestion, and/or processor usage, many researchers have considered the use of event-triggered communication and control \cite{dimarogonas2012distributed, 2014:nowzari_cortes}. In this section, we discuss an event-triggered operation of the proposed distributed algorithm  where each agent autonomously decides when communication and control updates should occur so that the resulting network executions still result in a weight balanced digraph after a finite number of steps in the presence of arbitrary (time-varying, inhomogeneous) but bounded time delays that might affect link transmissions. More specifically, following the proposed event-triggered strategy, we can prove that (i) all nodes eventually stop transmitting, and (ii) the proposed distributed algorithm is able to obtain a set of weights that balance the corresponding digraph after a finite number of iterations.

\subsection{Formal Description of Distributed Algorithm}

A formal description of the algorithm's event-triggered version is presented in Algorithm~\ref{algorithm:4}.

\begin{varalgorithm}{4}
\caption{Event-triggered distributed balancing in the presence of time delays}
\textbf{Input} \\ A strongly connected digraph $\mathcal{G}_d=(\mathcal{V},\mathcal{E})$ with $n=|\mathcal{V}|$ nodes and $m=|\mathcal{E}|$ edges.\\
\textbf{Initialization} \\ Set $k=0$; each node $v_j \in \mathcal{V}$ does the following:
\vspace{-0.1cm}
\begin{enumerate}
\item It assigns a unique priority order to its outgoing edges as $P_{lj}$, for $v_l \in \mathcal{N}_j^+$ (such that $\{P_{lj} \ | \ v_l \in \mathcal{N}_j^+\} = \{1,2,...,\mathcal{D}_j^+\}$). 
\item It sets its outgoing edge weights as
$$
f_{lj}[0] = \left\{ \begin{array}{ll}
         0, & \mbox{if $(v_l,v_j) \notin \mathcal{E}$,}\\
         1, & \mbox{if $(v_l,v_j) \in \mathcal{E}$.}\end{array} \right. 
$$
\item It sets its (perceived) incoming weights to be equal to unity, $\overline{f}_{ji}[0]=1$, $\forall v_i \in \mathcal{N}_j^-$.
\item It transmits the new weights $f_{lj}[0]$ on each outgoing edge $(v_l,v_j)\in \mathcal{E}$ to each $v_l \in \mathcal{N}_j^+$. 
\end{enumerate}

\textbf{Iteration} \\ For $k=0,1,2,\dots$, each node $v_j \in \mathcal{V}$ does the following:
\vspace{-0.1cm}
\begin{enumerate}
\item It receives the weights on its incoming edges $\overline{f}_{ji}[k+1]$. More specifically, for each in-neighboring node $v_i \in \mathcal{N}_j^-$ let $\mathcal{F}_{ji} = \{f_{ji}[s+\tau_{ji}[s]]$ $|$ $s+\tau_{ji}[s] = k+1\}$ be the set of weights of link $(v_j,v_i) \in \mathcal{E}$ that arrive at node $v_j$ at time step $k+1$. We have that
$$
\overline{f}_{ji}[k+1] = \left\{ \begin{array}{ll}
         \overline{f}_{ji}[k], & \mbox{if $\mathcal{F}_{ji} = \emptyset$,}\\
         \max\{\overline{f}_{ji}[k] , \max_{\overline{f}_{ji} \in \mathcal{F}_{ji}}\{\overline{f}_{ji}\}\}, & \mbox{if $\mathcal{F}_{ji} \neq \emptyset$.}\end{array} \right. 
$$
\item If $\overline{f}_{ji}[k+1] = \overline{f}_{ji}[k]$ for each $v_i \in \mathcal{N}_j^-$ then skip Steps~3,~4, and~5. Point out that this is the major difference from the previous algorithm.
\item It computes its weight imbalance according to the latest received weight values from its in-neighbors 
$$
\overline{x}_j[k + 1] = \overline{\mathcal{S}}_j^-[k + 1] - \mathcal{S}_j^+[k + 1]. 
$$
\item If  $\overline{x}_j[k + 1] = br_j^+ > 0$, it sets 
the values of the weights on its outgoing edges to $f_{lj}[k+1] = \left \lfloor \frac{\overline{S}_j^-[k + 1]}{\mathcal{D}_j^+} \right \rfloor$, $ \forall v_l \in \mathcal{N}_j^+$. 
Then, it chooses the set of the first (according to the priority order) $\overline{S}_j^-[k + 1] - \mathcal{D}_j^+ \left \lfloor \dfrac{\overline{S}_j^-[k + 1]}{\mathcal{D}_j^+} \right \rfloor$ outgoing edges, and increases their weight by 1 so that $\vert f_{lj}[k+1] - f_{hj}[k+1] \vert \leq 1, \forall v_l,v_h \in \mathcal{N}_j^+ $ and $\mathcal{S}_j^+[k+1] = \overline{\mathcal{S}}_j^-[k + 1]$.
\item It transmits the new weights $f_{lj}[k+1]$ on each outgoing edge $(v_l,v_j)\in \mathcal{E}$ to the corresponding out-neighbor $v_l \in \mathcal{N}_j^+$. 
\item It repeats (increases $k$ to $k+1$ and goes back to Step~1).
\end{enumerate}
\label{algorithm:4}
\end{varalgorithm}

\subsection{Execution Time Analysis of Distributed Algorithm}
\label{analysisalgDel}

\begin{prop}\label{eventconverge}
Under the above \textit{setup}, the proposed event-triggered balancing algorithm converges, in the presence of bounded delays ($\overline{\tau}>0$), to a set of weights $f_{lj} = f_{lj}^*, \ \forall (v_l,v_j) \in \mathcal{E}$, after a finite number of steps bounded by $O(n^6\overline{\tau})$ iterations (where the set of weights $f^*_{lj}$ is the set of weights obtained by the nominal algorithm that runs with no even-triggering and no delays).
\end{prop}

\begin{proof}

As shown in \cite{2013:RikosHadj}, when $\overline{\tau}=0$ the distributed algorithm in Section~\ref{chapteralgDel} reaches, after a finite number of steps bounded by $O(n^6)$ (where $n$ is the number of nodes in the digraph), a set of weights $f^*_{lj}$, $\forall (v_l, v_j) \in \mathcal{E}$, that forms a weight balanced digraph. This means that there exists $k_0 \in \mathbb{N}_0$, such that, for every $(v_l,v_j) \in \mathcal{E}$, we have $f^*_{lj}[k+1] = f^*_{lj}[k] = f^*_{lj}, \forall \ k \geq k_0$.

Consider now the event-triggered operation of the proposed distributed balancing algorithm in the presence of bounded delays in the communication links. The operation of the event-triggered version is identical to the operation of the proposed distributed algorithm with delays introduced in Section~\ref{chapteralgDel} if we assume that in the latter algorithm all transmissions of identical weights (that occur in the original version of the algorithm but not in the event-triggered version) suffer the maximum possible delay. As a result, since the operation of both algorithms is identical\footnote{The operation is identical under different delays in each case.}, we have that the event-triggered operation of the distributed algorithm will converge to a set of weights that form a weight balanced digraph after a finite number of steps bounded by $O(n^6\overline{\tau})$ iterations. Also, since $\exists \ k_0 \in \mathbb{N}_0$ for which $f_{lj}[k+1] = f_{lj}[k] = f^*_{lj}, \forall \ k \geq k_0$, from Step~3 of the algorithm, we can see that all nodes eventually stop transmitting (and the weights are identical to the weights obtained by the algorithm in Section~\ref{chapteralgDel}). 
\end{proof}

\begin{remark}
It is interesting to note here that event-triggering comes at the cost of speed in the sense that retransmissions of identical weights could have potentially allowed the receiving node to learn the weight change earlier (particularly if the delay after a triggering is large, in which case it could be offset by a smaller delay in a subsequent transmission).
\end{remark}

\begin{remark} 
It is important to note here that the proposed distributed balancing algorithm (along with its event triggered operation) is able to converge (with probability one) to a set of weights that form a balanced graph after a finite number of iterations in the case where there are different (possibly unbounded) delay distributions (except the uniform one which was thoroughly analyzed) in the communication links during the information exchange between nodes in the network.
\end{remark}

\section{Distributed Algorithm for Weight Balancing \\ in the Presence of Packet Dropping Links}
\label{packetalgDel}

Apart from bounded delays, \textit{unreliable} communication links in practical settings could also result in possible packet drops (i.e., unbounded delays) in the corresponding communication network. In this section, we analyze the performance of the proposed distributed weight balancing algorithm in the presence of possible packet drops in the communication links. To model packet drops, we assume that, at each time step $k$, a packet that is sent from node $v_i$ to node $v_j$ on link $(v_j, v_i) \in \mathcal{E}$ is dropped with probability $q_{ji}$, where we have $q_{ji} <1$. For simplicity, we assume independence between packet drops at different time steps or different links. We establish that, despite the presence of packet drops, the proposed distributed algorithm converges, with probability one, to a weight balanced digraph after a finite number of iterations. This weight balanced digraph is identical to the one obtained under no packet drops.

\begin{prop}
\label{dropsconverge}
Consider the above \textit{setup}, where the proposed balancing algorithm, with no packet drops and no delays, converges to a set of weights $f^*_{lj}$ that form a weight balanced digraph after a finite number of steps bounded by $O(n^6)$. In the presence of packet drops occurring with probability $q_{lj}$, $q_{lj} < 1$, $\ \forall (v_l,v_j) \in \mathcal{E}$ (independently between different links and different time steps), the proposed balancing algorithm also converges, with probability one, to a set of weights $f_{lj} = f_{lj}^*, \ \forall (v_l,v_j) \in \mathcal{E}$, after a finite number of iterations. 
\end{prop}

\begin{proof}
As mentioned earlier and shown in \cite{2013:RikosHadj} for the case where $\overline{\tau}=0$, the proposed distributed algorithm reaches a set of weights $F^*$ that forms a weight balanced digraph after a finite number of steps bounded by $O(n^6)$, where $n$ is the number of nodes in the digraph. 
This means that for every $(v_l,v_j) \in \mathcal{E}$, $\exists \ k_0 \in \mathbb{N}_0$ for which $f^*_{lj}[k+1] = f^*_{lj}[k], \forall \ k \geq k_0$.

Consider now an execution of the proposed distributed balancing algorithm where packets containing information are dropped with probability $q_{lj} <1$ for each communication link $(v_l,v_j) \in \mathcal{E}$, and assume independence between packet drops at different time steps and different links.

During transmissions on link $(v_l, v_j)$, we have that at each transmission, a packet goes through with probability $1 - q_{lj} > 0$. 
Thus, if we consider $k_{lj}$ consecutive uses of link $(v_l, v_j)$, the probability that at least one packet will go through is $1 - q_{lj}^{k_{lj}}$, which will be arbitrarily close to $1$ for a sufficiently large $k_{lj}$.

\noindent
Specifically, for any (arbitrarily small) $\epsilon > 0$, we can choose 
$$
k_{lj} = \left \lceil \frac{\log \epsilon}{\log q_{lj}} \right \rceil, 
$$
to ensure that each transmission goes through at least once within $k_{lj}$ steps with probability $1 - \epsilon$. 

Let $\overline{\tau} = \max_{(v_l, v_j) \in \mathcal{E}}\{k_{lj}\}$; then since the proposed distributed  algorithm completes under no packet drops in $O(n^6)$ steps, we can conclude that it will complete by $O(n^6\overline{\tau})$ steps with probability $(1 - \epsilon)^{n^6 | \mathcal{E} |}$ in the presence of packet drops (note that $| \mathcal{E} |$ is the number of edges in the given digraph). By making $\epsilon$ arbitrarily small we can make this probability arbitrarily close to $1$. Moreover, since this particular execution of the algorithm (that occurs with probability $(1 - \epsilon)^{n^6 | \mathcal{E} |}$) is essentially identical to an execution of the algorithm in Section~\ref{chapteralgDel} with delays that are bounded by $\overline{\tau}$, the final weights are identical to the weights of that algorithm (i.e., for large enough $k$ we have $f_{lj}[k] = f^*_{lj}$ for all $(v_l, v_j) \in \mathcal{E}$).
%
%
%
\end{proof}

\begin{remark} 
Note here that the presence of packet drops can be dealt in a way identical to the presence of unbounded delays in the communication links during the information exchange between nodes in the network. In fact, when the delays in the communication links are unbounded, then the proposed distributed algorithm is still able to obtain a set of weights that balance the corresponding digraph after a finite number of iterations.
\end{remark}

\begin{remark} 
It is worth pointing out that the proposed distributed algorithm is able to converge (with probability one) to a set of weights that form a balanced graph after a finite number of iterations when there are both possible packet drops {\em and} arbitrary but bounded time delays in the communication links, while the resulting weight balanced digraph is again unique and independent of how packet drops and delays manifest themselves in link transmissions.
\end{remark}

\subsection{Simulation Study}
\label{resultsalgDel}

In this section, we present simulation results and comparisons for the proposed distributed algorithm. 
Specifically, we present detailed numerical results for a random graph of size $n = 20$ and for the average of $1000$ random digraphs of $20$ and $50$ nodes each.
We illustrate the behavior of the proposed distributed algorithm for the following three different scenarios: 
(i) the scenario where there are no packet drops in the communication links $(v_j,v_i) \in \mathcal{E}$ and each node $v_j$ transmits the weights $f_{lj}[k+1]$ of each outgoing edge $(v_l,v_j)\in \mathcal{E}$ to each $v_l \in \mathcal{N}_j^+$ at each iteration $k$,
(ii) the scenario where there are packet drops with equal probability $q$ (where $0 \leq q < 1$) for every communication link $(v_j,v_i) \in \mathcal{E}$ and each node $v_j$ transmits the weights $f_{lj}[k+1]$ of each outgoing edge $(v_l,v_j)\in \mathcal{E}$ to each $v_l \in \mathcal{N}_j^+$ at each iteration $k$,
(iii) the scenario where there are no packet drops at the communication links $(v_j,v_i) \in \mathcal{E}$ and each node $v_j$ transmits only once the updated weights $f_{lj}[k+1]$ of each outgoing edge $(v_l,v_j)\in \mathcal{E}$ to each $v_l \in \mathcal{N}_j^+$. 
Each scenario of the proposed distributed algorithm is analyzed in a) the absence of time-delays in the communication links (i.e., $\tau_{lj}[k]=0$, $\forall \ (v_l,v_j) \in \mathcal{E}$) and b) the presence of time-delays in the communication links (i.e., $0 \leq \tau_{lj}[k] \leq \tau_{lj}$, $\forall \ (v_l,v_j) \in \mathcal{E}$).

Note here that in the case where $\tau_{lj}[k]=0$, $\forall \ (v_l,v_j) \in \mathcal{E}$ we have that $\overline{f}_{lj}[k] = f_{lj}[k]$ and the proposed distributed algorithm is identical to the algorithm presented in \cite{2013:RikosHadj}, where we consider the transmission between nodes to happen instantaneously.

Figure~\ref{journalRand20} shows what happens in the case of a randomly created graph of $20$ nodes, in which the operation of the proposed distributed algorithm includes no packet drops at the communication links $(v_j,v_i) \in \mathcal{E}$ and each node $v_j$ transmits the weights $f_{lj}[k+1]$ of each outgoing edge $(v_l,v_j)\in \mathcal{E}$ to each $v_l \in \mathcal{N}_j^+$ at each iteration $k$. 
We plot the \textit{total delayed imbalance} as a function of the number of iterations $k$ for the cases where $\overline{\tau}=0$ (solid line), $\overline{\tau} = 10$ (dashed line) and $\tau_{lj} = \overline{\tau} = 10$, $\forall (v_l,v_j) \in \mathcal{E}$ (dashed-dotted line). The plot demonstrates that the proposed distributed algorithm is able to obtain a set of weights that balance the corresponding digraph after a finite number of iterations as argued in the previous section.

\begin{figure}[ht] 
\centering    
\includegraphics[width=0.60\textwidth]{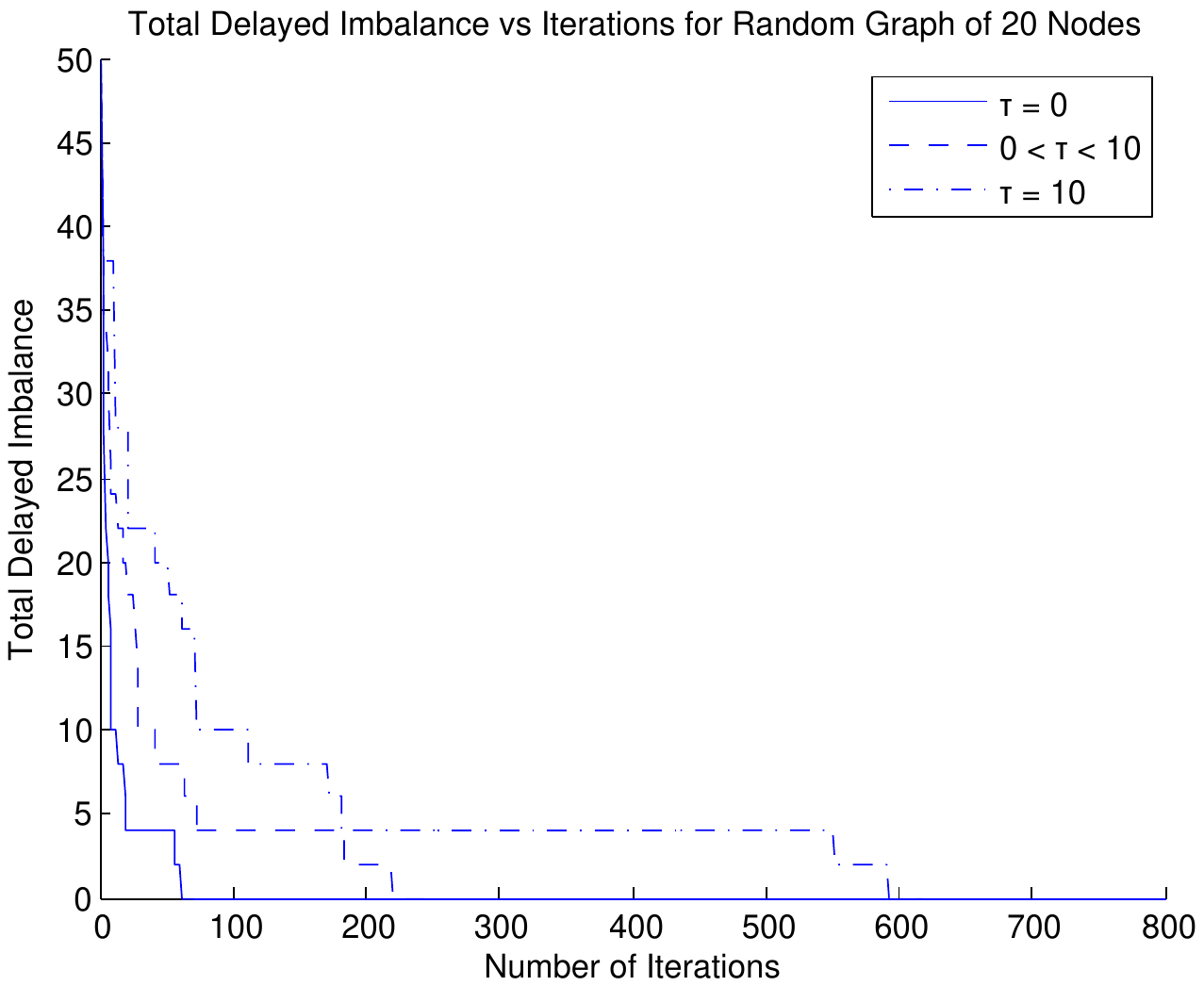}
\caption{Total delayed imbalance plotted against the number of iterations for a random digraph of $20$ nodes in the case where $\overline{\tau} = 0$ (solid line), $0 < \tau_{lj} < \overline{\tau}$ where $\overline{\tau} = 10$ (dashed line) and in the case where $\tau_{lj} = \overline{\tau} = 10$ (dashed-dotted line).}
\label{journalRand20}
\end{figure}

Figure~\ref{journalRand20Drop} shows the same case as Figure~\ref{journalRand20}, with the difference that 
the operation of the proposed distributed algorithm includes packet drops with equal probability $q$ (where $0 < q < 1$) for every communication link $(v_j,v_i) \in \mathcal{E}$ and each node $v_j$ transmits the weights $f_{lj}[k+1]$ of each outgoing edge $(v_l,v_j)\in \mathcal{E}$ to each $v_l \in \mathcal{N}_j^+$ at each iteration $k$. Here, the plot suggests, that the proposed distributed algorithm is able to obtain a set of weights that balance the corresponding digraph after a finite number of iterations.

\begin{figure}[ht] 
\centering    
\includegraphics[width=0.60\textwidth]{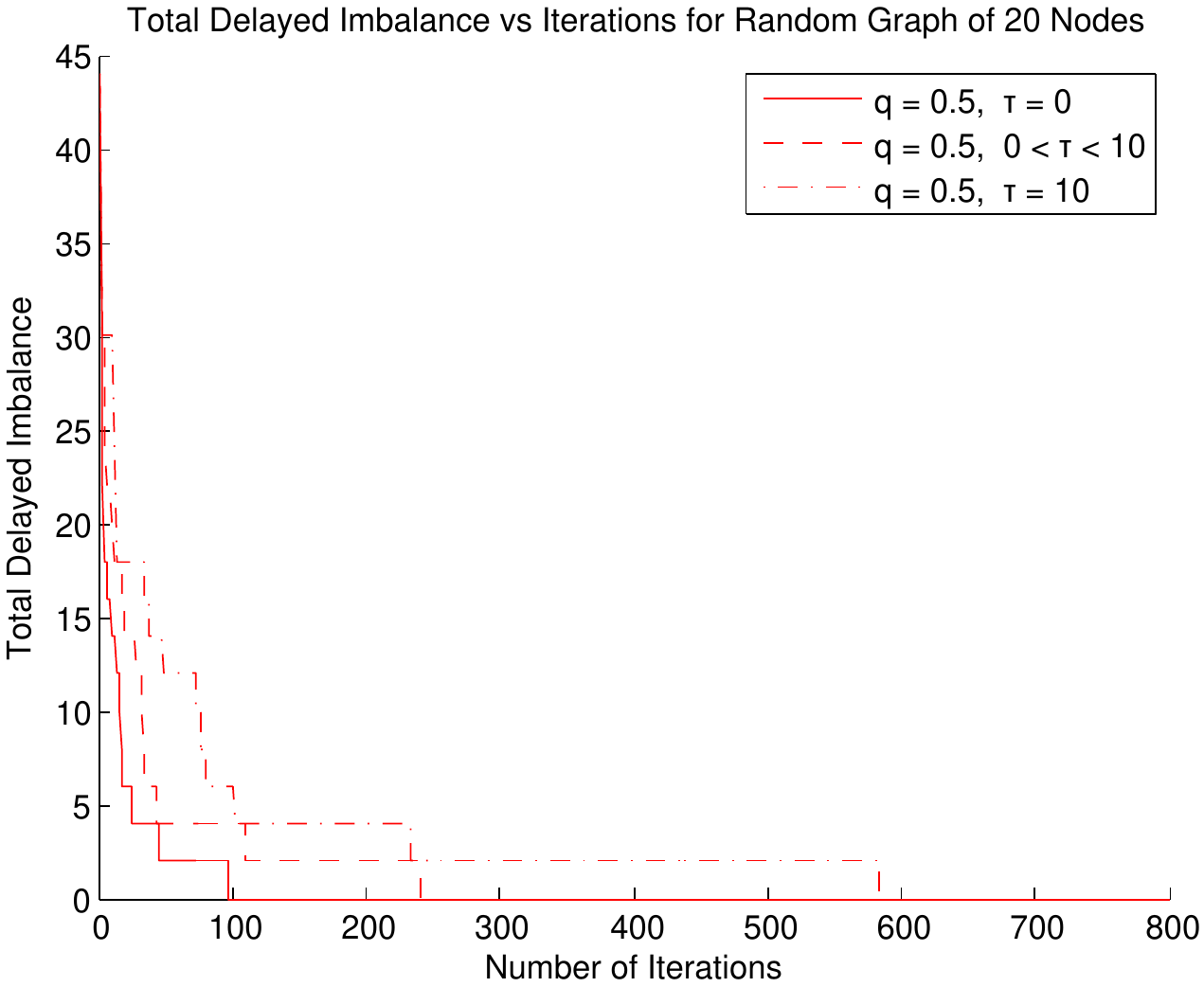}
\caption{Total delayed imbalance plotted against the number of iterations for a random digraph of $20$ nodes in the case where $\overline{\tau} = 0$ (solid line), $0 < \tau_{lj} < \overline{\tau}$ where $\overline{\tau} = 10$ (dashed line) and in the case where $\tau_{lj} = \overline{\tau} = 10$ (dashed-dotted line).}
\label{journalRand20Drop}
\end{figure}

Figure~\ref{journalRand20Once} shows the same case as Figures~\ref{journalRand20} and \ref{journalRand20Drop}, with the difference being that the operation of the proposed distributed algorithm includes no packet drops at the communication links $(v_j,v_i) \in \mathcal{E}$ but each node $v_j$ transmits only once the updated weights $f_{lj}[k+1]$ of each outgoing edge $(v_l,v_j)\in \mathcal{E}$ to each $v_l \in \mathcal{N}_j^+$. Here, the plot demonstrates that the proposed distributed algorithm is able to obtain a set of weights that balance the corresponding digraph after a finite number of iterations.

\begin{figure}[ht] 
\centering    
\includegraphics[width=0.60\textwidth]{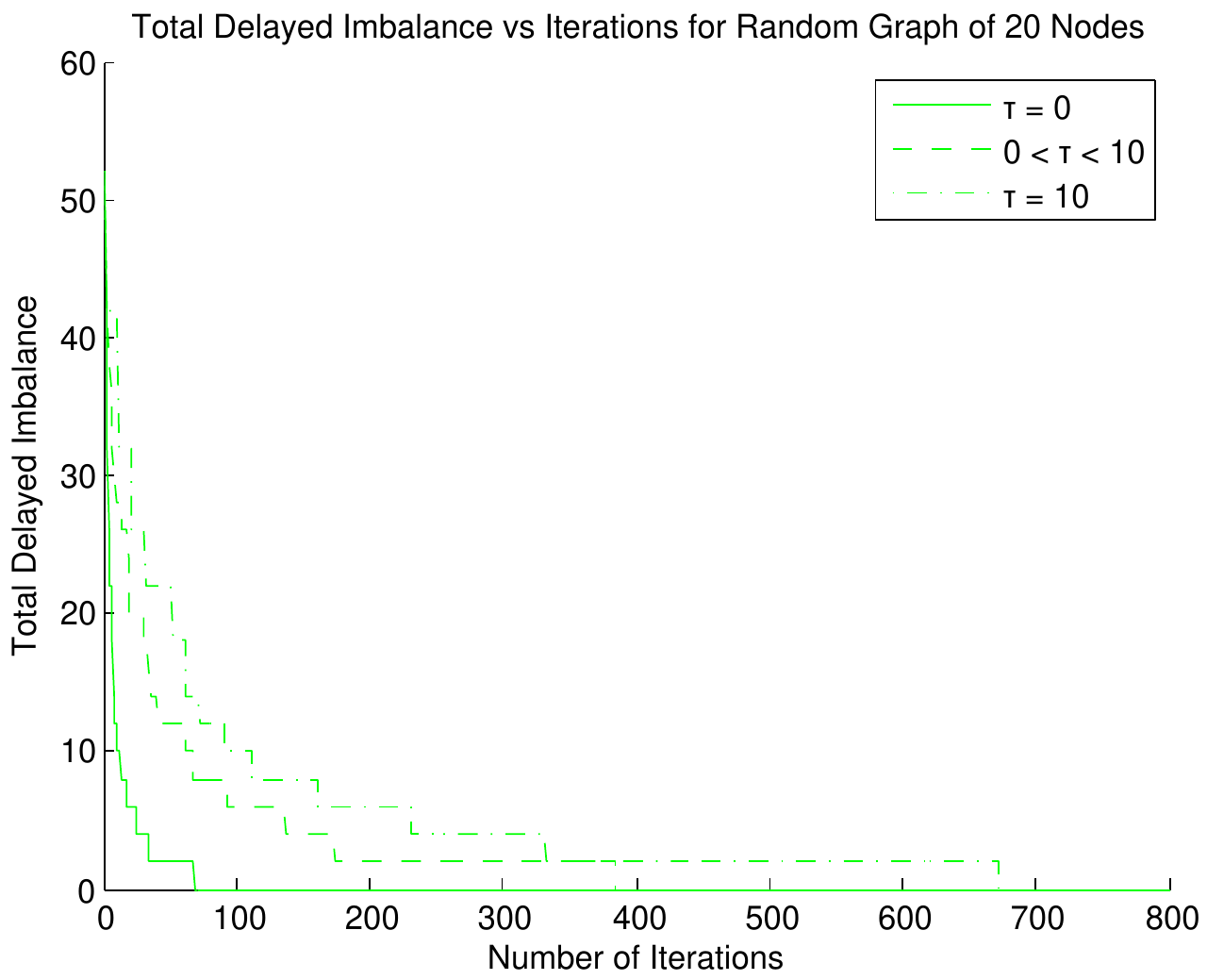}
\caption{Total delayed imbalance plotted against the number of iterations for a random digraph of $20$ nodes in the case where $\overline{\tau} = 0$ (solid lines), $0 < \tau_{lj} < \overline{\tau}$ where $\overline{\tau} = 10$ (dashed lines) and in the case where $\tau_{lj} = \overline{\tau} = 10$ (dashed-dotted lines).}
\label{journalRand20Once}
\end{figure}

Figure~\ref{journalAver20} shows what happens for the average of $1000$ random digraphs of $20$ nodes each for the three scenarios of the proposed distributed algorithm presented in Figures~\ref{journalRand20}, \ref{journalRand20Drop} and \ref{journalRand20Once} respectively.
Note that the plot colors of the three scenarios in Figures~\ref{journalRand20}, \ref{journalRand20Drop} and \ref{journalRand20Once}, remain the same in Figure~\ref{journalAver20} (i.e., the scenarios of Figures~\ref{journalRand20}, \ref{journalRand20Drop} and \ref{journalRand20Once} are shown with blue, red and green, respectively, in Figure~\ref{journalAver20}). 
 We plot the average total delayed imbalance as a function of the number of iterations $k$ in logarithmic scale for $\overline{\tau}=0$ (solid lines), $\overline{\tau} = 10$ (dashed lines) and $\tau_{lj} = \overline{\tau} = 10$, $\forall (v_l,v_j) \in \mathcal{E}$ (dashed-dotted lines). 
Here we can see that the first scenario of the proposed distributed algorithm (presented in Figure~\ref{journalRand20}) is identical to the third one (presented in Figure~\ref{journalRand20Once}) for the case where there are no time-delays in the communication links (i.e., $\tau_{lj}[k]=0$, $\forall \ (v_l,v_j) \in \mathcal{E}$). However, the first scenario (Figure~\ref{journalRand20}) generally outperforms the second and third scenarios (Figures~\ref{journalRand20Drop} and \ref{journalRand20Once}, respectively) for the case where there are time-delays in the communication links (i.e., $0 \leq \tau_{lj}[k] \leq \tau_{lj}$, $\forall \ (v_l,v_j) \in \mathcal{E}$).

\begin{figure}[ht] 
\centering    
\includegraphics[width=0.60\textwidth]{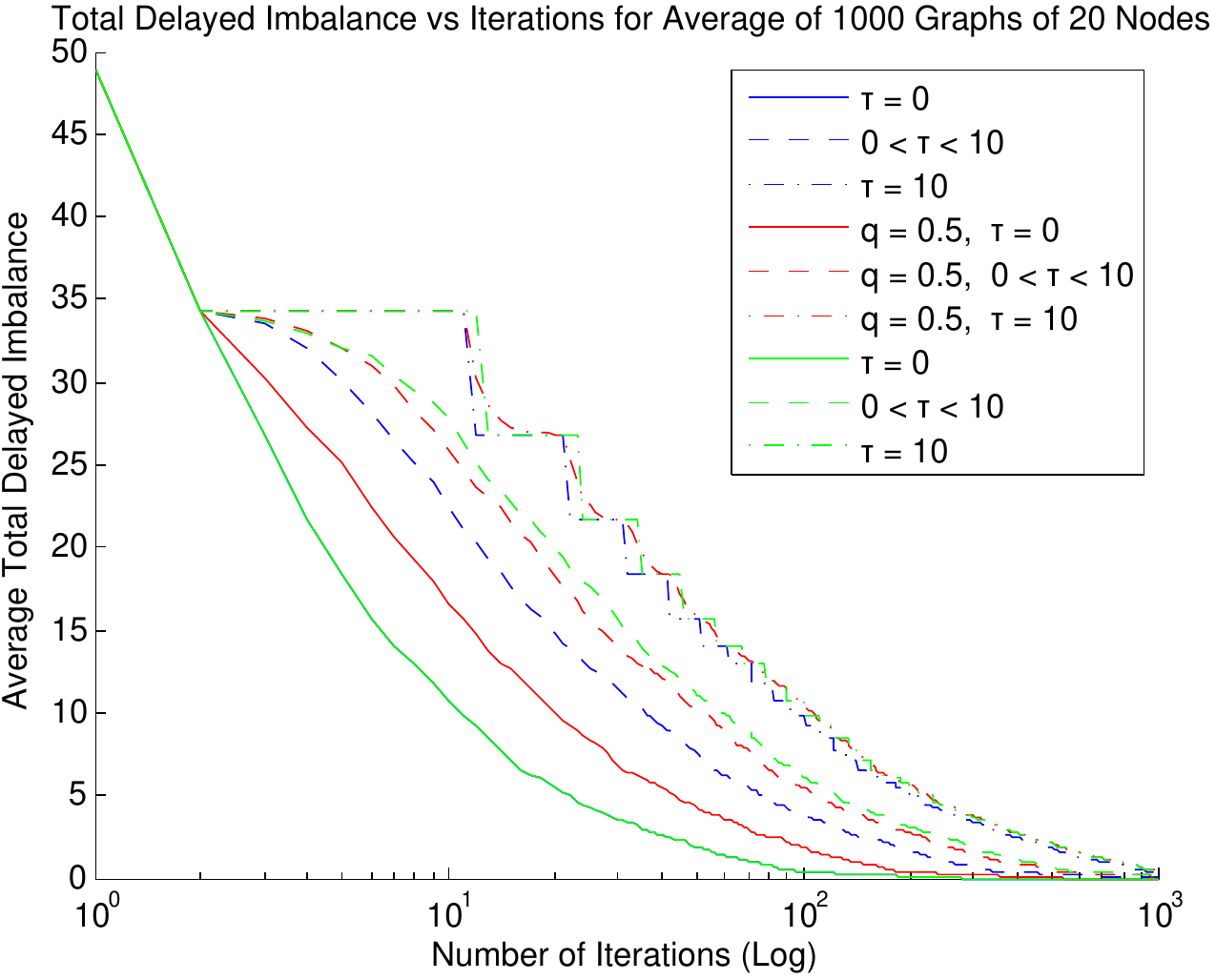}
\caption{Average total delayed imbalance plotted against the number of iterations in logarithmic scale for $1000$ random digraphs of $20$ nodes each in the case where $\overline{\tau} = 0$ (solid lines), $0 < \tau_{lj} < \overline{\tau}$ where $\overline{\tau} = 10$ (dashed lines) and in the case where $\tau_{lj} = \overline{\tau} = 10$ (dashed-dotted lines).}
\label{journalAver20}
\end{figure}

Figure~\ref{journalAver50} shows the same cases as Figure~\ref{journalAver20}, with the difference being that the network consists of $50$ nodes. The performance of the proposed distributed algorithm does not change due to the network size and the conclusions are the same as in Figure~\ref{journalAver20}.

\begin{figure}[ht] 
\centering    
\includegraphics[width=0.60\textwidth]{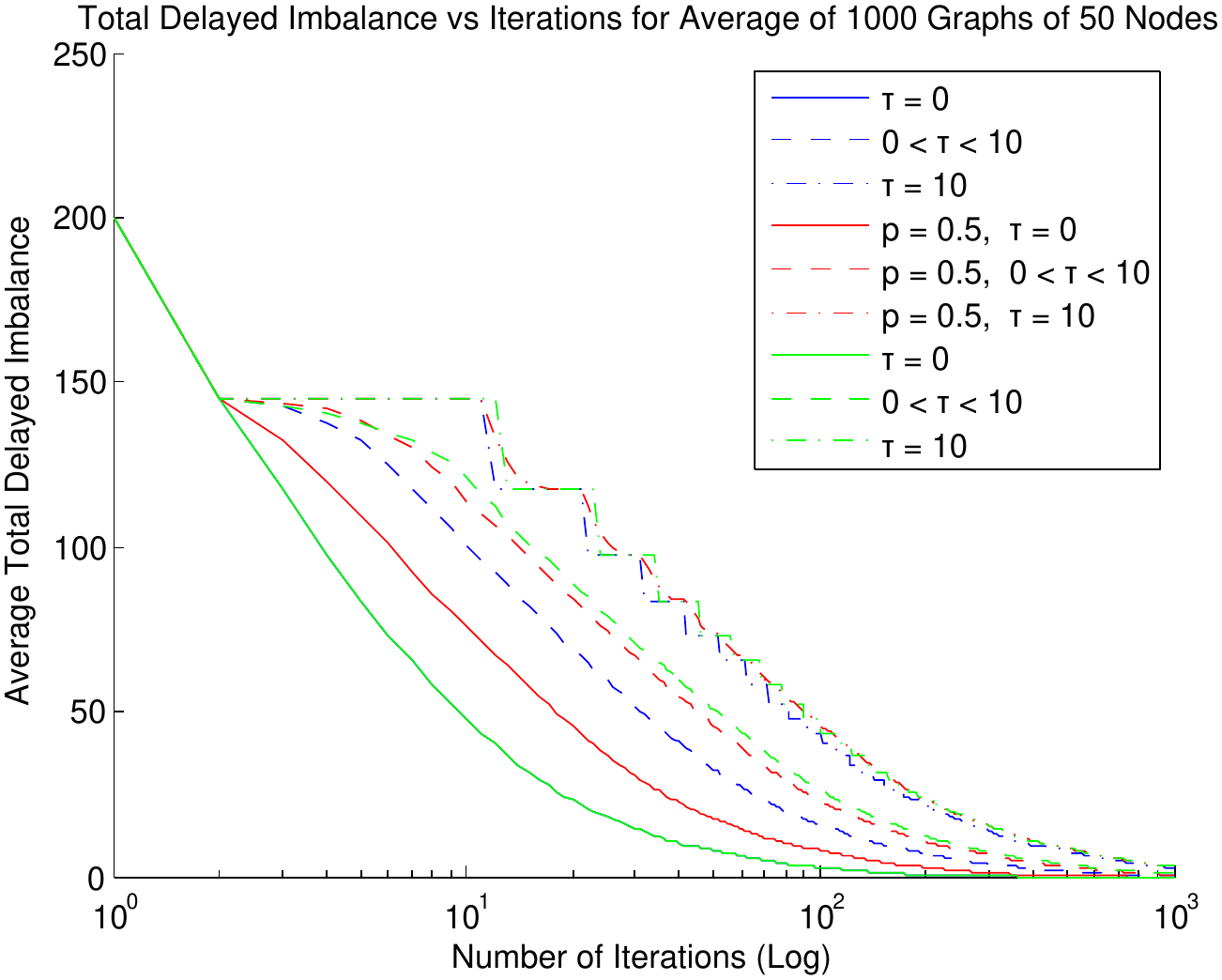}
\caption{Average total delayed imbalance plotted against the number of iterations in logarithmic scale for $1000$ random digraphs of $50$ nodes each in the case where $\overline{\tau} = 0$ (solid lines), $0 < \tau_{lj} < \overline{\tau}$ where $\overline{\tau} = 10$ (dashed lines) and in the case where $\tau_{lj} = \overline{\tau} = 10$ (dashed-dotted lines).}
\label{journalAver50}
\end{figure}

\section{Chapter Summary}
\label{summaryalgDel}

In this chapter, we introduced a novel distributed iterative algorithm and established that it converges to a weight balanced digraph after a finite number of steps. 
We have also bounded the execution time of the proposed algorithm as $O(n^6\overline{\tau})$, where $n$ is the number of nodes and $\overline{\tau}$ is the maximum delay in the digraph, and argued that the resulting weight balanced digraph is unique and independent of how the delays that affect link transmissions materialize.  
We also added extensions to handle the cases of packet drops over the communication links and event-triggered operation. In both scenarios, the proposed algorithm converges (with probability one) to a set of weights that form a balanced graph after a finite number of iterations while the resulting weight balanced digraph is unique and independent of how packet drops affect link transmissions.
We also demonstrated the operation, performance, and advantages of the proposed algorithm via various simulations.

\clearpage

\lhead{\emph{Weight Balancing under Link Capacity Constraints}}

\chapter{Weight Balancing under \\ Link Capacity Constraints}
\label{constraintsbalancing}

In this chapter, we present a novel distributed algorithm which deals with the problem of balancing a weighted digraph in the presence of link capacity constraints over the communication links. 
The algorithm presented in this chapter has appeared in \cite{2016:RikHadj, 2018:RikosHadj}.

This chapter is organized as follows. In Section~\ref{ProbStatementUpperLower} we present additional notation and the problem formulation.
In Section~\ref{CircConditions} we present the conditions for the existence of a set of integer weights (within the allowable intervals) that balance a weighted digraph. 
In Section~\ref{upperloweralgorithm} we introduce a novel distributed algorithm which achieves integer weight balancing in a multi-component system, in the presence of specified lower and upper constraints on the edge weights. 
We present a formal description of the proposed distributed algorithm and demonstrate its performance via an illustrative example. 
Then we show that, as long as the conditions presented in Section~\ref{CircConditions} hold, the proposed distributed algorithm converges to a weight-balanced digraph after a finite number of iterations and we conclude by we presenting simulation results and comparisons. 

\section{Graph-Theoretic Notions and Problem Statement}
\label{ProbStatementUpperLower}

In this chapter, we assume that a pair of nodes $v_j$ and $v_i$ that are connected by an edge in the digraph $\mathcal{G}_d$ (i.e., $(v_j, v_i) \in \mathcal{E}$ and/or $(v_i, v_j) \in \mathcal{E}$) can exchange information among themselves (in both directions). In other words, the {\em communication topology} is captured by the undirected graph $\mathcal{G}_u = (\mathcal{V}, \mathcal{E}_u)$ that corresponds to a given directed graph $\mathcal{G}_d = (\mathcal{V}, \mathcal{E})$, where
$$
\displaystyle
\mathcal{E}_u = \cup_{(v_j, v_i) \in \mathcal{E}} \{ (v_j, v_i), (v_i, v_j) \} = \mathcal{E} \cup \mathcal{E}_r \; ,
$$
with $\mathcal{E}_r = \{ (v_i, v_j) \; | \; (v_j, v_i) \in \mathcal{E} \}$. [Recall that a graph is undirected if and only if $(v_j, v_i) \in \mathcal{E}$ implies $(v_i, v_j) \in \mathcal{E}$.]

Also, we assume that node $v_j$ assigns a unique order in the set $\{0,1,..., \mathcal{D}_j-1\}$ to each of its outgoing and incoming edges. 
The order of link $(v_l,v_j)$ (or $(v_j,v_i)$) is denoted by $P^{(j)}_{lj}$ (or $P^{(j)}_{ji}$) (such that $\{P^{(j)}_{lj} \; | \; v_l \in \mathcal{N}^+_j\} \cup \{P^{(j)}_{ji} \; | \; v_i \in \mathcal{N}^-_j\} = \{0,1,..., \mathcal{D}_j-1\}$) and will be used later on as a way of allowing node $v_j$ to make changes to its outgoing and incoming edge weights in a unique predetermined order. 
This unique order is used during the execution of the proposed distributed algorithm as a way of allowing node $v_j$ to transmit messages to its out- and in-neighbors in a \textit{round-robin}\footnote{Each node $v_j$ transmits to its out- and in-neighbors by following a unique predetermined order. The next time it needs to transmit to an out- or in-neighbor, it will continue from the outgoing or incoming edge it stopped the previous time and cycle through the edges in a round-robin fashion according to the unique predetermined ordering.} fashion.

Given a digraph $\mathcal{G}_d = (\mathcal{V}, \mathcal{E})$ we can associate nonnegative integer weights $f_{ji} \in \mathbb{N}_0$ on each edge $(v_j, v_i) \in \mathcal{E}$. 
In this thesis, these weights will be restricted to have positive integer values and lie in an interval $[l_{ji}, u_{ji}]$, i.e., $0 < l_{ji} \leq f_{ji} \leq u_{ji}$ and $f_{ji} \in \mathbb{N}$, for every $(v_j, v_i) \in \mathcal{E}$. 
We will also use matrix notation to denote (respectively) the integer weight, lower limit, and upper limit matrices by the $n \times n$ matrices $F = [ f_{ji} ]$, $L=[l_{ji}]$, and $U=[u_{ji}]$, where $F(j,i)=f_{ji}$, $L(j,i)=l_{ji}$, $U(j,i)=u_{ji}$, and $f_{ji} \in \mathbb{N}$, for every $(v_j, v_i) \in \mathcal{E}$ (obviously $f_{ji} = l_{ji} = u_{ji} = 0$ when $(v_j,v_i) \notin \mathcal{E}$).

\subsection{Problem Statement}

We are given a strongly connected digraph $\mathcal{G}_d = (\mathcal{V}, \mathcal{E})$, as well as lower and upper bounds $l_{ji}$ and $u_{ji}$ ($l_{ji} \leq u_{ji}$, $l_{ji}, u_{ji} \in \mathbb{R}$) on each each edge $(v_j, v_i) \in \mathcal{E}$. We want to develop a distributed algorithm that allows the nodes to iteratively adjust the weights on their edges so that they eventually obtain a set of integer weights $\{ f_{ji} \; | \; (v_j, v_i) \in \mathcal{E} \}$ that satisfy the following:
\begin{enumerate}
\item $ f_{ji} \in \mathbb{N}_0$ for each edge $(v_j,v_i) \in \mathcal{E}$.
\item $l_{ji} \leq f_{ji} \leq u_{ji}$ for each edge $(v_j,v_i) \in \mathcal{E}$;
\item $\mathcal{S}_j^+ = \mathcal{S}_j^-$ for each $v_j \in \mathcal{V}$;
\end{enumerate}
The distributed algorithm needs to respect the communication constraints imposed by the undirected graph $\mathcal{G}_u$ that corresponds to the given directed graph~$\mathcal{G}_d$.

\begin{remark}
One of the main differences of this chapter with chapters~\ref{centrvsdistr} and \ref{distrdelpacket} is that the algorithms presented in this chapter require a bidirectional communication topology, whereas the aforementioned algorithms assume a communication topology that matches the weight (physical) topology. 
We should point out that direct application of these earlier algorithms to the problem that is of interest in this chapter will generally fail (because weights are restricted to lie within lower and upper limits). 
Also, note that there are many applications where the physical topology is directed but the communication topology is bidirectional (e.g., traffic weight in an one way street is directional, but communication between traffic lights at the end points of the street will, in fact, be bidirectional). 
In such cases, the algorithms proposed here are directly applicable. 
More generally, in many applications, the communication topology does not necessarily match the physical one.
\end{remark}

\section{Integer Circulation Conditions}\label{CircConditions}

Given a strongly connected digraph $\mathcal{G}_d=(\mathcal{V}, \mathcal{E})$, with lower and upper bounds $l_{ji}$ and $u_{ji}$ ($0 < l_{ji} \leq u_{ji}$) on each edge $(v_j, v_i) \in \mathcal{E}$, the necessary and sufficient conditions for the existence of a set of \textit{integer} weights $\{ f_{ji} \; | \; (v_j, v_i) \in \mathcal{E} \}$ that satisfy the {\em capacity constraints} (i.e., $l_{ji} \leq f_{ji} \leq u_{ji}$ for each edge $(v_j,v_i) \in \mathcal{E}$), and {\em balance constraints} (i.e., $\mathcal{S}_j^+ = \mathcal{S}_j^-$ for every $v_j \in \mathcal{V}$), can be stated (by adopting Theorem~$3.1$ in \cite{2010:Fulkerson}) as follows:
\begin{enumerate}
\item[(i)] for every $(v_j,v_i) \in \mathcal{E}$, we have
\begin{equation}\label{ConstCondition1}
\lceil l_{ji} \rceil \leq \lfloor u_{ji} \rfloor,
\end{equation}
and
\item[(ii)] for each subset of nodes $\mathcal{S}$, $\mathcal{S} \subset \mathcal{V}$, we have
\begin{equation}
\label{EQnsconditions}
\sum_{(v_j, v_i) \in \mathcal{E}^-_\mathcal{S}} \lceil l_{ji} \rceil \leq \sum_{(v_l, v_j) \in \mathcal{E}^+_\mathcal{S}} \lfloor u_{lj} \rfloor \; ,
\end{equation}
\end{enumerate}
where 
\begin{eqnarray}
\mathcal{E}^-_\mathcal{S} & = & \{ (v_j, v_i) \in \mathcal{E} \; | \; v_j \in \mathcal{S}, \; v_i \in \mathcal{V}-\mathcal{S} \} \; , \label{REALEQinS}\\
\mathcal{E}^+_\mathcal{S} & = & \{ (v_l, v_j) \in \mathcal{E} \; | \; v_j \in \mathcal{S}, \; v_l \in \mathcal{V}-\mathcal{S} \} \label{REALEQoutS}\; .
\end{eqnarray}

\section{Distributed Algorithm for Weight Balancing \\ under Link Capacity Constraints}
\label{upperloweralgorithm}

In this section we provide an overview of the operation of a distributed balancing algorithm (Algorithm~\ref{constralg1}) and discuss a possible enhancement.
The algorithm is iterative and operates by having, at each iteration, nodes with \textit{positive} weight imbalance attempt to change the integer weights on both their incoming and outgoing edges so that they become weight balanced. 
At each iteration $k$, each node $v_j$ compares the total \textit{in-weight} from the weights of its incoming edges against the total \textit{out-weight} from the weights of its outgoing edges. 
If its weight imbalance is positive ($x_j[k]>0$) then it increases (decreases) the weights of its outgoing (incoming) edges according to the order chosen at initialization. 
Finally, it transmits the amount of change it desires for each incoming/outgoing edge to the corresponding in/out neighbor; the node also receives the amount of change its neighbors desire for the corresponding edges, based on which it calculates the new edge weights; the above procedure is repeated at each iteration.
We establish that, if the necessary and sufficient \textit{integer circulation conditions} for the existence of a set of \textit{integer} weights that balance the given digraph are satisfied, the algorithm completes after a finite number of iterations. We next describe the iterative algorithm in more detail.

\textbf{Initialization.} At initialization, each node is aware of the feasible weight interval on each of its incoming and outgoing edges, i.e., node $v_j$ is aware of $l_{ji}, u_{ji}$ for each $v_i \in \mathcal{N}^-_j$ and $l_{lj}, u_{lj}$ for each $v_l \in \mathcal{N}^+_j$. 
Furthermore, the weights are initialized at the ceiling of the lower bound of the feasible interval, i.e., $f_{ji}[0] = \lceil l_{ji} \rceil$.
This initialization is always feasible but not critical and could be any integer value in the feasible weight interval $[l_{ji}, u_{ji}]$ (according to condition (\ref{ConstCondition1}) an integer exists in the interval $[l_{ji}, u_{ji}]$). 
Also each node $v_j$ chooses a unique order $P_{lj}^{(j)}$ and $P_{ji}^{(j)}$ for its outgoing links $(v_l,v_j)$ and incoming links $(v_j,v_i)$ respectively, such that $\{P_{lj}^{(j)} \; | \; v_l \in \mathcal{N}^+_j\} \cup \{P_{ji}^{(j)} \; | \; v_i \in \mathcal{N}^-_j\} = \{0,1,..., \mathcal{D}_j-1\}$.

\textbf{Iteration.} At each iteration $k \geq 0$, node $v_j$ is aware of the integer weights on its incoming edges $\{ f_{ji}[k] \; | \: v_i \in \mathcal{N}^-_j \}$ and outgoing edges $\{ f_{lj}[k] \; | \: v_l \in \mathcal{N}^+_j \}$, and calculates its weight imbalance $x_j[k]$ according to Definition~\ref{DEFnodebalance}.

\noindent
{\em A. Selecting desirable weights:} Each node $v_j$ with positive weight imbalance (i.e., with $x_j[k] > 0$) attempts to change the integer weights in both its incoming edges and its outgoing edges. 
No attempt to change weights is made if node $v_j$ has negative or zero weight imbalance. 
When $x_j[k] > 0$, node $v_j$ attempts to change the weights at its incoming edges $\{ f_{ji}[k+1] \; | \; v_i \in \mathcal{N}^-_j \}$, and outgoing edges $\{ f_{lj}[k+1] \; | \; v_l \in \mathcal{N}^+_j \}$ in a way that drives its weight imbalance $x_j[k+1]$ to zero (at least if no other changes are inflicted on the weights).  
More specifically, it goes through the links (incoming and outgoing) according to their ordering and changes their weights by a unit value, by $+1$ or $-1$, depending whether they are outgoing or incoming edges, respectively. 
If an outgoing (incoming) edge has reached its max (min) value then its weight does not change and node $v_j$ proceeds in changing the next one according to the predetermined order. 
According to the \textit{integer circulation conditions}, each node $v_j \in \mathcal{V}$ with positive weight imbalance at iteration $k$ ($x_j[k] > 0$) will always be able to calculate a weight assignment for its incoming and outgoing edge weights so that its weight imbalance becomes zero (at least if no other changes are inflicted on the weights of its incoming or outgoing links). This means that the selection of desirable weights in the above algorithm is always feasible.
The resulting additive {\em change} desired by node $v_j$ on $f_{ji}[k]$ of edge $(v_j,v_i) \in \mathcal{E}$ at iteration $k$ will be denoted by $c_{ji}^{(j)}[k]$.

\textit{Note:} Next time node $v_j$ needs to change the weights of its incoming/outgoing edges, it will continue from the edge it stopped the previous time and cycle through the edge weights in a round-robin fashion according to the ordering chosen at initialization.
If an edge reaches its maximum or minimum value, then its weight does not increase of decrease and node $v_j$ proceeds in changing the next one (according to the chosen ordering).
The total {\em change amount} desired by node $v_j$ on the weight $f_{ji}[k]$ of edge $(v_j,v_i) \in \mathcal{E}$ (or the weight $f_{lj}[k]$ of edge $(v_l,v_j) \in \mathcal{E}$) at iteration $k$ will be denoted by $c_{ji}^{(j)}[k]$ (or $c_{lj}^{(j)}[k]$).

\noindent
{\em B. Integer weight adjustment.} Since the integer weight $f_{ji}$ on each edge $(v_j, v_i) \in \mathcal{E}$ affects positively the weight imbalance $x_j[k]$ of node $v_j$ and negatively the weight imbalance $x_i[k]$ of node $v_i$, we need to take into account the possibility that both nodes are attempting to inflict changes on the integer weights simultaneously. 
Suppose that after one iteration of the algorithm (say at iteration $k$), node $v_i$ increases the weight $f_{ji}[k]$ of its outgoing edge $(v_j, v_i)$ by an integer value $c_{ji}^{(i)}[k]$ while node $v_j$ decreases the weight $f_{ji}[k]$ of its incoming edge $(v_j,v_i)$ by an integer value $c_{ji}^{(j)}[k]$. The new weight on edge $(v_j,v_i) \in \mathcal{E}$ is taken to be $f_{ji}[k+1] = f_{ji}[k] + c_{ji}^{(i)}[k] + c_{ji}^{(j)}[k]$. 
(Since we have $l_{ji} \leq f_{ji}[k] + c_{ji}^{(i)}[k] \leq u_{ji}$ and $l_{ji} \leq f_{ji}[k] + c_{ji}^{(j)}[k] \leq u_{ji}$, where $c_{ji}^{(i)} \geq 0$ and $c_{ji}^{(j)} \leq 0$, we have that $l_{ji} \leq f_{ji}[k+1] = f_{ji}[k] + c_{ji}^{(i)}[k] + c_{ji}^{(j)}[k] \leq u_{ji}$.]

\begin{remark}
In the enhanced version of Algorithm~\ref{constralg1}, each node $v_j$ with negative weight imbalance $x_j[k] < -1$ attempts to add $1$ to the weights of its incoming edges $\{ f_{ji}[k+1] \; | \; v_i \in \mathcal{N}^-_j \}$, and subtract $1$ from the weights of its outgoing edges $\{ f_{lj}[k+1] \; | \; v_l \in \mathcal{N}^+_j \}$ one at a time, following the predetermined order in a round-robin fashion, until its weight imbalance $x_j[k+1]$ becomes equal to $-1$ (at least if no other changes are inflicted on the weights).
\end{remark}

\begin{remark}
The above weight adjustment signifies that after iteration $k$ of the proposed distributed algorithm, once nodes $v_i$ and $v_j$ determine whether to increase (decrease) by $c_{ji}^{(i)}[k]$ ($c_{ji}^{(j)}[k]$) the weight of edge $(v_j,v_i)$ (where $c_{ji}^{(i)}[k], -c_{ji}^{(j)}[k] \in \mathbb{N}_0$), the new weight of edge $(v_j,v_i)$ will be $f_{ji}[k+1] = f_{ji}[k] + c_{ji}^{(i)}[k] + c_{ji}^{(j)}[k]$. According to the weight adjustment we have that $1 \leq l_{ji} \leq f_{ji}[k] + c_{ji}^{(i)}[k] \leq u_{ji}$ and $1 \leq l_{ji} \leq f_{ji}[k] + c_{ji}^{(j)}[k] \leq u_{ji}$, where $c_{ji}^{(i)}[k] \geq 0$, and $c_{ji}^{(j)}[k] \leq 0$ respectively. As a result we have that $1 \leq l_{ji} \leq f_{ji}[k] + c_{ji}^{(i)}[k] + c_{ji}^{(j)}[k] \leq u_{ji}$ $\Rightarrow$  $1 \leq l_{ji} \leq f_{ji}[k+1] \leq u_{ji}$ and $f_{ji}[k+1] \in \mathbb{N}_0$.  
\end{remark}

\subsection{Formal Description of Distributed Algorithm}
\label{formalupperloweralgorithm}

A formal description of the proposed distributed algorithm is presented in Algorithm~\ref{constralg1}.

\begin{varalgorithm}{5}
\caption{Distributed Balancing Under Link Capacity Constraints}
\textbf{Input} \\ 1) A strongly connected digraph $\mathcal{G}_d = (\mathcal{V}, \mathcal{E})$ with $n=|\mathcal{V}|$ nodes, $m=|\mathcal{E}|$ edges.\\ 2) $l_{ji},u_{ji}$ for every $(v_j, v_i) \in \mathcal{E}$. \\
\textbf{Initialization} \\ Set $k=0$; each node $v_j \in \mathcal{V}$ does:
\\1) It sets the weights on its incoming/outgoing edge weights as 
$$
f_{ji}[0] = \lceil l_{ji} \rceil, \ \forall  v_i \in \mathcal{N}_j^-,
$$
$$
f_{lj}[0] = \lceil l_{lj} \rceil, \ \forall  v_l \in \mathcal{N}_j^+.
$$
2) It assigns a unique order to its outgoing (or incoming) edges $(v_l,v_j)$ (or $(v_j,v_i)$) as $P^{(j)}_{lj}$ (or $P^{(j)}_{ji}$), for $v_l \in \mathcal{N}_j^+$ (or $v_i \in \mathcal{N}_j^-$) (such that $\{P^{(j)}_{lj} \; | \; v_l \in \mathcal{N}^+_j\} \cup \{P^{(j)}_{ji} \; | \; v_i \in \mathcal{N}^-_j\} = \{0,1,..., \mathcal{D}_j-1\}$).
\\
\textbf{Iteration} \\ For $k=0,1,2,\dots$, each node $v_j \in \mathcal{V}$ does the following:
\\ 1) It computes its \textit{weight imbalance} as
$$
x_j[k] = \sum_{v_i \in \mathcal{N}_j^-} f_{ji}[k] - \sum_{v_l \in \mathcal{N}_j^+} f_{lj}[k]. 
$$
2) If  $x_j[k] > 0$, it increases (decreases) by $1$ the integer weights $f_{lj}[k]$ ($f_{ji}[k]$) of its outgoing (incoming) edges $v_l \in \mathcal{N}_j^+$ ($v_i \in \mathcal{N}_j^-$) one at a time, following the predetermined order $P^{(j)}_{lj}$ ($P^{(j)}_{ji}$) until its weight imbalance becomes zero (if an edge has reached its maximum (minimum) value and it cannot be increased (decreased) further, its weight does not change and node $v_j$ proceeds in changing the next one according to the predetermined order.
\\ (\textit{Enhanced} version only) If  $x_j[k] < -1$, it decreases (increases) by $1$ the integer weights $f_{lj}[k]$ ($f_{ji}[k]$) of its outgoing (incoming) edges $v_l \in \mathcal{N}_j^+$ ($v_i \in \mathcal{N}_j^-$) one at a time, following the predetermined order $P^{(j)}_{lj}$ ($P^{(j)}_{ji}$) until its weight imbalance becomes $-1$ (if an edge has reached its minimum (maximum) value and it cannot be decreased (increased) further, its weight does not change and node $v_j$ proceeds in changing the next one according to the predetermined order).
\\ 3) It transmits the amount of change $c_{lj}^{(j)}[k]$ (or $c_{ji}^{(j)}[k]$) on each outgoing (or incoming) edge.
\\ 4) It receives the amount of change $c_{lj}^{(l)}[k]$ (or $c_{ji}^{(i)}[k]$) from each outgoing (or incoming) edge. Then, it sets its incoming and outgoing weights to be 
$$
f_{ji}[k+1] = \max(\min(f_{ji}[k] + c_{ji}^{(i)}[k] + c_{ji}^{(j)}[k], u_{ji}), l_{ji})
$$ 
for its incoming weights and 
$$
f_{lj}[k+1] = \max(\min(f_{lj}[k] + c_{lj}^{(l)}[k] + c_{lj}^{(j)}[k], u_{lj}) , l_{lj})
$$
for its outgoing weights.
\\ 5) It repeats (increases $k$ to $k+1$ and goes back to Step~1).
\label{constralg1}
\end{varalgorithm}

\begin{remark}
\label{REMfeasibility}
In Steps~3 and 4 of the algorithm, after node $v_j$ calculates the new weight assignment for both its incoming and outgoing edges, it transmits the amount of change on each outgoing and incoming edge. 
The weight assignment on edge $(v_j, v_i)$ is determined by the two incident nodes ($v_j$ and $v_i$) and the new weight becomes $f_{ji}[k+1] = f_{ji}[k] + c_{ji}^{(i)}[k] + c_{ji}^{(j)}[k]$ where $c_{ji}^{(i)}[k]$ (or $c_{ji}^{(j)}[k]$) is the change desired by node $v_i$ (or $v_j$). 
According to the Integer Circulation Conditions in Section~\ref{CircConditions}, each node $v_j \in \mathcal{V}$ with nonzero weight imbalance will always be able to calculate desired changes for its incoming and outgoing edge weights, so that its weight imbalance becomes zero (or equal to $-1$). 
This means that the selection of desirable weights in the above algorithm is always feasible.
\end{remark}

\begin{remark}
Note that, in the enhanced version of Algorithm~\ref{constralg1}, no attempt to change weights is made if node $v_j$ has weight imbalance equal to $-1$ or zero.
If nodes with negative weight imbalance were as aggressive as nodes with positive weight imbalance (and tried to make weight changes that would make their balance zero), then one could run into periodicity problems. As an example, consider the case of a ring digraph $\mathcal{G}_d$ with nodes $\mathcal{V} = \{ v_1, v_2, v_3, v_4 \}$ and edges $\mathcal{E} = \{ e_{21}, e_{32}, e_{43}, e_{14} \}$ where $e_{21} = (v_2,v_1)$, $e_{32} = (v_3,v_2)$, $e_{43} = (v_4,v_3)$ and $e_{14} = (v_1,v_4)$. 
Suppose the edge orders are as follows: for $v_1$, $P^{(1)}_{21} = \{0\}, P^{(1)}_{14} = \{1\}$; for  $v_2$, $P^{(2)}_{32} = \{0\}, P^{(2)}_{21} = \{1\}$; for $v_3$, $P^{(3)}_{43} = \{0\}, P^{(3)}_{32} = \{1\}$; and for $v_4$, $P^{(4)}_{14} = \{0\}, P^{(4)}_{43} = \{1\}$ (i.e., each node will first change the weight of its outgoing link and then the weight of its incoming link). If all nodes (with positive or negative weight imbalance) tried to make weight changes to balance themselves, the edge weights for time steps $k=0, 1, 2, 3, 4$, would be
\begin{eqnarray}
    k =0 & : & w_{21}[0] = 1,  w_{32}[0] = 1, w_{43}[0] = 2, w_{14}[0] = 2,  \nonumber \\
    k =1 & : & w_{21}[1] = 2,  w_{32}[1] = 1, w_{43}[1] = 1, w_{14}[1] = 2,  \nonumber \\
    k =2 & : & w_{21}[2] = 2,  w_{32}[2] = 2, w_{43}[2] = 1, w_{14}[2] = 1,  \nonumber \\ 
    k =3 & : & w_{21}[3] = 2,  w_{32}[3] = 1, w_{43}[3] = 1, w_{14}[3] = 2,  \nonumber \\
    k =4 & : & w_{21}[4] = 1,  w_{32}[4] = 1, w_{43}[4] = 2, w_{14}[4] = 2.  \nonumber
\end{eqnarray}
We see that at iteration $k=4$, weights (thus, node balances) are exactly the same as at $k=0$; moreover, the ordering in which changes will be made at the edges of each node is also exactly the same as in iteration $k=0$. We conclude that we have encountered periodic behavior and the digraph $\mathcal{G}_d$ will never become balanced (i.e., $w_{lj}[k] = w_{lj}[k+4]$, $\forall (v_l,v_j) \in \mathcal{E}$ and $x_j[k] = x_j[k+4]$, $\forall v_j \in \mathcal{V}$).
\end{remark}

\begin{remark}
The operation of Algorithm~\ref{constralg1} and its enhanced version, can be extended also for the case where each node $v_j$ does not necessarily want a weight imbalance equal to zero, but rather demands a weight imbalance equal to $x_j^{(s)} \neq 0$. 
When $\sum_{j=1}^{n} x_j^{(s)} = 0$, and for each $\mathcal{S}$ (where $\mathcal{S} \subset \mathcal{V}$) we have
$\sum_{(v_j, v_i) \in \mathcal{E}^-_\mathcal{S}} \lceil l_{ji} \rceil \leq \sum_{(v_l, v_j) \in \mathcal{E}^+_\mathcal{S}} \lfloor u_{lj} \rfloor + \sum_{v_j \in \mathcal{S}} x_j^{(s)}$ 
(where $\mathcal{E}^-_\mathcal{S}$ and $\mathcal{E}^+_\mathcal{S}$ were defined in (\ref{REALEQinS}) and (\ref{REALEQoutS})), then every node can obtain a nonzero weight imbalance equal to $x_j^{(s)}$ as long as the nodes with positive weight imbalance operate according to Step~2 until their weight imbalance becomes equal to $x_j^{(s)}$, and the nodes with negative weight imbalance operate according to Step~3 until their weight imbalance becomes equal to $x_j^{(s)} - 1$.
\end{remark}

We now illustrate the distributed algorithm via an example and then explain why it results in a weight balanced digraph after a finite number of iterations. We also obtain bounds on its execution time.

\subsection{Illustrative Example of Distributed Algorithm}
\label{exampleupperloweralgorithm}

Consider the digraph $\mathcal{G}_d=(\mathcal{V},\mathcal{E})$ in Figure~\ref{example-initial-upperlower}, where $\mathcal{V}=(v_1,v_2,\dots,v_5)$, $\mathcal{E}=(m_{21},\dots,m_{45})$, $\mathcal{E} \subseteq \mathcal{V} \times \mathcal{V} - \{ (v_j,v_j)$ $|$ $v_j \in \mathcal{V} \}$, $L=[l_{ji}]$ and $U=[u_{ji}]$ for every $(v_j, v_i) \in \mathcal{E}$. The weight on each edge is initialized to $f_{ji}[0] = \lceil l_{ji} \rceil$ for $(v_l,v_j) \in \mathcal{E}$ and each node assigns a unique order to each of its outgoing and incoming edges. For the purposes of this example, let us assume that this order is as follows:
\begin{itemize}
\item $v_1 : P^{(1)}_{21}=1, P^{(1)}_{31}=2, P^{(1)}_{41}=3$,
\item $v_2 : P^{(2)}_{32}=1, P^{(2)}_{52}=2, P^{(2)}_{21}=3$,
\item $v_3 : P^{(3)}_{43}=1, P^{(3)}_{53}=2, P^{(3)}_{32}=3, P^{(3)}_{31}=4$,
\item $v_4 : P^{(4)}_{14}=1, P^{(4)}_{43}=2, P^{(4)}_{45}=3$,
\item $v_5 : P^{(5)}_{45}=1, P^{(5)}_{52}=2, P^{(5)}_{53}=3$.
\end{itemize}
\noindent (For example, node $v_2$ will first increase $f_{32}$, then $f_{52}$ and then it will decrease $f_{21}$.) As a first step, each node computes its weight imbalance $x_j[0] = \sum_{v_i \in \mathcal{N}_j^-} f_{ji}[0] - \sum_{v_l \in \mathcal{N}_j^+} f_{lj}[0]$ (these values are shown in Figure~\ref{example-initial-upperlower}).

\begin{figure} [ht]
\centering
\includegraphics[width=0.50\textwidth]{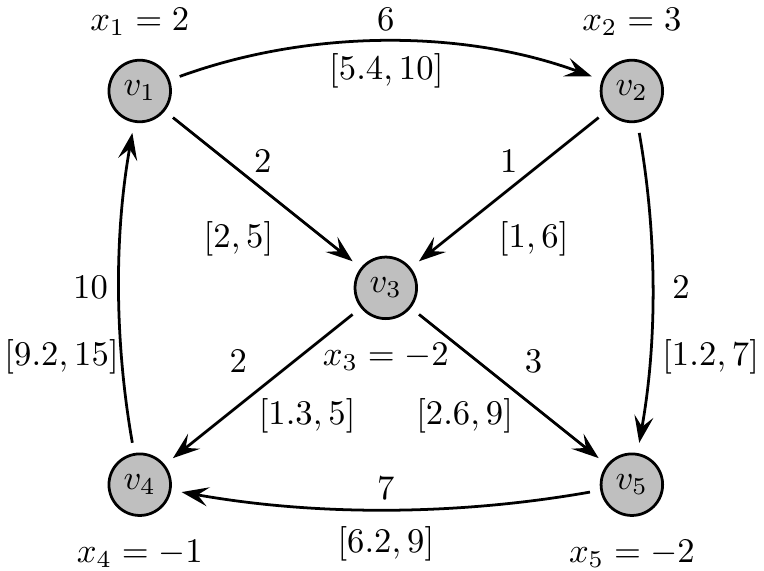}
\caption{Weighted digraph with initial weights and initial imbalances for each node.}
\label{example-initial-upperlower}
\end{figure}

Once each node computes its imbalance, the distributed algorithm requires each node with positive weight imbalance to increase (decrease) by $1$ the integer weights $f_{lj}[k]$ ($f_{ji}[k]$) of its outgoing (incoming) edges $v_l \in \mathcal{N}_j^+$ ($v_i \in \mathcal{N}_j^-$) one at a time, following the predetermined order $P_{lj}$ until its weight imbalance becomes zero. 
In this case, the nodes that have positive weight imbalance are nodes $v_1$ and $v_2$ (equal to $x_1[0]=2$ and $x_2[0]=3$) respectively), and they increase their outgoing edges as shown in Figure~\ref{example-initial-upperlower2}.

\begin{figure} [ht]
\centering
\includegraphics[width=0.50\textwidth]{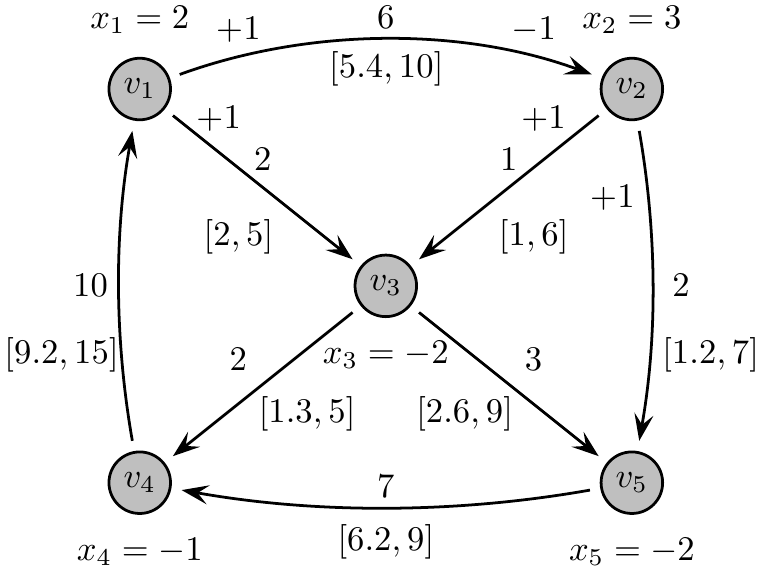}
\caption{Distribution of imbalance from positively imbalanced nodes.}
\label{example-initial-upperlower2}
\end{figure}

In the next step of the distributed algorithm each node transmits the amount of change $c_{lj}^{(j)}[k]$ (or $c_{ji}^{(j)}[k]$) on each outgoing (or incoming) edge and receives the amount of change $c_{lj}^{(l)}[k]$ (or $c_{ji}^{(i)}[k]$) from each outgoing (or incoming) edge. Then, it sets its incoming (outgoing) weights to be $f_{ji}[k+1] = f_{ji}[k] + c_{ji}^{(i)}[k] + c_{ji}^{(j)}[k]$ ($f_{lj}[k+1] = f_{lj}[k] + c_{lj}^{(l)}[k] + c_{lj}^{(j)}[k]$). This can be seen in Figure~\ref{example-initial-upperlower3}.

\begin{figure} [ht]
\centering
\includegraphics[width=0.50\textwidth]{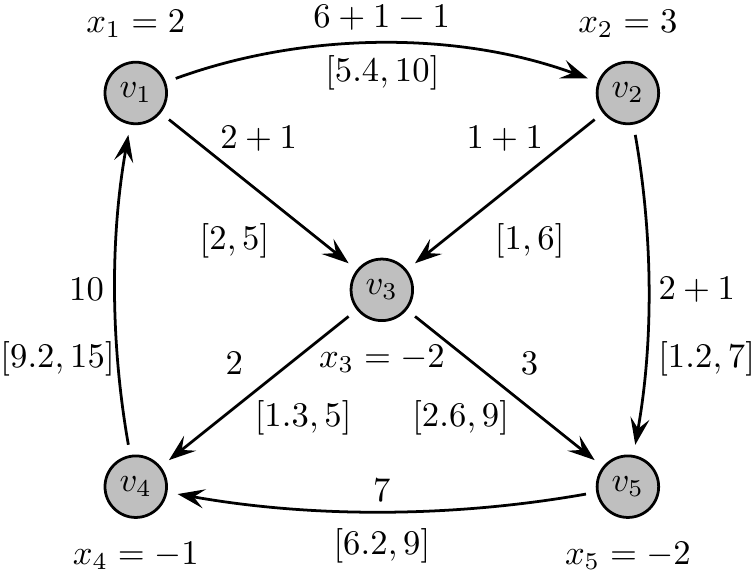}
\caption{Calculation of edge weights from positively imbalanced nodes.}
\label{example-initial-upperlower3}
\end{figure}

Each node, after the integer weight update on its outgoing and incoming edges, recalculates its imbalances $x_j[1]$, $\forall v_j \in \mathcal{V}$, and the process is repeated.
After a finite number of iterations, shown in the next section, we reach the weighted digraph with integer weights shown in Figure.~\ref{example-initial-upperlower4}.

\begin{figure} [ht]
\centering
\includegraphics[width=0.50\textwidth]{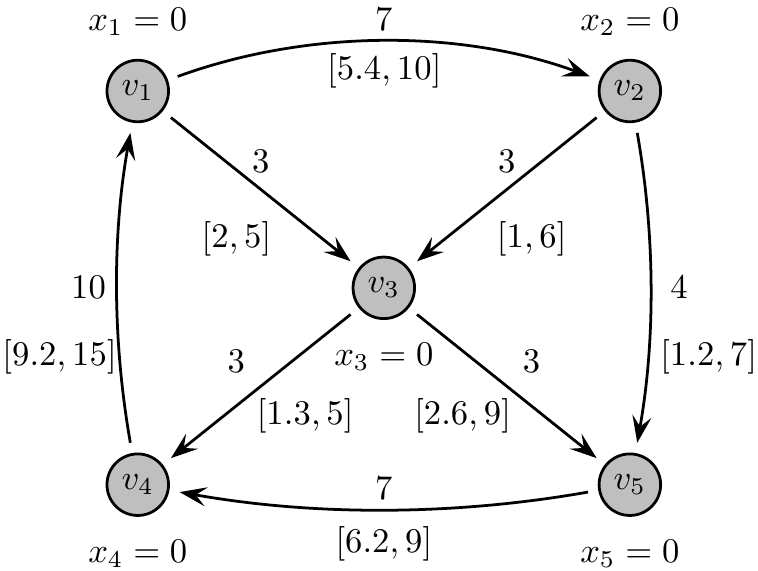}
\caption{Final weight balanced digraph.}
\label{example-initial-upperlower4}
\end{figure}

\subsection{Proof of Algorithm Completion}\label{constr_convergence}

We show that, as long as the Integer Circulation Conditions in Section~\ref{CircConditions} hold, then the total imbalance $\varepsilon[k]$ in Definition~\ref{defn:totalim} goes to zero after a finite number of iterations of Algorithm~\ref{constralg1}. 
This implies that the weight imbalance $x_j[k]$ for each node $v_j \in \mathcal{V}$ goes to zero after a finite number of iterations, and thus (from the updates in Algorithm~\ref{constralg1}) the integer weight $f_{ji}[k]$ on each edge $(v_j, v_i) \in \mathcal{E}$ stabilizes to an integer value $f_{ji}^*$ (where $f_{ji}^* \in \mathbb{N}_0$) within the given lower and upper limits, i.e., $1 \leq l_{ji} \leq f_{ji}^* \leq u_{ji}$ for all $(v_j, v_i) \in \mathcal{E}$.

\noindent
We begin by establishing some preliminary results.

\begin{prop}
\label{constrPROP1}
Consider the problem formulation described in Section~\ref{ProbStatementUpperLower}. At each iteration $k$ during the execution of Algorithm~\ref{constralg1}, it holds that
\begin{enumerate}
\item For any subset of nodes $\mathcal{S} \subset \mathcal{V}$, let $\mathcal{E}^-_\mathcal{S}$ and $\mathcal{E}^+_\mathcal{S}$ be defined by (\ref{REALEQinS}) and (\ref{REALEQoutS}) respectively. Then,
$$
\sum_{v_j \in \mathcal{S}} x_j[k] = \sum_{(v_j, v_i) \in \mathcal{E}^-_\mathcal{S}} f_{ji}[k] - \sum_{(v_l, v_j) \in \mathcal{E}^+_\mathcal{S}} f_{lj}[k] \; ;
$$
\item $\sum_{j=1}^n x_j[k] = 0$;
\item $\varepsilon[k] = 2 \sum_{v_j \in \mathcal{V}^-[k]} \vert x_j[k] \vert$ where $\mathcal{V}^-[k] = \{ v_j \in \mathcal{V} \; | \; x_j[k] < 0 \}$.
\end{enumerate}
\end{prop}

\begin{proof}
To prove the first statement, let 
$$
\mathcal{E}_\mathcal{S} = \{ (v_j, v_i) \in \mathcal{E} \; | \; v_j \in \mathcal{S}, \; v_i \in \mathcal{S} \}
$$
be the set of edges that are internal to the set $\mathcal{S}$. From the definition of the weight imbalance for node $v_j$, we have (after re-arranging the summations)
\begin{eqnarray*}
\sum_{v_j \in \mathcal{S}} x_j[k] & = & \sum_{v_j \in \mathcal{S}} \left ( \sum_{v_i \in \mathcal{N}_j^-} f_{ji}[k] - \sum_{v_l \in \mathcal{N}_j^+} f_{lj}[k] \right ) \\
 & = & \sum_{(v_j, v_i) \in \mathcal{E}^-_\mathcal{S}} f_{ji}[k] - \sum_{(v_l, v_j) \in \mathcal{E}^+_\mathcal{S}} f_{lj}[k] + \\
 &    & + \sum_{(v_j, v_i) \in \mathcal{E}_\mathcal{S}} f_{ji}[k] - \sum_{(v_l, v_j) \in \mathcal{E}_\mathcal{S}} f_{lj}[k] \\ 
 & = & \sum_{(v_j, v_i) \in \mathcal{E}^-_\mathcal{S}} f_{ji}[k] - \sum_{(v_l, v_j) \in \mathcal{E}^+_\mathcal{S}} f_{lj}[k] \; .
\end{eqnarray*}

For the second statement, we can take any $\mathcal{S} \subset \mathcal{V}$ and argue that
\begin{eqnarray*}
\sum_{v_j \in \mathcal{V}} x_j[k] & = & \sum_{v_j \in \mathcal{S}} x_j[k] + \sum_{v_{j} \in \mathcal{V}-\mathcal{S}} x_{j}[k] \\
 & = & \sum_{(v_j, v_i) \in \mathcal{E}^-_\mathcal{S}} f_{ji}[k] - \sum_{(v_l, v_j) \in \mathcal{E}^+_\mathcal{S}} f_{lj}[k] + \\
 &    & + \sum_{(v_{j}, v_i) \in \mathcal{E}^-_{\mathcal{V}-\mathcal{S}}} f_{ji}[k] - \sum_{(v_l, v_{j}) \in \mathcal{E}^+_{\mathcal{V}-\mathcal{S}}} f_{lj}[k] \\
 & = & 0 \; ,
\end{eqnarray*}
where the last line follows from the fact that $\mathcal{E}^+_\mathcal{S} = \mathcal{E}^-_{\mathcal{V}-\mathcal{S}}$ and $\mathcal{E}^-_\mathcal{S} = \mathcal{E}^+_{\mathcal{V}-\mathcal{S}}$.

For the third statement, notice that, from the definition of $\varepsilon[k]$ in Definition~\ref{defn:totalim}, we have
\begin{eqnarray*}
\varepsilon[k] & = & \sum_{v_j \in \mathcal{V}} | x_j[k] | \\
 & = & \sum_{v_j \in \mathcal{V}^-[k]} | x_j[k] | + \sum_{v_j \in \mathcal{V}-\mathcal{V}^-[k]} |x_j[k]| \\
 & = & \sum_{v_j \in \mathcal{V}^-[k]} | x_j[k] | + \sum_{v_j \in \mathcal{V}-\mathcal{V}^-[k]}  x_j[k]  \\
 & = & 2 \sum_{v_j \in \mathcal{V}^-[k]} | x_j[k] | \; ,
\end{eqnarray*}
where the third line follows from the definition of $\mathcal{V}^-[k]$ (all nodes have nonnegative balance) and the last line follows from the second statement of this proposition.
\end{proof}

\begin{prop}
\label{constrPROP2}
Consider the problem formulation described in Section~\ref{ProbStatementUpperLower}. Let $\mathcal{V}^-[k] \subset \mathcal{V}$ be the set of nodes with negative weight imbalance at iteration $k$, i.e., $\mathcal{V}^-[k] = \{ v_j \in \mathcal{V} \; | \; x_j[k] < 0 \}$.  During the execution of Algorithm~\ref{constralg1}, we have that
$$
\mathcal{V}^-[k+1] \subseteq \mathcal{V}^-[k].
$$
\end{prop}

\begin{proof}
We will argue that nodes with nonnegative weight imbalance at iteration $k$ can never reach negative weight imbalance at iteration $k+1$, thus establishing the proof of the proposition.

Consider a node $v_j$ with a nonnegative weight imbalance $x_j[k] \geq 0$. Node $v_j$ may attempt to make weight changes on its edges: $c_{ji}^{(j)}[k] \leq 0$ for all $v_i \in \mathcal{N}_j^-$ and $c_{lj}^{(j)}[k] \geq 0$ for all $v_l \in \mathcal{N}_j^+$. If no in-neighbor or out-neighbor of node $v_j$ attempts to inflict changes on the weights of these edges, then it is not hard to see that the weight imbalance of node $v_j$ at iteration $k+1$ will be
\begin{eqnarray*}
x_j[k+1] & = & \sum_{v_i \in \mathcal{N}_j^-} f_{ji}[k+1] - \sum_{v_l \in \mathcal{N}_j^+} f_{lj}[k+1] \\
 & \kern-2em = & \kern-2em \sum_{v_i \in \mathcal{N}_j^-} (f_{ji}[k] + c_{ji}^{(j)}[k]) - \sum_{v_l \in \mathcal{N}_j^+} (f_{lj}[k] + c_{lj}^{(j)}[k]) \\
 & \kern-2em = & \kern-1em x_j[k] - x_j[k] = 0
\end{eqnarray*}
(because, by design, the changes in the weights are chosen so that the balance becomes zero).

If one (or more) of the in-neighbors or out-neighbors of node $v_j$ have nonegative balance, then they will also attempt to make changes on the weights. In particular,
$$
f_{ji}[k+1] = f_{ji}[k] + c_{ji}^{(j)}[k] + c_{ji}^{(i)}[k] \; ,
$$
where $c_{ji}^{(i)}[k] \geq 0$, and 
$$
f_{lj}[k+1] = f_{lj}[k] + c_{lj}^{(j)}[k] + c_{lj}^{(l)}[k] \; ,
$$
where $c_{lj}^{(l)}[k] \leq 0$. Putting these together we have 
\begin{eqnarray*}
x_j[k+1] & = & \sum_{v_i \in \mathcal{N}_j^-} f_{ji}[k+1] - \sum_{v_l \in \mathcal{N}_j^+} f_{lj}[k+1] \\
 & = & \sum_{v_i \in \mathcal{N}_j^-} ( f_{ji}[k] + c_{ji}^{(j)}[k] + c_{ji}^{(i)}[k] ) - \\
 & & \; \; - \sum_{v_l \in \mathcal{N}_j^+} ( f_{lj}[k] + c_{lj}^{(j)}[k] + c_{lj}^{(l)}[k] ) \\
 & = & 0 + \sum_{v_i \in \mathcal{N}_j^-} c_{ji}^{(i)}[k] - \sum_{v_l \in \mathcal{N}_j^+} c_{lj}^{(l)}[k] \\
 & \geq & 0 \; .
\end{eqnarray*}

For the first case where all neighbors of node $v_j$ do not belong in $\mathcal{V}^-[k]$ then $x_i[k], x_l[k] \geq 0$, for every $v_i \in  \mathcal{N}_j^-$ and $v_l \in  \mathcal{N}_j^+$. This means that during iteration $k$ of Algorithm~\ref{constralg1}, the edge weights of $v_j$ will change by its in/out neighbors (i.e., the weights of the incoming edges $(v_j,v_i)$ will be increased by $c_{ji}^{(i)}[k]$ where $v_i \in \mathcal{N}_j^-$ while the weights of the outgoing edges will be decreased by $c_{lj}^{(l)}[k]$ where $v_l \in \mathcal{N}_j^+$). In that case the weight imbalance of node $v_j$ will increase ($x_j[k+1] > x_j[k]$). For the new weight imbalance of $v_j$ (i.e., $x_j[k+1]$) we have that it will either remain negative (i.e., $x_j[k+1] < 0$), or it will become nonegative $v_j \in \mathcal{V}^-[k+1]$.  
For the second case, if all the neighbors of node $v_j$ belong in $\mathcal{V}^-[k]$ they will not make any weight changes on the edges that connect them with node $v_j$ and so we have that $v_j \in \mathcal{V}^-[k+1]$. But if some of the neighbors $v_i$ and $v_l$ (where $v_i \in  \mathcal{N}_j^-$ and $v_l \in  \mathcal{N}_j^+$) do not belong in $\mathcal{V}^-[k]$ (i.e., $x_i[k], x_l[k] \geq 0$) then they will increase the incoming and decrease the outgoing edge weights by $c_{ji}^{(i)}[k]$ and $c_{lj}^{(l)}[k]$ respectively and so the weight imbalance of node $v_j$ will be improved (i.e., $x_j[k+1] > x_j[k]$). For the new weight imbalance of $v_j$ (i.e., $x_j[k+1]$) we have that it will either remain negative (i.e., $x_j[k+1] < 0$) or it will become nonegative and remain so for the rest of the iterations (i.e., $x_j[k+1] \geq 0$).
As a result we have that during the execution of Algorithm~\ref{constralg1}, $\mathcal{V}^-[k+1] \subseteq \mathcal{V}^-[k], \  \forall k \geq 0$. 
\end{proof}

\begin{prop}
\label{constrPROP3}
Consider the problem formulation described in Section~\ref{ProbStatementUpperLower}. During the execution of Algorithm~\ref{constralg1}, it holds that
$$
0 \leq \varepsilon[k+1] \leq \varepsilon[k] \; , \; \; \forall k \geq 0 \; ,
$$
where $\varepsilon[k] \geq 0$ is the total imbalance of the network at iteration~$k$ (as defined in Definition~\ref{defn:totalim}).


\end{prop}

\begin{proof}
From the third  statement of Proposition~\ref{constrPROP1}, we have $\varepsilon[k+1] = 2 \sum_{v_j \in \mathcal{V}^-[k+1]} \vert x_j[k+1] \vert$ and $\varepsilon[k] = 2 \sum_{v_j \in \mathcal{V}^-[k]} \vert x_j[k] \vert$, whereas from Proposition~\ref{constrPROP2}, we have $\mathcal{V}^-[k+1] \subseteq \mathcal{V}^-[k]$.

Consider a node $v_j \in \mathcal{V}^-[k]$ with weight imbalance $x_j[k]<0$. We analyze below the following two cases:
\begin{enumerate}
\item All neighbors of node $v_j$ have negative or zero weight imbalance;
\item At least one neighbor of node $v_j$ has positive weight imbalance.
\end{enumerate}
In both cases, node $v_j$ above will not make any weight changes on its edges. In the first case, the weight imbalance of node $v_j$ will not change (i.e., $x_j[k+1] = x_j[k] < 0$). In the second case, we have $x_i[k]\geq 0$ or $x_l[k] \geq 0$, for some $v_i \in  \mathcal{N}_j^-$ or $v_l \in  \mathcal{N}_j^+$. This means that during iteration $k$ of Algorithm~\ref{constralg1}, the edge weights of $v_j$ might change by its in/out neighbors (i.e., the weight of an incoming edge $(v_j,v_i)$ might be increased by $c_{ji}^{(i)}[k] \geq 0$ (for some) $v_i \in \mathcal{N}_j^-$ or the weight of an outgoing edge might be decreased by $c_{lj}^{(l)}[k] \geq 0$ (for some) $v_l \in \mathcal{N}_j^+$). Thus, in the second case, the weight imbalance of node $v_j$ will satisfy $x_j[k+1] \geq x_j[k]$. In fact, we will have that either $x_j[k+1] \geq 0$ (i.e., $v_j \notin \mathcal{V}^-[k+1]$) or $x_j[k+1]<0$ and $|x_j[k+1]| \leq |x_j[k]|$. As a result, in both cases, we have $\varepsilon[k+1] \leq \varepsilon[k]$ (using the third statement in Proposition~\ref{constrPROP1}).
%
\end{proof}

According to the above proposition we have that the total imbalance $\varepsilon[k]$ of the network at iteration~$k$ (as defined in Definition~\ref{defn:totalim}) will be reduced after a finite number of iterations.

\begin{prop}
\label{constrPROP4}
Consider the problem formulation described in Section~\ref{ProbStatementUpperLower} where the integer circulation conditions are satisfied. Algorithm~\ref{constralg1} balances the weights in the graph in a finite number of steps (i.e., $\exists \ k_0$ s.t. $\forall k \geq k_0$, $f_{ji}[k_0] = f_{ji}[k]$, $\forall (v_j,v_i) \in \mathcal{E}$ and $x_j[k] = x_j[k_0] =0$, $\forall v_j \in \mathcal{V}$). 
\end{prop}

\begin{proof}
By contradiction, suppose Algorithm~\ref{constralg1} runs for an infinite number of iterations and the total imbalance remains positive (i.e., $\varepsilon[k]>0$ for all $k$). 
Then, there is always (at each $k$) at least one node with positive weight imbalance. 
Let $\mathcal{V}^{+}[k] = \{ v_j \in \mathcal{V} \; | \; x_j[k] \geq 0 \}$ be the set of nodes that have positive weight imbalance at time step $k$. 
Let $\mathcal{V}^{+}$ denote the set of nodes that have positive weight imbalance {\em infinitely often}. [Since nodes with positive weight imbalance can become balanced (but not negatively balanced), this means that nodes in the set ${\mathcal V}^{+}$ could become balanced at some iteration, as long as they become positively imbalanced at later iterations.] 
Also, we can define the set of nodes ${\mathcal V}^{-}$ as $\mathcal{V}^{-} = \lim_{k \rightarrow \infty} \mathcal{V}^{-}[k]$, where $\mathcal{V}^{-}[k] = \{ v_j \in \mathcal{V} \; | \; x_j[k] < 0 \}$. 
This set is well defined (due to the fact that positively imbalanced nodes cannot become negatively balanced) and contains at least one node with negative weight imbalance (otherwise the graph is balanced). 
The above discussion implies that as $k$ goes to infinity, the set of nodes $\mathcal{V}$ can be partitioned into three sets: $\mathcal{V}^-$, $\mathcal{V}^+$, and $\mathcal{V} - (\mathcal{V}^+ \cup \mathcal{V}^-)$ (the latter is the set of nodes that remain balanced after a finite number of steps --and never obtain positive imbalance again). This is shown in Fig.~\ref{FlowConv}.

\vspace{-0.3cm}

\begin{figure}[h]
\begin{center}
\includegraphics[width=0.40\columnwidth]{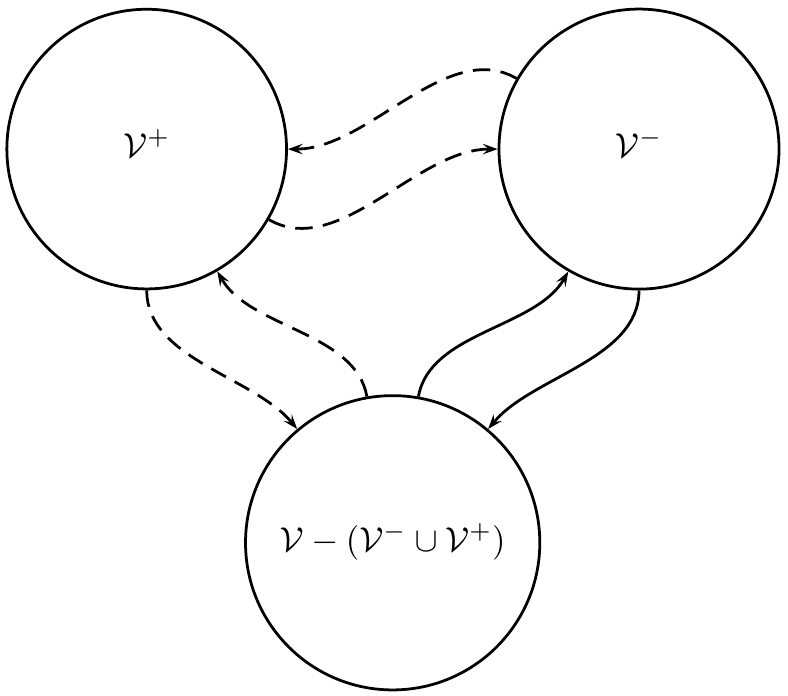}
\caption{Example of digraph where the Integer Circulation Conditions in Section~\ref{CircConditions} do not hold for the dashed edges.}
\label{FlowConv}
\end{center}
\end{figure}

\vspace{-0.4cm}

Since the graph is strongly connected, nodes in the set $\mathcal{V}^+$ need to be connected to/from nodes in the other two sets. This is shown via the dashed edges in Fig.~\ref{FlowConv} (note that the presence of all four types of edges is not necessary, but there has to be at least one edge from a node in $\mathcal{V}^+$ to a node in one of the two other sets, and at least one edge from a node in one of the two other sets to a node in $\mathcal{V}^+$).

Take $\mathcal{S} \subset \mathcal{V}$ to be $\mathcal{V}^{+}$ and note that $\mathcal{S}$ has at least one node (since the number of nodes is finite). Consider a node $v_j \in \mathcal{V}^{-}$ that has at least one in-neighbor $v_i$ in $\mathcal{S}$ and/or at least one out-neighbor $v_l$ in $\mathcal{S}$. 
Since $v_i$ (and/or $v_l$) has a positive weight imbalance infinitely often, it will eventually attempt to change the weight $f_{ji}[k]$ (or $f_{lj}[k]$) by $c_{ji}^{(i)}[k] \geq 0$ (or $c_{lj}^{(l)}[k] \leq 0$). 
If this change happens, then $x_j[k+1] > x_j[k]$ so that $v_j$ either arrives at a nonegative weight imbalance (i.e., $x_j[k+1] \geq 0$, which is a contradiction) or $0 > x_j[k+1] > x_j[k]$ (implying\footnote{From the third  statement of Proposition~\ref{constrPROP1}, we have $\varepsilon[k+1] = 2 \sum_{v_j \in \mathcal{V}^{-}[k+1]} \vert x_j[k+1] \vert$ and $\varepsilon[k] = 2 \sum_{v_j \in \mathcal{V}^{-}[k]} \vert x_j[k] \vert$.} 
that $\varepsilon[k+1] < \varepsilon[k]$, which is also a contradiction because, if the integer valued $\varepsilon[k]$ decreases infinitely often, it will become zero, thus $x_j[k]=0$ for all $v_j \in \mathcal{V}$).

Consider now a node $v_{j'}$ that has \textit{zero} weight imbalance and has at least one in-neighbor $v_{i'}$ in $\mathcal{S}$ and/or at least one out-neighbor $v_{l'}$ in $\mathcal{S}$.
Since $v_{i'}$ (or $v_{l'}$) has a positive weight imbalance infinitely often, it will eventually attempt to change the weight $f_{j'i'}[k]$ (or $f_{l'j'}[k]$) by $c_{j'i'}^{(i')}[k] \geq 0$ (or $c_{l'j'}^{(l')}[k] \leq 0$).
If this change happens, then node $v_{j'}$ would eventually reach positive weight imbalance at some iteration, and this would happen infinitely often which is a contradiction because $v_{j'}$ belongs in the set of nodes with zero weight imbalance (at least after a large enough number of steps).

Thus, for Algorithm~\ref{constralg1}, the only possibility left is that the weights of edges outgoing from nodes in $\mathcal{S}$ cannot increase and the weights of edges incoming to nodes in $\mathcal{S}$ cannot decrease. 
In other words, for $k \geq k_0$ for some large enough $k_0$ we have
\begin{eqnarray*}
f_{ji}[k] = \lceil l_{ji} \rceil & & \forall (v_j, v_i) \in \mathcal{E}^-_\mathcal{S} \; , \\
f_{lj}[k] = \lfloor u_{lj} \rfloor & & \forall (v_l, v_j) \in \mathcal{E}^+_\mathcal{S} \; ,
\end{eqnarray*}
where $\mathcal{E}^-_\mathcal{S}$ and $\mathcal{E}^+_\mathcal{S}$ are defined by \eqref{REALEQinS} and \eqref{REALEQoutS} respectively.

From the first statement of Proposition~\ref{constrPROP1}, for the set $\mathcal{S}$, we have that $\sum_{v_j \in \mathcal{S}} x_j[k] = \sum_{(v_j, v_i) \in \mathcal{E}^-_\mathcal{S}} f_{ji}[k] - \sum_{(v_l, v_j) \in \mathcal{E}^+_\mathcal{S}} f_{lj}[k]$. 
Thus, we have
$$
\sum_{(v_j, v_i) \in \mathcal{E}^-_\mathcal{S}} l_{ji} - \sum_{(v_l, v_j) \in \mathcal{E}^+_\mathcal{S}} u_{lj} = \sum_{v_j \in \mathcal{S}} x_j[k] > 0 \; , 
$$
which means that the Integer Circulation Conditions in Section~\ref{CircConditions} do not hold (i.e., we reach a contradiction).

As a result we have that if the Integer Circulation Conditions in Section~\ref{CircConditions} hold, the total imbalance $\varepsilon[k]$ decreases after a finite number of iterations, and Algorithm~\ref{constralg1} results in a weight-balanced digraph after a finite number of iterations. 
\end{proof}

As a result we have that if the \textit{integer circulation conditions} hold, the total imbalance $\varepsilon[k]$ decreases after a finite number of iterations, and the algorithm results in a weight balanced digraph after a finite number of iterations.

\subsection{Simulation Study}
\label{resultsupperloweralgorithm}

In this section, we present simulation results for the proposed distributed algorithm. 
Specifically, we first present numerical results for a random graph of size $n=20$ illustrating the behavior of Algorithm~\ref{constralg1} for two different cases: (i) the case when the \textit{integer circulation conditions} do {\em not} hold, thus, a set of integer weights that balance the digraph cannot be obtained; (ii) the case when the \textit{integer circulation conditions} hold and a set of integer weights that balance the graph can be obtained. 
The weights are initialized at the ceiling of the lower bound of the feasible interval, i.e., $f_{ji}[0] = \lceil l_{ji} \rceil$.

Figure~\ref{violation20} shows what happens in the case of a randomly created graph of $20$ nodes, in which the \textit{integer circulation conditions} do not hold. 
In the first case, we plot the \textit{absolute imbalance} $\varepsilon = \sum_{j=1}^{n} \vert x_j \vert$, $\forall v_j \in \mathcal{V}$ (as defined in Definition~\ref{defn:totalim}) and in the second case the \textit{nodes weight imbalances} $x_j[k]$ (as defined in Definition~\ref{DEFnodebalance}) as a function of the number of iterations $k$ for the distributed algorithm. 
The plots suggest that the proposed distributed algorithm is unable to obtain a set of weights that balance the corresponding digraph.

\begin{figure}[ht] 
\centering    
\includegraphics[width=0.60\textwidth]{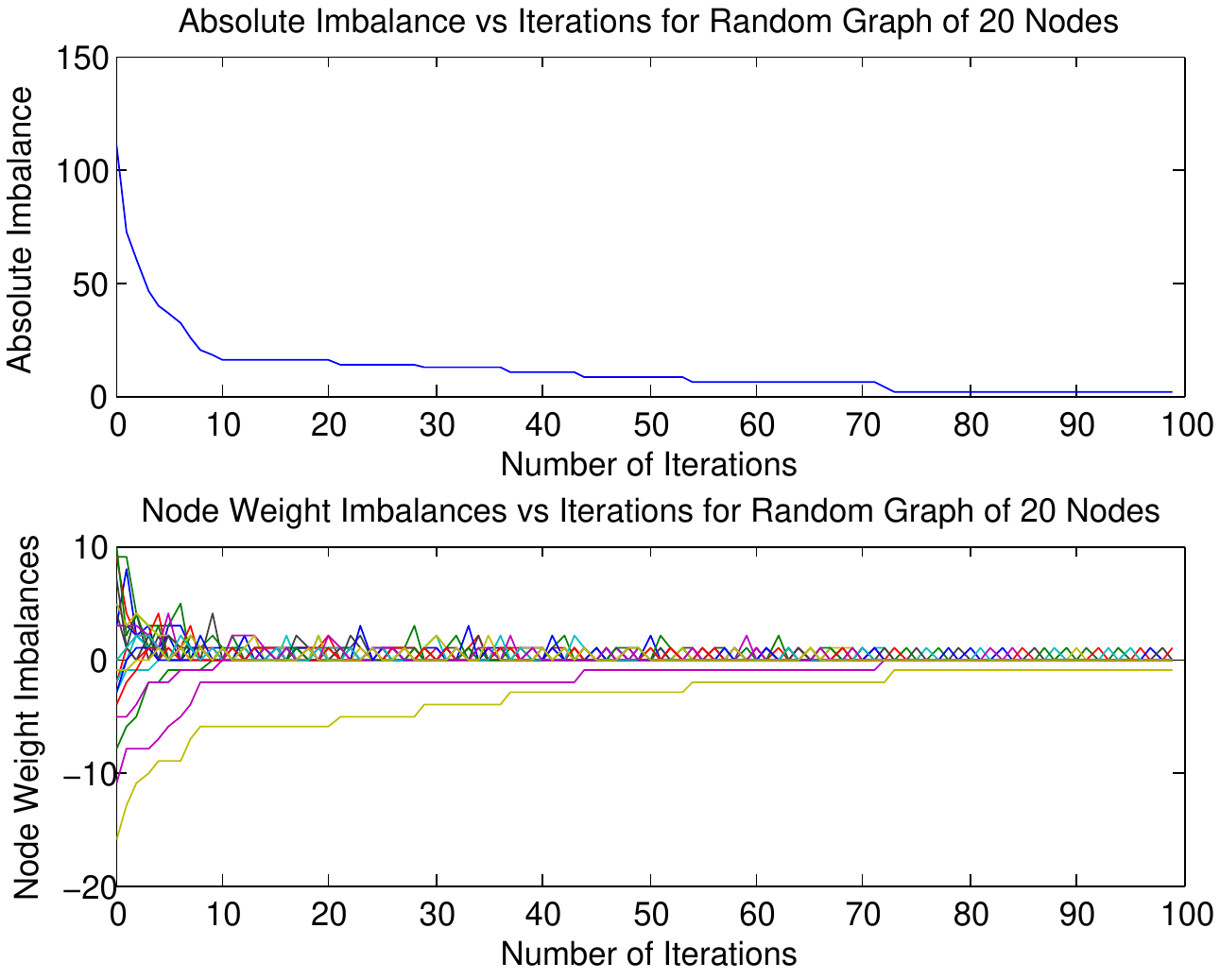}
\caption{Execution of Algorithm~\ref{constralg1} for the case when the integer circulation conditions do not hold for a random graph of $20$ nodes. \emph{Top figure:} Absolute imbalance $\varepsilon[k]$ plotted against number of iterations. \emph{Bottom figure:} Node weight imbalances $x_j[k]$ plotted against number of iterations.}
\label{violation20}
\end{figure}

Figure~\ref{working20} shows the same case as Figure~\ref{violation20}, with the difference that the \textit{integer circulation conditions} hold. 
Here, the plots suggest that the proposed distributed algorithm is able to obtain a set of integer weights that balance the corresponding digraph after a finite number of iterations.

\begin{figure}[ht] 
\centering    
\includegraphics[width=0.60\textwidth]{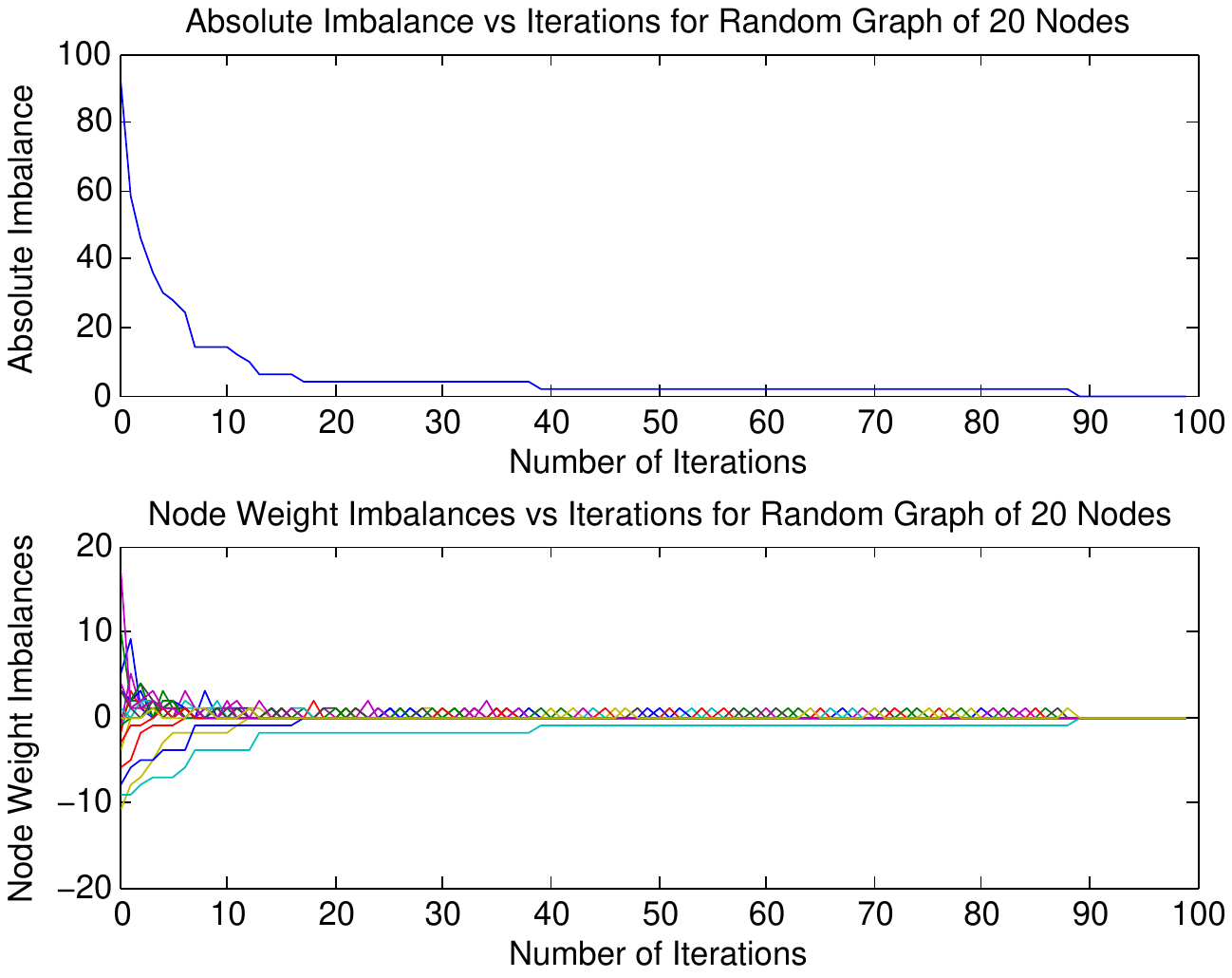}
\caption{Execution of Algorithm~\ref{constralg1} for the case when the integer circulation conditions hold for a random graph of $20$ nodes. \emph{Top figure:} Absolute imbalance $\varepsilon[k]$ plotted against number of iterations. \emph{Bottom figure:} Node weight imbalances $x_j[k]$ plotted against number of iterations.}
\label{working20}
\end{figure}

\begin{remark}
Note that both Figures.~\ref{violation20} and~\ref{working20} illustrate some of the properties established in the analysis in the previous section: for example: once nodes have positive weight imbalance, they retain a positive or zero weight imbalance; while nodes have negative weight imbalance, their imbalance increases monotonically; the absolute imbalance is monotonically non-increasing, and so forth.
\end{remark}

Figure~\ref{AVworking20} shows what happens in the case of $100$ averaged graphs of $20$ and $50$ nodes each when the \textit{integer circulation conditions} hold. We plot the average total (absolute) imbalance $\varepsilon[k] = \sum_{j=1}^n | x_j[k] |$ (as defined in Definition~\ref{defn:totalim}) as a function of the number of iterations $k$ for the distributed algorithm. The plot suggests that the proposed distributed algorithm is able to obtain a set of integer weights that balance the corresponding graph after a finite number of iterations.

\begin{figure}[ht] 
\centering    
\includegraphics[width=0.60\textwidth]{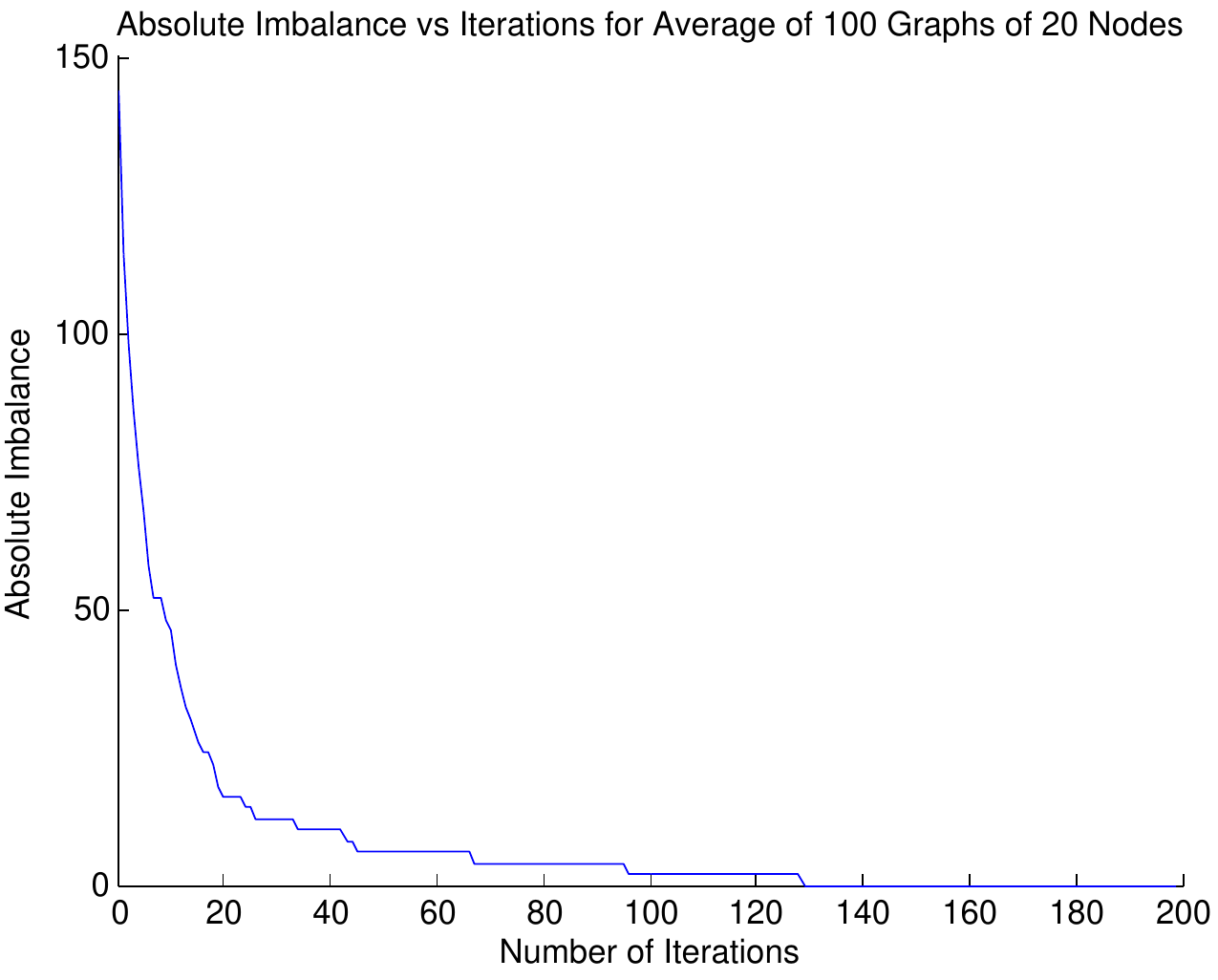}
\caption{Total (absolute) imbalance $\varepsilon[k]$ plotted against the number of iterations for the distributed algorithm (averaged over $100$ graphs of $20$ nodes each) in the case where the integer circulation conditions hold.}
\label{AVworking20}
\end{figure}


\begin{figure}[ht] 
\centering    
\includegraphics[width=0.60\textwidth]{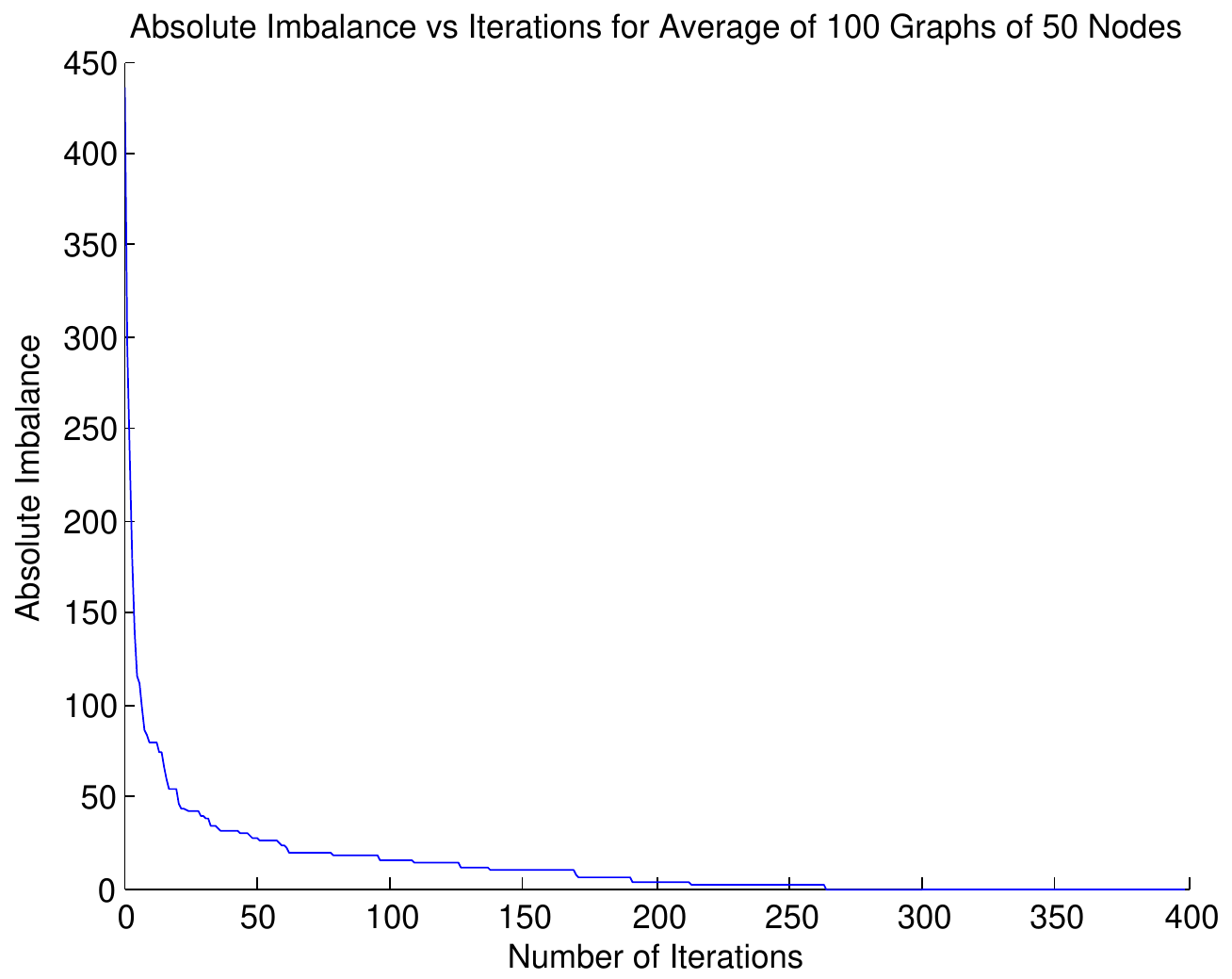}
\caption{Total (absolute) imbalance $\varepsilon[k]$ plotted against the number of iterations for the distributed algorithm (averaged over $100$ graphs of $50$ nodes each) in the case where the integer circulation conditions hold.}
\label{AVworking50}
\end{figure}


\section{Chapter Summary}
\label{summaryupperloweralgorithm}

In this chapter, we introduced and analyzed a novel distributed algorithm which achieves integer weight balancing in a multi-component system in the presence of lower and upper constraints on the edge weights.  
We analyzed its functionality, established its correctness and showed that it achieves integer weight balancing after a finite number of steps. 
We also demonstrated the operation, performance, and advantages of the proposed algorithm via various simulations.

\clearpage

\lhead{\emph{Weight Balancing under Link Constraints over Unreliable Communication}}

\chapter{Weight Balancing under \\ Link Capacity Constraints \\ over Unreliable Communication}
\label{constraintsbalancing_delays}

In this chapter, we present a novel distributed algorithm which deals with the problem of balancing a weighted digraph under link capacity constraints in the presence of time delays and packet drops over the communication links.
The algorithms presented in this chapter have appeared in \cite{2017:RikosHadj_CDC}.

This chapter is organized as follows. 
In Section~\ref{UpperLower_ModelDelaysPacket} we present the additional notation needed in this chapter and we recall the modeling of time delays and packet drops and the way they manifest themselves. 
In Section~\ref{ProbStatementUpperLower_delays} we present the problem formulation.
In Section~\ref{upperloweralgorithm_delays} we introduce a novel distributed algorithm which achieves integer weight balancing under link capacity constraints in the presence of time delays over the communication links. 
We present a formal description of the proposed distributed algorithm and show that as long as the conditions presented in Section~\ref{CircConditions} hold, then the proposed distributed algorithm converges to a weight balanced digraph after a finite number of iterations in the presence of bounded time delays over the communication links. 
In Section~\ref{triggeredalgDel_delays} we discuss an event-triggered operation of the proposed distributed algorithm and show that it results in a weight balanced digraph after a finite number of iterations in the presence of arbitrary (time-varying, inhomogeneous) but bounded time delays over the communication links.
In Section~\ref{packetalgDel_delays} we show that the proposed distributed algorithm is also able to converge (with probability one) to a weight balanced digraph in the presence of unbounded delays (packet drops).  
In Section~\ref{resultsupperloweralgorithm_delays} we present simulation results and comparisons and the chapter is concluded in Section~\ref{SummaryUpperLower_delays}.

%
%

\section{Modeling Time Delays and Packet Drops}
\label{UpperLower_ModelDelaysPacket}

In this chapter, we assume that a transmission from node $v_j$ to node $v_l$ at time step $k$ undergoes an \textit{a priori unknown} delay $\tau^{(j)}_{lj}[k]$ which is an integer that satisfies $0 \leq \tau^{(j)}_{lj}[k] \leq \overline{\tau}_{lj} \leq \infty$ (i.e., delays are bounded). 
The maximum delay is denoted by $\overline{\tau} = \max_{(v_l,v_j) \in \mathcal{E}}\overline{\tau}_{lj}$. 
In the weight balancing setting we consider, node $v_j$ is in charge of assigning the {\em actual} weight $f_{lj}[k]$ to each link $(v_l,v_j)$, and then transmits to node $v_l$ the amount of change $c^{(j)}_{lj}[k]$ it desires at time step $k$. 
Under the above delay model (which assumes bidirectional communication), node $v_l$ ($v_j$) receives the change amount $c^{(j)}_{lj}[k]$ ($c^{(l)}_{lj}[k]$), proposed by node $v_j$ ($v_l$) over the actual (perceived) weight $f_{lj}[k]$ ($f^{(p)}_{lj}[k]$), at time step $k + \tau^{(j)}_{lj}[k]$ ($k + \tau_{lj}^{(l)}[k]$).

\noindent
From the perspective of node $v_j$, the delayed weight change for link $(v_l,v_j)$, $\forall v_l \in N_j^+$, at time step $k$ is given by 
\begin{align} \label{delay_equation}
\overline{c}^{(l)}_{lj}[k]= \sum_{k_0 = k - \overline{\tau}, \ k_0 + \tau^{(l)}_{lj}[k_0] = k}^{k} c^{(l)}_{lj}[k_0],  
\end{align}
i.e., $\overline{c}^{(l)}_{lj}[k]$ is the sum of weight changes $c^{(l)}_{lj}$ sent from $v_l$ and seen from node $v_j$ at time step $k$.

Apart from bounded delays, \textit{unreliable} communication links in practical settings could also result in possible packet drops (i.e., unbounded delays) in the corresponding communication network. 
To model packet drops, we assume that a transmission on each link $(v_j, v_i)$ from node $v_i$ to node $v_j$ is { \em unreliable} which means that each particular edge may drop packets with some (non-total) probability. 
We assume independence between packet drops at different time steps or different links (or even different directions of the same link), so that, we can model a packet drop via a Bernoulli random variable:
\begin{equation}\label{dropsmodel}
Pr\{ x_k(j,i)=m \} = \left\{ \begin{array}{ll}
         q_{ji}, & \mbox{if $m = 0$,}\\
         1 - q_{ji}, & \mbox{if $m = 1$,}\end{array} \right.
\end{equation}
where $x_k(j,i)=1$ if the transmission from node $v_i$ to node $v_j$ at time step $k$ is successful.

We also define the matrix $Q = [q_{ji}]$ where $q_{ji}$ is the entry at the $j^{th}$ row and $i^{th}$ column; we take $q_{ji} = 1$ for $(v_j, v_i) \notin \mathcal{E}_u$ and can set $q_{ji} = 0$ if the link is reliable.
We establish that, despite the presence of packet drops, the proposed distributed algorithm converges, with probability one, to a weight balanced digraph after a finite number of iterations (as long as a feasible solution exists and $q_{ji} <1$ for all links $(v_j, v_i) \in \mathcal{E}_u$).

Note that, in this chapter, the integer weight $f_{ji}$ on edge $(v_j, v_i) \in \mathcal{E}$ is assigned by node $v_i$. 
More specifically, $f_{ji}$ is assigned by node $v_i$; due to possible (bounded or unbounded) time delays the perceived weight $f^{(p)}_{ji}$ on this link by node $v_j$ might be different. 
This means that each node will know exactly the weights on its outgoing edges $f_{lj}$ but only have access to { \em perceived} weights $f^{(p)}_{ji}$ on its incoming edges, which will be equal to $f_{lj}$ if node $v_j$ is able to successfully communicate with node $v_l$.
\begin{defn}\label{DEFpercnodebalance}
Given a weighted digraph $\mathcal{G}_d=(\mathcal{V},\mathcal{E},\mathcal{F})$, along with a perceived weight assignment $F_p=[f^{(p)}_{ji}]$, the total {\em perceived} in-weight of node $v_j$ is defined as $\mathcal{S}_j^{-(p)}$ and is defined as $\mathcal{S}_j^{-(p)} = \sum_{v_i \in \mathcal{N}_j^-} f^{(p)}_{ji}$.
\end{defn}
\begin{defn}\label{DEFnodebalance_del}
Given a weighted digraph $\mathcal{G}_d=(\mathcal{V},\mathcal{E},\mathcal{F})$ of order $n$, the {\em perceived weight imbalance} $x^{(p)}_j$ of node $v_j$ is $x^{(p)}_j = \mathcal{S}_j^{-(p)} - \mathcal{S}_j^+$ while the perceived total imbalance of digraph $\mathcal{G}_d$ is defined as $\varepsilon^{(p)} = \sum_{j=1}^{n} \vert x^{(p)}_j \vert$.
\end{defn}

\section{Problem Statement}
\label{ProbStatementUpperLower_delays}

We are given a strongly connected digraph $\mathcal{G}_d = (\mathcal{V}, \mathcal{E})$, as well as lower and upper limits $l_{ji}$ and $u_{ji}$ ($0 < l_{ji} \leq u_{ji}$, where $l_{ji}, u_{ji} \in \mathbb{R}$) on each each edge $(v_j, v_i) \in \mathcal{E}$. 
Considering that link transmissions undergo arbitrary, bounded or unbounded delays (i.e., packet drops), we want to develop a distributed algorithm that allows the nodes to iteratively adjust the integer weights on their edges so that they eventually obtain a set of integer weights $\{ f_{ji} \; | \; (v_j, v_i) \in \mathcal{E} \}$ that satisfy the following:
\begin{enumerate}
\item $ f_{ji} \in \mathbb{N}$ for each edge $(v_j,v_i) \in \mathcal{E}$;
\item $l_{ji} \leq f_{ji} \leq u_{ji}$ for each edge $(v_j,v_i) \in \mathcal{E}$;
\item $\mathcal{S}_j^+ = \mathcal{S}_j^- = \mathcal{S}_j^{-(p)}$ for each $v_j \in \mathcal{V}$.
\end{enumerate}
The distributed algorithm needs to respect the communication constraints imposed by the undirected graph $\mathcal{G}_u$ that corresponds to the given directed graph~$\mathcal{G}_d$.

We introduce and analyze a distributed algorithm that allows each node to assign integer weights to its outgoing links, so that the resulting weight assignment is balanced. 

%
%

\section{Distributed Algorithm for Weight Balancing \\ in the Presence of Time Delays}
\label{upperloweralgorithm_delays}

In this section we provide an overview of the distributed weight balancing algorithm operation; the formal description of the algorithm is provided in Algorithm~\ref{constralg1_delays}. 
The algorithm is iterative and operates by having, at each iteration, nodes with \textit{positive} perceived weight imbalance attempt to change the integer weights on both their incoming and/or outgoing edges so that they become weight balanced. 
We assume that each node is in charge of assigning the weights on its outgoing edges. 
More specifically, $f_{ji}$ is assigned by node $v_i$; due to possible time delays the perceived weight $f^{(p)}_{ji}$ on this link by node $v_j$ might be different. 
This means that each node will know exactly the weights on its outgoing edges but only have access to perceived weights on its incoming edges. 
We use $k$ to denote the iteration and index variables. 
For example, the {\em change amount} from node $v_j$ on the weight $f_{ji}[k]$ of edge $(v_j,v_i) \in \mathcal{E}$ at iteration $k$ will be denoted by $c_{ji}^{(j)}[k]$.

\textit{Note:} Each node $v_j$ can only calculate its {\em perceived weight imbalance} $x^{(p)}_j$ (as defined in Definition~\ref{DEFnodebalance_del}) at each iteration $k$. 
This means that it has no access to the total (or perceived total) imbalance of the digraph $\mathcal{G}_d$.

\begin{remark}
Note here that the integer weight $f_{lj}$ on edge $(v_l, v_j) \in \mathcal{E}$ is assigned by node $v_j$. 
Thus, node $v_j$ has access to the true weight $f_{lj}$ of edge $(v_l, v_j)$ while node $v_l$ has access to a {\em perceived} weight $f^{(p)}_{lj}$, which will be equal to $f_{lj}$ if node $v_j$ is able to successfully communicate with node $v_l$.
\end{remark}

We describe the operation of the iterative algorithm and establish that, if the necessary and sufficient Integer Circulation Conditions in Section~\ref{CircConditions} are satisfied, the algorithm completes after a finite number of iterations.

\textbf{Initialization.} At initialization, each node is aware of the feasible weight interval on each of its incoming and outgoing edges, i.e., node $v_j$ is aware of $l_{ji}, u_{ji}$ for each $v_i \in \mathcal{N}^-_j$ and $l_{lj}, u_{lj}$ for each $v_l \in \mathcal{N}^+_j$. 
Furthermore, the weights are initialized at the ceiling of the lower bound of the feasible interval, i.e., $f_{ji}[0] = \lceil l_{ji} \rceil$.
This initialization is always feasible but not critical and could be any integer value in the feasible weight interval $[l_{ji}, u_{ji}]$ (according to Section~\ref{CircConditions}) an integer always exists in the interval $[l_{ji}, u_{ji}]$).
Also each node $v_j$ chooses a unique order $P_{lj}^{(j)}$ and $P_{ji}^{(j)}$ for its outgoing links $(v_l,v_j)$ and incoming links $(v_j,v_i)$ respectively, such that $\{P_{lj}^{(j)} \; | \; v_l \in \mathcal{N}^+_j\} \cup \{P_{ji}^{(j)} \; | \; v_i \in \mathcal{N}^-_j\} = \{0,1,..., \mathcal{D}_j-1\}$.

\textbf{Iteration.} At each iteration $k \geq 0$, node $v_j$ is aware of the {\em perceived} integer weights on its incoming edges $\{ f^{(p)}_{ji}[k] \; | \: v_i \in \mathcal{N}^-_j \}$ and the (actual) weights on its outgoing edges $\{ f_{lj}[k] \; | \: v_l \in \mathcal{N}^+_j \}$, which allows it to calculate its {\em perceived} weight imbalance $x^{(p)}_j[k]$ according to Definition~\ref{DEFnodebalance_del}.

\noindent
{\em A. Selecting Desirable Weights.} Each node $v_j$ with positive {\em perceived} weight imbalance (i.e., $x^{(p)}_j[k] > 0$) attempts to change the weights on its incoming edges $\{ f_{ji}[k] \; | \; v_i \in \mathcal{N}^-_j \}$ and/or outgoing edges $\{ f_{lj}[k] \; | \; v_l \in \mathcal{N}^+_j \}$ in a way that drives its perceived weight imbalance $x^{(p)}_j[k+1]$ to zero (at least if no other changes are inflicted on the weights).  
No attempt to change weights is made if node $v_j$ has negative or zero perceived weight imbalance. 
Specifically, node $v_j$ attempts to add $+1$ (or subtract $-1$) to its outgoing (or incoming) integer weights one at a time, according to a predetermined (cyclic) order until its perceived weight imbalance becomes zero. 
If an outgoing (incoming) edge has reached its max (min) value (according to the feasible interval on that particular edge), then its weight does not change and node $v_j$ proceeds to change the next one according to the predetermined order, in a round-robin fashion. 
The desired weight change by node $v_j$ on edge $(v_j,v_i) \in \mathcal{E}$ at iteration $k$ will be denoted by $c_{ji}^{(j)}[k]$; similarly, the desired weight by node $v_j$ on edge $(v_l,v_j) \in \mathcal{E}$ at iteration $k$ will be denoted by $c_{lj}^{(j)}[k]$.

\textit{Note:} Next time node $v_j$ has positive {\em perceived} weight imbalance it continues increasing (decreasing) its outgoing (incoming) edges by $1$, one at a time, following the (cyclic) predetermined order starting from the edge it stopped the previous time it had positive weight imbalance.

\noindent
{\em B. Exchanging Desirable weights.} 
Once the nodes with positive {\em perceived} weight imbalance calculates the desirable weight change for each incoming $\{ c^{(j)}_{ji}[k] \; | \; v_i \in \mathcal{N}^-_j \}$ and outgoing $\{ c^{(j)}_{lj}[k] \; | \; v_l \in \mathcal{N}^+_j \}$ weight, they take the following steps in sequence:

\noindent
1) Node $v_j$ transmits the desirable weight change $c^{(j)}_{ji}[k]$ ($c^{(j)}_{lj}[k]$) to each in- (out-) neighbor $v_i$ ($v_l$).

\noindent
2) Node $v_j$ receives the delayed desired weight changes $\overline{c}^{(i)}_{ji}[k]$ ($\overline{c}^{(l)}_{lj}[k]$) from each in- (out-) neighbor $v_i$ ($v_l$).
If no weight change is received due to time delays, then node $v_j$ assumes that $\overline{c}^{(i)}_{ji}[k] = 0$ ($\overline{c}^{(l)}_{lj}[k] = 0$) for the corresponding incoming (outgoing) edge $(v_j, v_i)$ ($(v_l, v_j)$).

\noindent
3) It calculates its new outgoing ({\em perceived} incoming) weights $ f_{lj}[k+1] = f_{lj}[k] + c_{lj}^{(j)}[k] + \overline{c}_{lj}^{(l)}[k] $ ($f_{ji}^{(p)}[k+1] = f_{ji}^{(p)}[k] + c_{ji}^{(j)}[k] + \overline{c}_{ji}^{(i)}[k]$).
Then, the new outgoing ({\em perceived} incoming) weights are adjusted so that the new weight is projected onto the feasible interval $[l_{lj}, u_{lj}]$ ($[l_{ji}, u_{ji}]$) of the corresponding edge.
This (along with all the parameters involved) can be seen in Figure~\ref{nodes_exch_delays}.

\begin{remark}
Since the weight $f_{ji}$ on each edge $(v_j, v_i) \in \mathcal{E}$ affects positively the weight imbalance $x_j[k]$ of node $v_j$ and negatively the weight imbalance $x_i[k]$ of node $v_i$, we need to take into account the possibility that both nodes desire a change on the weight simultaneously. 
Thus, the proposed algorithm attempts to coordinate the weight change. 
The challenge however, is the fact that time delays may occur during transmissions (in either direction) while the nodes are trying to agree on a weight value.
\end{remark}

\begin{figure}[h]
\begin{center}
\includegraphics[width=0.7\columnwidth]{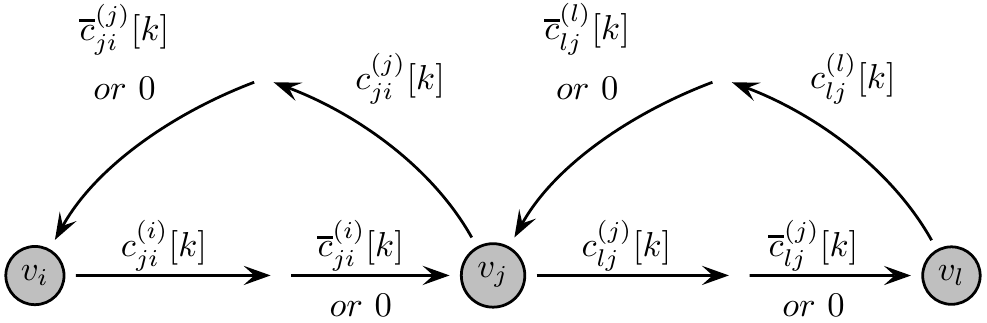}
\caption{Digraph where nodes exchange their desirable weights in the presence of time delays.}
\label{nodes_exch_delays}
\end{center}
\end{figure}

\subsection{Formal Description of Distributed Algorithm}
\label{formalupperloweralgorithm_delays}

A formal description of the proposed distributed algorithm is presented in Algorithm~\ref{constralg1_delays}.

\begin{varalgorithm}{6}
\caption{Distributed weight Balancing Algorithm in the Presence of Time Delays}
\textbf{Input} \\ 1) A strongly connected digraph $\mathcal{G}_d = (\mathcal{V}, \mathcal{E})$ with $n=|\mathcal{V}|$ nodes and $m=|\mathcal{E}|$ edges.\\ 2) $l_{ji},u_{ji}$ for every $(v_j, v_i) \in \mathcal{E}$. \\
\textbf{Initialization} \\ Set $k=0$; each node $v_j \in \mathcal{V}$ does:
\\1) It sets the weights on its {\em perceived} incoming and outgoing edge weights as 
$$
f^{(p)}_{ji}[0] = \lceil l_{ji} \rceil, \ \forall  v_i \in \mathcal{N}_j^-,
$$
$$
f_{lj}[0] = \lceil l_{lj} \rceil, \ \forall  v_l \in \mathcal{N}_j^+.
$$ 
2) It assigns a unique order to its outgoing and incoming edges as $P^{(j)}_{lj}$, for $v_l \in \mathcal{N}_j^+$ or $P^{(j)}_{ji}$, for $v_i \in \mathcal{N}_j^-$ (such that $\{P^{(j)}_{lj} \ | \ v_l \in \mathcal{N}_j^+\} \cup \{P^{(j)}_{ji} \ | \ v_i \in \mathcal{N}_j^-\} = \{0,1,...,\mathcal{D}_j -1 \}$).
\\
\textbf{Iteration} \\ For $k=0,1,2,\dots$, each node $v_j \in \mathcal{V}$ does the following:
\\ 1) It computes its \textit{perceived weight imbalance} as in Definition~\ref{DEFnodebalance_del}
$$
x^{(p)}_j[k] = \sum_{v_i \in \mathcal{N}_j^-} f_{ji}^{(p)}[k] - \sum_{v_l \in \mathcal{N}_j^+} f_{lj}[k]. 
$$
 2) If  $x^{(p)}_j[k] > 0$, it calculates the desired amount of change for the weights on its outgoing and incoming edges. Specifically, it increases (decreases) by $1$ the integer weights $f_{lj}[k]$ ($f_{ji}^{(p)}[k]$) of its outgoing (incoming) edges $v_l \in \mathcal{N}_j^+$ ($v_i \in \mathcal{N}_j^-$) one at a time, following the predetermined order $P^{(j)}_{lj}$ ($P^{(j)}_{ji}$) until its weight imbalance becomes zero (if an edge has reached its maximum value, its weight does not change and node $v_j$ proceeds in changing the next one according to the predetermined order). 
Then, it stores the desired change amount for each outgoing edge as $c_{lj}^{(j)}[k]$ and each incoming edge as $c_{ji}^{(j)}[k]$.
\\ 3) If  $x^{(p)}_j[k] > 0$, it transmits the desired weight change $c^{(j)}_{lj}[k]$ ($c^{(j)}_{ji}[k]$)  on each outgoing (incoming) edge.
\\ 4) It receives the (possibly delayed) desired weight change $\overline{c}^{(l)}_{lj}[k]$ ($\overline{c}^{(i)}_{ji}[k]$) from each outgoing (incoming) edge. 
[If no weight change is received due to time delays it assumes $\overline{c}^{(l)}_{lj}[k] = 0$ ($\overline{c}^{(i)}_{ji}[k] = 0$) for the corresponding outgoing (incoming) edge.]
\\ 5) It sets its new outgoing weights to be 
$$
f_{lj}[k+1] = f_{lj}[k] + c_{lj}^{(j)}[k] + \overline{c}_{lj}^{(l)}[k] ,
$$
and its new {\em perceived} incoming weights to be 
$$
f_{ji}^{(p)}[k+1] = f_{ji}^{(p)}[k] + c_{ji}^{(j)}[k] + \overline{c}_{ji}^{(i)}[k] .
$$
6) It adjusts the new outgoing weights according to the corresponding upper and lower weight constraints as
$$
f_{lj}[k+1] =\max \{ l_{lj}, \min\{u_{lj}, f_{lj}[k+1]\} \} ,
$$ 
and its new {\em perceived} incoming weights according to the corresponding upper and lower weight constraints as
$$
f_{ji}^{(p)}[k+1] =\max \{ l_{ji}, \min\{u_{ji}, f_{ji}^{(p)}[k+1]\} \} .
$$
7) It repeats (increases $k$ to $k+1$ and goes back to Step~1).
\label{constralg1_delays}
\end{varalgorithm}

\begin{remark}
According to the Integer Circulation Conditions in Section~\ref{CircConditions}, each node $v_j \in \mathcal{V}$ with positive {\em perceived} weight imbalance at iteration $k$ ($x^{(p)}_j[k] > 0$) will always be able to calculate a weight assignment for its incoming and outgoing edge weights so that its {\em perceived} weight imbalance becomes zero (at least if no other changes are inflicted on the weights of its incoming or outgoing links). 
This means that the selection of desirable weights in Algorithm~\ref{constralg1_delays} is always {\em feasible}.
\end{remark}

\begin{remark}\label{perceivedSmaller_del}
It is important to note here that the total {\em perceived} in-weight $\mathcal{S}_j^{-(p)}$ of node $v_j$ might be affected from possible time delays at Step~4 of Algorithm~\ref{constralg1_delays}. 
Specifically, if transmissions are affected from possible time delays then $v_j$ sets $f_{ji}^{(p)}[k+1] = f_{ji}^{(p)}[k] + c_{ji}^{(j)}[k]$ where $c_{ji}^{(j)}[k] < 0$ (since nodes only attempt to make changes on the weights if their perceived balance is positive, node $v_i$ will only attempt to increases the weight $f_{ji}[k]$ of edge $(v_j, v_i)$).
This means that during the execution of Algorithm~\ref{constralg1_delays} we have $f_{ji}^{(p)}[k] \leq f_{ji}[k]$ for each edge $(v_j,v_i) \in \mathcal{E}$, at each time step $k$.
\end{remark}

\begin{remark}\label{PROPupperlower}
The weight adjustment in Algorithm~\ref{constralg1_delays} signifies that after iteration $k$ of the proposed distributed algorithm, once node $v_j$ calculates the desired weight changes for its incoming (outgoing) edges $c_{ji}^{(j)}[k]$ ($c_{lj}^{(j)}[k]$), $\forall \ v_i \in \mathcal{N}_j^-$ ($\forall \ v_l \in \mathcal{N}_j^+$), it receives the delayed desired weight changes from its in- (out-) neighbors $\overline{c}_{ji}^{(i)}[k]$ ($\overline{c}_{lj}^{(l)}[k]$). 
Then, the new weight of edge $(v_j,v_i)$ ($(v_l,v_j)$) will be $f^{(p)}_{ji}[k+1] = f^{(p)}_{ji}[k] + \overline{c}_{ji}^{(i)}[k] + c_{ji}^{(j)}[k]$ ($f_{lj}[k+1] = f_{lj}[k] + \overline{c}_{lj}^{(l)}[k] + c_{lj}^{(j)}[k]$). 
According to Step~6 of the proposed algorithm we always have that $0 < l_{ji} \leq f^{(p)}_{ji}[k+1] \leq u_{ji}$ and $0 < l_{lj} \leq f_{lj}[k+1] \leq u_{lj}$. 
\end{remark}

\subsection{Proof of Algorithm Completion}
\label{analysisalgDel_delays}

In this section we analyze the functionality of the distributed algorithm and we prove that it solves the weight balancing problem in the presence of arbitrary (time-varying, inhomogeneous) but bounded time delays that may appear during the information exchange between agents in the system.

\begin{prop}
\label{PROP2_del}
Consider the problem formulation described in Section~\ref{ProbStatementUpperLower_delays}. 
Let $\mathcal{V}^-[k] \subset \mathcal{V}$ be the set of nodes with negative weight imbalance at iteration $k$, i.e., $\mathcal{V}^-[k] = \{ v_j \in \mathcal{V} \; | \; x_j[k] < 0 \}$.  
During the execution of Algorithm~\ref{constralg1_delays}, we have that
$$
\mathcal{V}^-[k+1] \subseteq \mathcal{V}^-[k].
$$
\end{prop}

\begin{proof}
We will first argue that nodes with nonnegative {\em perceived} weight imbalance at iteration $k$ can never reach negative perceived weight imbalance at iteration $k+1$. 
Combining this with the fact that the perceived weight imbalance of a node is always below its actual weight imbalance, we establish the proof of the proposition.

Consider a node $v_j$ with a nonnegative perceived weight imbalance $x^{(p)}_j[k] \geq 0$ (from Remark~\ref{perceivedSmaller_del}, since $x^{(p)}_j[k] \leq x_j[k]$, $\forall \ k \geq 0$, we have that also $x_j[k] \geq 0$). \\
\noindent
We analyze below the following two cases:
\begin{enumerate}
\item at least one neighbor of node $v_j$ has positive perceived weight imbalance,
\item all neighbors of node $v_j$ have negative or zero perceived weight imbalance.
\end{enumerate}
In both cases, since $x^{(p)}_j[k] \geq 0$, node $v_j$ will attempt to change the weights of (some of) its incoming and outgoing edges. 
Specifically, node $v_j$ will calculate the desirable weight change $c^{(j)}_{ji}[k]$ ($c^{(j)}_{lj}[k]$) for its incoming (outgoing) edges $(v_j, v_i)$ ($(v_l, v_j)$) where $v_i \in  \mathcal{N}_j^-$ ($v_l \in  \mathcal{N}_j^+$).
Then, 
it transmits the desired weight change $c^{(j)}_{ji}[k]$ ($c^{(j)}_{lj}[k]$) to its incoming (outgoing) edges $(v_j, v_i)$ ($(v_l, v_j)$) where $v_i \in  \mathcal{N}_j^-$ ($v_l \in  \mathcal{N}_j^+$).
In the first case, we have (i) $x^{(p)}_i[k] > 0$ for some $v_i \in  \mathcal{N}_j^-$, or (ii) $x^{(p)}_l[k] > 0$ for some $v_l \in  \mathcal{N}_j^+$.

For (i) we have that during iteration $k$ of Algorithm~\ref{constralg1_delays}, the incoming edge weights of $v_j$ might change by its in-neighbors (i.e., the weight of an incoming edge $(v_j,v_i)$ might be increased to be equal to $f_{ji}[k+1] = f_{ji}[k] + c_{ji}^{(i)}[k]$ for some $v_i \in \mathcal{N}_j^-$).
In this case, since the transmission of $c_{ji}^{(i)}[k]$ from $v_i$ to $v_j$ might suffer time delay, we have that $v_j$ sets its outgoing weights to be $ f_{lj}[k+1] = f_{lj}[k] + c_{lj}^{(j)}[k] $, and its {\em perceived} incoming weights to be $ f_{ji}^{(p)}[k+1] = f_{ji}^{(p)}[k] + c_{ji}^{(j)}[k] $. Thus, we have that $x^{(p)}_j[k+1] = 0$. 
[Note that, after $\tau^{(i)}_{ji}[k]$ time steps (during the iteration $k + \tau^{(i)}_{ji}[k]$) node $v_j$ will receive the desired weight change $c_{ji}^{(i)}[k]$ which was sent from node $v_i$ at time step $k$. 
Then it will update its its {\em perceived} incoming weights to be $ f_{ji}^{(p)}[k+\tau^{(i)}_{ji}[k]+1] = f_{ji}^{(p)}[k+\tau^{(i)}_{ji}[k]] + \overline{c}_{ji}^{(i)}[k+\tau^{(i)}_{ji}[k]] $, which means that $ x^{(p)}_j[k+\tau^{(i)}_{ji}[k]+1] > 0 $.]
As a result, for (i) we have that nonnegative {\em perceived} weight imbalance of node $v_j$ at iteration $k$ remains nonnegative at iteration $k+1$.

For (ii) we have that the outgoing edge weights of $v_j$ might change by its out-neighbors $v_l \in  \mathcal{N}_j^+$ and it can be argued in a similar manner.

In the second case, we have $x^{(p)}_i[k] \leq 0$ for every $v_i \in  \mathcal{N}_j^-$, and $x^{(p)}_l[k] \leq 0$ for every $v_l \in  \mathcal{N}_j^+$. 
This means that the neighbors of $v_j$ will not attempt to change the weights of its incoming and outgoing edges. 
As a result, since $v_j$ will transmit its desired weight changes and then set its outgoing weights to be $ f_{lj}[k+1] = f_{lj}[k] + c_{lj}^{(j)}[k] $ and its {\em perceived} incoming weights to be $ f_{ji}^{(p)}[k+1] = f_{ji}^{(p)}[k] + c_{ji}^{(j)}[k] $, we have that $x^{(p)}_j[k+1] = 0$. 

As a result we have that during iteration $k$ of Algorithm~\ref{constralg1_delays}, nodes with nonnegative {\em perceived} weight imbalance can never reach negative {\em perceived} weight imbalance at iteration $k+1$.
From Remark~\ref{perceivedSmaller_del}, since $x^{(p)}_j[k] \leq x_j[k]$, $\forall \ k \geq 0$, we have that also nodes with nonnegative weight imbalance can never reach negative weight imbalance, thus establishing the proof of the proposition.
\end{proof}

\begin{prop}
\label{PROP3_del}
Consider the problem formulation described in Section~\ref{ProbStatementUpperLower_delays}. 
During the execution of Algorithm~\ref{constralg1_delays}, it holds that
$$
0 \leq \varepsilon[k+1] \leq \varepsilon[k] \; , \; \; \forall k \geq 0 \; ,
$$
where $\varepsilon[k] \geq 0$ is the total imbalance of the network at iteration~$k$ (see Definition~\ref{defn:totalim}).
\end{prop}

\begin{proof} 
From the third  statement of Proposition~\ref{constrPROP1}, we have $\varepsilon[k+1] = 2 \sum_{v_j \in \mathcal{V}^-[k+1]} \vert x_j[k+1] \vert$ and $\varepsilon[k] = 2 \sum_{v_j \in \mathcal{V}^-[k]} \vert x_j[k] \vert$, whereas from Proposition~\ref{PROP2_del}, we have $\mathcal{V}^-[k+1] \subseteq \mathcal{V}^-[k]$.

Consider a node $v_j \in \mathcal{V}^-[k]$ with weight imbalance $x_j[k]<0$ (obviously we have that also $x^{(p)}_j[k]<0$ since $x^{(p)}_j[k] \leq x_j[k]$, $\forall \ k \geq 0$ from Remark~\ref{perceivedSmaller_del}). \\
\noindent
We analyze below the following two cases:
\begin{enumerate}
\item all neighbors of node $v_j$ have negative or zero perceived weight imbalance,
\item at least one neighbor of node $v_j$ has positive perceived weight imbalance.
\end{enumerate}
In both cases, node $v_j$ will not make any changes on its edges. 
In the first case, the weight imbalance of node $v_j$ will not change (i.e., $x_j[k+1] = x_j[k] < 0$). 
In the second case, we have (i) $x^{(p)}_i[k] \geq 0$ for some $v_i \in  \mathcal{N}_j^-$, or (ii) $x^{(p)}_l[k] \geq 0$ for some $v_l \in  \mathcal{N}_j^+$. 

For (i) we have that during the iteration $k$ of Algorithm~\ref{constralg1_delays}, the incoming edge weights of $v_j$ might change by its in-neighbors (i.e., the weight of an incoming edge $(v_j,v_i)$ might be increased to be equal to $f_{ji}[k+1] = f_{ji}[k] + c_{ji}^{(i)}[k]$ for some $v_i \in \mathcal{N}_j^-$).
In this case (regardless of whether we have a delay during the transmission of $c_{ji}^{(i)}[k]$ from $v_i$ to $v_j$) we have $\varepsilon[k+1] \leq \varepsilon[k]$ (using the third statement in Proposition~\ref{constrPROP1}).
For (ii) we have that during iteration $k$ of Algorithm~\ref{constralg1_delays}, the out-neighbor of $v_j$ might transmit the desired change amount of the outgoing edge weights to node $v_j$.
In this case, if the transmission of $c_{lj}^{(l)}[k]$ is delayed, then the weight imbalance of $v_j$ will not change (i.e., $x_j[k+1] = x_j[k] < 0$), but when $v_j$ receives $\overline{c}_{lj}^{(l)}[k+\tau^{(l)}_{lj}[k]]$ then the weight imbalance of node $v_j$ will satisfy $x_j[k+1] \geq x_j[k]$.
As a result, for both cases, we have $\varepsilon[k+1] \leq \varepsilon[k]$ (using the third statement in Proposition~\ref{constrPROP1}).
\end{proof}

\begin{prop}
\label{PROP4_del}
Consider the problem formulation described in Section~\ref{ProbStatementUpperLower_delays} where the Integer Circulation Conditions in Section~\ref{CircConditions} are satisfied. 
Algorithm~\ref{constralg1_delays} balances the weights in the graph in a finite number of steps (i.e., $\exists \ k_0$ so that $\forall k \geq k_0$, $f_{ji}[k_0] = f_{ji}[k]$, $\forall (v_j,v_i) \in \mathcal{E}$ and $x_j[k] = x_j[k_0] =0$, $\forall \ v_j \in \mathcal{V}$). 
\end{prop}

\begin{proof} 
By contradiction, suppose Algorithm~\ref{constralg1_delays} runs for an infinite number of iterations and its total imbalance remains positive (i.e., $\varepsilon[k]>0$ for all $k$). 
During the execution of the proposed distributed balancing algorithm, transmissions on each communication link $(v_l,v_j) \in \mathcal{E}$ are affected by arbitrary (time-varying and inhomogeneous) bounded time delays.
We have that the delays that affect transmissions on each link $(v_l,v_j) \in \mathcal{E}$ are bounded (i.e., $0 \leq \tau^{(j)}_{lj}[k] \leq \overline{\tau}_{lj} \leq \infty$).
Thus, the packets transmitted on each link $(v_l,v_j) \in \mathcal{E}$ will eventually reach the corresponding node after a finite number of steps.

Suppose now that Algorithm~\ref{constralg1_delays} runs for an infinite number of iterations and by contradiction its total imbalance remains positive (i.e., $\varepsilon[k]>0$ for all $k$).
This means that always (at each $k$) there will exist at least one node with positive weight imbalance and thus the proof of this Proposition becomes identical to the proof of Proposition~\ref{constrPROP4} (because, as argued above, its perceived imbalance will eventually become positive, once all transmitted packets are received). 

As a result, we have that if the Integer Circulation Conditions in Section~\ref{CircConditions} hold, the total imbalance $\varepsilon[k]$ decreases after a finite number of iterations, and Algorithm~\ref{constralg1_delays} results in a weight-balanced digraph after a finite number of iterations. 
\end{proof}

\section{Extension to Event-Triggered Operation}
\label{triggeredalgDel_delays}

Motivated by the need to reduce energy consumption, communication bandwidth, network congestion, and/or processor usage, many researchers have considered the use of event-triggered communication and control \cite{dimarogonas2012distributed, 2014:nowzari_cortes}. In this section, we discuss an event-triggered operation of the proposed distributed algorithm  where each agent autonomously decides when communication and control updates should occur so that the resulting network executions still result in a weight-balanced digraph after a finite number of steps in the presence of arbitrary (time-varying, inhomogeneous) but bounded time delays that might affect link transmissions. More specifically, following the proposed event-triggered strategy, we can prove that (i) all nodes eventually stop transmitting, and (ii) the proposed distributed algorithm is able to obtain a set of weights that balance the corresponding digraph after a finite number of iterations.

\subsection{Formal Description of Distributed Algorithm}
\label{formalupperloweralgorithm_event}

A formal description of the algorithm's event-triggered operation is presented in Algorithm~\ref{constralg1_event}.

\begin{varalgorithm}{7}
\caption{Event-Triggered Distributed Balancing in the Presence of Time Delays}
\textbf{Input} \\ (Inputs are the same as Algorithm~\ref{constralg1_delays}). \\
\textbf{Initialization} \\ For $k=0$, each node $v_j \in \mathcal{V}$ does the following: \\ (Steps~1, 2 are the same as Algorithm~\ref{constralg1_delays}).
\\ 3) (Same as Iteration-Step~1 in Algorithm~\ref{constralg1_delays}).
\\ 4) (Same as Iteration-Step~2 in Algorithm~\ref{constralg1_delays}).
\\ 5) (Same as Iteration-Step~3 in Algorithm~\ref{constralg1_delays}).
\\ 6) If $x^{(p)}_j[0] > 0$, it sets its outgoing weights to be 
$$
f_{lj}[1] = f_{lj}[0] + c_{lj}^{(j)}[0] ,
$$
and its new {\em perceived} incoming weights to be 
$$
f_{ji}^{(p)}[1] = f_{ji}^{(p)}[0] + c_{ji}^{(j)}[0] .
$$
\textbf{Iteration} \\ For $k=1,2,3,\dots$, each node $v_j \in \mathcal{V}$ does the following:
\\ 1) \textit{Event triggered condition:} If no weight change is received due to time delays then node~$v_j$ skips Steps~2, 3, 4, 5, 6, and~7 below; otherwise (event triggered condition) it receives the delayed desired weight change $\overline{c}^{(l)}_{lj}[k]$ ($\overline{c}^{(i)}_{ji}[k]$) from each outgoing (incoming) edge and performs the steps below.
\\ 2) It sets its outgoing weights to be 
$$
f_{lj}[k+1] = f_{lj}[k+1] + \overline{c}_{lj}^{(l)}[k] ,
$$
and its new {\em perceived} incoming weights to be 
$$
f_{ji}^{(p)}[k+1] = f_{ji}^{(p)}[k+1] + \overline{c}_{ji}^{(i)}[k] .
$$
3) (Same as Step~1 in Algorithm~\ref{constralg1_delays}).
\\ 4) (Same as Step~2 in Algorithm~\ref{constralg1_delays}).
\\ 5) (Same as Step~3 in Algorithm~\ref{constralg1_delays}).
\\ 6) If $x^{(p)}_j[k] > 0$, it sets its outgoing weights to be 
$$
f_{lj}[k+1] = f_{lj}[k] + c_{lj}^{(j)}[k] ,
$$
and its new {\em perceived} incoming weights to be 
$$
f_{ji}^{(p)}[k+1] = f_{ji}^{(p)}[k] + c_{ji}^{(j)}[k] .
$$
7) (Same as Step~6 in Algorithm~\ref{constralg1_delays}).
\\ 8) It repeats (increases $k$ to $k+1$ and goes back to Step~1).
\label{constralg1_event}
\end{varalgorithm}

\subsection{Proof of Algorithm Completion}
\label{analysisalgDel}

\begin{prop}\label{eventconverge}
Consider the problem formulation described in Section~\ref{ProbStatementUpperLower_delays} where the integer circulation conditions in Section~\ref{CircConditions} are satisfied. 
Algorithm~\ref{constralg1_event} balances, the weights in the graph in a finite number of steps [even in the presence of bounded delays] (i.e., $\exists \ k_0$ so that $\forall k \geq k_0$, $f_{ji}[k_0] = f_{ji}[k]$, $\forall (v_j,v_i) \in \mathcal{E}$ and $x_j[k] = x_j[k_0] =0$, $\forall \ v_j \in \mathcal{V}$). 
\end{prop}

\begin{proof} 
The event-triggered operation of Algorithm~\ref{constralg1_event} is identical to the operation of Algorithm~\ref{constralg1_event} with delays if we assume that in the latter algorithm all transmissions of desired weight changes suffer the maximum possible delay. 
As a result, since the operation of both algorithms is identical\footnote{The operation is identical under different delays in each case.}, we have that Algorithm~\ref{constralg1_event} will converge to a set of weights that form a weight-balanced digraph after a finite number of steps. 
Also, since $\exists \ k_0$ so that $\forall k \geq k_0$, $f_{ji}[k_0] = f_{ji}[k]$, $\forall (v_j,v_i) \in \mathcal{E}$ and $x_j[k] = x_j[k_0] =0$, $\forall \ v_j \in \mathcal{V}$, from Step~1 of Algorithm~\ref{constralg1_event}, we can see that all nodes eventually stop transmitting. 
\end{proof}


\section{Distributed Algorithm for Weight Balancing \\ in the Presence of Packet Dropping Links}
\label{packetalgDel_delays}

In this section we provide an overview of the distributed weight balancing algorithm operation; the formal description of the algorithm is provided in Algorithm~\ref{constralg1_packet}. 
The algorithm is iterative and operates by having, at each iteration, nodes with \textit{positive} perceived weight imbalance attempt to change the integer weights on both their incoming and/or outgoing edges so that they become weight balanced. 
Again, we assume that each node is in charge of assigning the weights on its outgoing edges (i.e., $f_{ji}$ is assigned by node $v_i$; due to possible packet drops the perceived weight $f^{(p)}_{ji}$ on this link by node $v_j$ might be different) which means that each node will know exactly the weights on its outgoing edges but only have access to perceived weights on its incoming edges. 

Note that the operation of Algorithm~\ref{constralg1_packet} is similar to Algorithm~\ref{constralg1_delays} with the main difference being that each node is required to calculate and transmit the desirable \textit{weights} (and not the desired change amounts) for its incoming and outgoing edges.

We describe the operation of iterative algorithm and establish that, if the necessary and sufficient integer circulation conditions for the existence of a set of \textit{integer} weights that balance the given digraph are satisfied, the algorithm completes, almost surely, after a finite number of iterations.

\textbf{Initialization.} Same as Algorithm~\ref{constralg1_delays}.

\textbf{Iteration.} At each iteration $k \geq 0$, node $v_j$ is aware of the {\em perceived} integer weights on its incoming edges $\{ f^{(p)}_{ji}[k] \; | \: v_i \in \mathcal{N}^-_j \}$ and the (actual) weights on its outgoing edges $\{ f_{lj}[k] \; | \: v_l \in \mathcal{N}^+_j \}$, which allow it to calculate its {\em perceived} weight imbalance $x^{(p)}_j[k]$ according to Definition~\ref{DEFnodebalance_del}.

\noindent
{\em A. Selecting Desirable weights:} Each node $v_j$ with positive {\em perceived} weight imbalance (i.e., $x^{(p)}_j[k] > 0$) attempts to change the weights on its incoming edges $\{ f_{ji}[k] \; | \; v_i \in \mathcal{N}^-_j \}$ and/or outgoing edges $\{ f_{lj}[k] \; | \; v_l \in \mathcal{N}^+_j \}$ in a way that drives its perceived weight imbalance $x^{(p)}_j[k+1]$ to zero (at least if no other changes are inflicted on the weights).  
No attempt to change weights is made if node $v_j$ has negative or zero perceived weight imbalance. 
Specifically, node $v_j$ attempts to add $+1$ (or subtract $-1$) to its outgoing (or incoming) integer weights one at a time, according to a predetermined (cyclic) order until its perceived weight imbalance becomes zero. 
If an outgoing (incoming) edge has reached its max (min) value (according to the feasible interval on that particular edge), then its weight does not change and node $v_j$ proceeds in changing the next one according to the predetermined order. 
The desired weight by node $v_j$ on edge $(v_j,v_i) \in \mathcal{E}$ at iteration $k$ will be denoted by $f_{ji}^{(j)}[k]$; similarly, the desired weight by node $v_j$ on edge $(v_l,v_j) \in \mathcal{E}$ at iteration $k$ will be denoted by $f_{lj}^{(j)}[k]$.

\textit{Note:} Next time node $v_j$ has positive {\em perceived} weight imbalance it continues increasing (decreasing) its outgoing (incoming) edges by $1$, one at a time following the (cyclic) predetermined order starting from the edge it stopped the previous time it had positive balance.

\noindent
{\em B. Exchanging Desirable weights:} 
Once the nodes with positive {\em perceived} weight imbalance calculate the desirable incoming $\{ f^{(j)}_{ji}[k] \; | \; v_i \in \mathcal{N}^-_j \}$ and outgoing $\{ f^{(j)}_{lj}[k] \; | \; v_l \in \mathcal{N}^+_j \}$ weights, they take the following steps in sequence: 

\noindent
1) Node $v_j$ transmits (receives) the calculated desirable weights $f^{(j)}_{ji}[k]$ ($f^{(l)}_{lj}[k]$) to (from) their in- (out-) neighbor $v_i$ ($v_l$).
[Nodes with non-positive perceived weight imbalance simply transmit the values $f^{(p)}_{ji}[k]$.]

\noindent
2) If no weight is received from out-neighbor $v_l$ (due to a packet drop), then node $v_j$ assumes that $f^{(l)}_{lj}[k] = f_{lj}[k]$ for the corresponding outgoing edge $(v_l, v_j)$ which suffered a packet drop on the transmission on the reverse link from node $v_l$ to node $v_j$.
Then it calculates its new outgoing weights $ f_{lj}[k+1] = f_{lj}^{(l)}[k] + f_{lj}^{(j)}[k] - f_{lj}[k] $  (projected onto the feasible interval $[l_{lj}, u_{lj}]$) and it transmits them to each corresponding out-neighbor $v_l \in \mathcal{N}^+_j$. 

\noindent
3) It receives the new incoming weights $\{ f^{(p)}_{ji}[k+1] \; | \; v_i \in \mathcal{N}^-_j \}$ from each corresponding in-neighbor. 
If no weight is received then node $v_j$ assumes that $f^{(p)}_{ji}[k+1] = f^{(j)}_{ji}[k]$ for the corresponding incoming edge $(v_j, v_i)$ which suffered a packet drop.
This (along with all the parameters involved) can be seen in Figure~\ref{nodes_exchV2}.


\begin{figure}[h]
\begin{center}
\includegraphics[width=0.7\columnwidth]{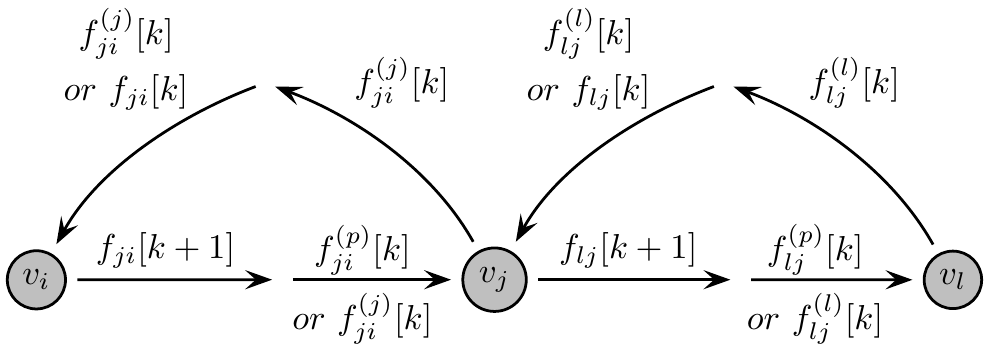}
\caption{Digraph where nodes exchange their desirable weights.}
\label{nodes_exchV2}
\end{center}
\end{figure}

\begin{remark}
The different weights that the nodes are exchanging (and what happens in the case of a packet drop) are shown in Fig.~\ref{nodes_exchV2}.
Specifically, at each iteration $k$, each node $v_j$ calculates its {\em perceived} weight imbalance $x^{(p)}_j[k]$ and if its positive then it calculates the desired weights for its incoming and outgoing edges ($f_{ji}^{(j)}[k]$ and $f_{lj}^{(j)}[k]$ respectively).
Then, it sends the desired incoming weight $f_{ji}^{(j)}[k]$ to its in-neighbors $v_i \in \mathcal{N}_j^-$ while it receives (if there is no packet drop) the desired incoming weight $f_{lj}^{(l)}[k]$ from its out-neighbors $v_l \in \mathcal{N}_j^+$.
If there is a packet drop, it assumes that $f_{lj}^{(l)}[k] = f_{lj}[k]$.
It then calculates the new weight for its outgoing edges $f_{lj}[k+1]$, $\forall \ v_l \in \mathcal{N}_j^+$, and transmits them to $v_l \in \mathcal{N}_j^+$, while it receives the new weights from its in-neighbors (if no incoming weight is received then it assumes $f_{ji}^{(p)}[k+1] = f_{ji}^{(j)}[k]$, otherwise  $f_{ji}^{(p)}[k+1] = f_{ji}[k+1]$).
\end{remark}

Depending on the possible packet drops that might occur during the exchange of the desirable weights, we have the following four cases:
\begin{enumerate}
\item $f^{(j)}_{ji}[k]$ is dropped,
\item both $f^{(j)}_{ji}[k]$ and $f_{ji}[k+1]$ are dropped,
\item $f_{ji}[k+1]$ is dropped,
\item no packet is dropped.
\end{enumerate}
For the first two cases, the new weight on edge $(v_j,v_i) \in \mathcal{E}$ is taken to be $f_{ji}[k+1] = f^{(i)}_{ji}[k]$ where $l_{ji} \leq f^{(i)}_{ji}[k] \leq u_{ji}$ (the difference in the two cases is that in the second case the perceived value of the weight at node $v_j$ is $f_{ji}^{(p)}[k+1] = f_{ji}^{(j)}[k]$. 

\noindent
For the third and fourth cases, the new weight on edge $(v_j,v_i) \in \mathcal{E}$ is taken to be $f_{ji}[k+1] = [f^{(i)}_{ji}[k+1] + f^{(j)}_{ji}[k+1] - f_{ji}[k]]_{\lceil l_{ji} \rceil}^{\lfloor u_{ji} \rfloor}$ (where $[x]_{\lceil l_{ji} \rceil}^{\lfloor u_{ji} \rfloor}$ denotes the projection onto the interval).
The difference in the two cases is that in the third case $f_{ji}^{(p)}[k+1] = f_{ji}^{(j)}[k]$, while in the fourth $f_{ji}^{(p)}[k+1] = f^{(i)}_{ji}[k+1] + f^{(j)}_{ji}[k+1] - f_{ji}[k]$.


\subsection{Formal Description of Distributed Algorithm}
\label{formalupperloweralgorithm_packet}

A formal description of the proposed distributed algorithm is presented in Algorithm~\ref{constralg1_packet}.

\begin{varalgorithm}{8}
\caption{Distributed weight Balancing Algorithm in the Presence of Packet Drops}
\textbf{Input} \\ (Inputs are the same as Algorithm~\ref{constralg1_delays}). \\
\textbf{Initialization} \\ (Steps~1, 2 are the same as Algorithm~\ref{constralg1_delays}).
\\
\textbf{Iteration} \\ For $k=0,1,2,\dots$, each node $v_j \in \mathcal{V}$ does the following:
\\ 1) It computes its \textit{perceived weight imbalance} as in Definition~\ref{DEFnodebalance_del}
$$
x^{(p)}_j[k] = \sum_{v_i \in \mathcal{N}_j^-} f_{ji}^{(p)}[k] - \sum_{v_l \in \mathcal{N}_j^+} f_{lj}[k]. 
$$
\\ 2) If  $x^{(p)}_j[k] > 0$, it increases (decreases) by $1$ the integer weights $f_{lj}[k]$ ($f_{ji}^{(p)}[k]$) of its outgoing (incoming) edges $v_l \in \mathcal{N}_j^+$ ($v_i \in \mathcal{N}_j^-$) one at a time, following the predetermined order $P^{(j)}_{lj}$ ($P^{(j)}_{ji}$) until its weight imbalance becomes zero (if an edge has reached its maximum value, its weight does not change and node $v_j$ proceeds in changing the next one according to the predetermined order). 
It stores the desirable weights on each incoming edge as $f_{ji}^{(j)}[k]$ and each outgoing edge as $f_{lj}^{(j)}[k]$.
\\ 3) If $x^{(p)}_j[k] \leq 0$, it sets $f_{lj}^{(j)}[k] = f_{lj}[k]$ (and $f_{ji}^{(j)}[k] = f_{ji}^{(p)}[k]$) for its outgoing (incoming) edges in $\mathcal{E}_j^+$ ($\mathcal{E}_j^-$).
\\ 4) It transmits the new weight $f_{ji}^{(j)}[k]$ on each incoming edge.
\\ 5) It receives the new weight $f_{lj}^{(l)}[k]$ from each outgoing edge (if no weight was received then it assumes that $f_{lj}^{(l)}[k] = f_{lj}[k]$).
\\ 6) It sets its outgoing weights to be $f_{lj}[k+1] = f_{lj}^{(l)}[k] + f_{lj}^{(j)}[k] - f_{lj}[k]$.
\\ 7) It transmits the new weight $f_{lj}[k+1]$ on each outgoing edge.
\\ 8) It receives new weight $f_{ji}^{(p)}[k+1]$ from each incoming edge (if no weight is received then it assumes that $f_{ji}^{(p)}[k+1] = f_{ji}^{(j)}[k]$). 
\\ 9) It repeats (increases $k$ to $k+1$ and goes back to Step~1).
\label{constralg1_packet}
\end{varalgorithm}

\vspace{0.1cm}

\begin{remark}
According to the integer circulation conditions, each node $v_j \in \mathcal{V}$ with positive {\em perceived} weight imbalance at iteration $k$ ($x^{(p)}_j[k] > 0$) will always be able to calculate a weight assignment for its incoming and outgoing edge weights so that its {\em perceived} weight imbalance becomes zero (at least if no other changes are inflicted on the weights of its incoming or outgoing links). 
This means that the selection of desirable weights in Algorithm~\ref{constralg1_packet} is always {\em feasible}.
\end{remark}

\begin{remark}\label{perceivedSmaller}
It is important to note here that the total {\em perceived} in-weight $\mathcal{S}_j^{-(p)}$ of node $v_j$ might be affected from possible packet drops at Step~7 of Algorithm~\ref{constralg1_packet}. 
Specifically, if a packet drop occurs; then $v_j$ assumes $f_{ji}^{(p)}[k+1] = f_{ji}^{(j)}[k]$ where $f_{ji}^{(j)}[k] \leq f_{ji}[k+1]$ (since nodes only attempt to make changes on the weights if their perceived weight imbalance is positive, node $v_i$ can only increase the weight of edge $(v_j, v_i)$).
This means that during the execution of Algorithm~\ref{constralg1_packet} we have $f_{ji}^{(p)}[k] \leq f_{ji}[k]$ for each edge $(v_j,v_i) \in \mathcal{E}$, at each time step $k$.
\end{remark}

\subsection{Proof of Algorithm Completion}
\label{analysisalgDel_packet}

In this section, we show that, as long as the Integer Circulation Conditions in Section~\ref{CircConditions} hold, then the total imbalance $\varepsilon[k]$ in Definition~\ref{defn:totalim} goes to zero after a finite number of iterations of Algorithm~\ref{constralg1_packet}. 
This implies that the weight imbalance $x_j[k]$ for each node $v_j \in \mathcal{V}$ goes to zero after a finite number of iterations, and thus (from the weight updates in Algorithm~\ref{constralg1_packet}) the integer weight $f_{ji}[k]$ on each edge $(v_j, v_i) \in \mathcal{E}$ stabilizes to an integer value $f_{ji}^*$ (where $f_{ji}^* \in \mathbb{N}_0$) within the given lower and upper limits, i.e., $1 \leq l_{ji} \leq f^*_{ji} \leq u_{ji}$ for all $(v_j, v_i) \in \mathcal{E}$.

\begin{prop}
\label{PROP2_packDr}
Consider the problem formulation described in Section~\ref{ProbStatementUpperLower_delays}. 
Let $\mathcal{V}^-[k] \subset \mathcal{V}$ be the set of nodes with negative weight imbalance at iteration $k$, i.e., $\mathcal{V}^-[k] = \{ v_j \in \mathcal{V} \; | \; x_j[k] < 0 \}$.  During the execution of Algorithm~\ref{constralg1_packet}, we have that
$$
\mathcal{V}^-[k+1] \subseteq \mathcal{V}^-[k].
$$
\end{prop}

\begin{proof}
We will first argue that nodes with nonnegative {\em perceived} weight imbalance at iteration $k$ can never reach negative perceived weight imbalance at iteration $k+1$. 
Combining this with the fact that the perceived weight imbalance of a node is always below its actual weight imbalance, we establish the proof of the proposition.

Consider a node $v_j$ with a nonnegative perceived weight imbalance $x^{(p)}_j[k] \geq 0$ (from Remark~\ref{perceivedSmaller}, since $x^{(p)}_j[k] \leq x_j[k]$, $\forall \ k \geq 0$ we have that also $x_j[k] \geq 0$). \\
\noindent
We analyze below the following two cases:
\begin{enumerate}
\item at least one neighbor of node $v_j$ has positive perceived weight imbalance,
\item all neighbors of node $v_j$ have negative or zero perceived weight imbalance.
\end{enumerate}
In both cases, since $x^{(p)}_j[k] \geq 0$, node $v_j$ will attempt to change the weights of (some of) its incoming and outgoing edges. 
Specifically, node $v_j$ calculates the desirable weight $f^{(j)}_{ji}[k]$ ($f^{(j)}_{lj}[k]$) for its incoming (outgoing) edges $(v_j, v_i)$ ($(v_l, v_j)$) where $v_i \in  \mathcal{N}_j^-$ ($v_l \in  \mathcal{N}_j^+$).

In the first case, both in- and out-neighbors ($v_i$ and $v_l$ respectively) of $v_j$ will calculate the desirable weights for their incoming and outgoing edges.
Depending on the possible packet drops that might occur during the transmissions from node $v_i$ to node $v_j$, we consider the following two scenarios:
\begin{enumerate}
\item[a)] No packet is dropped,
\item[b)] At least one packet is dropped.
\end{enumerate}

\noindent
Recall that from the perceptive of node $v_j$ the following transmissions take place: first, node $v_j$ sends $f^{(j)}_{ji}[k]$ to each in-neighbor $v_i \in  \mathcal{N}_j^-$. 
Then it receives $f^{(l)}_{lj}[k]$ from every out-neighbor $v_l \in  \mathcal{N}_j^+$ and finally, once it calculates the new weights $f_{lj}[k+1]$ for its outgoing edges $(v_l, v_j)$ (where $v_l \in  \mathcal{N}_j^+$), it transmits them to every out-neighbor $v_l \in  \mathcal{N}_j^+$.

\noindent
For the first scenario (a), we have 
\begin{eqnarray}
x^{(p)}_j[k+1] & = & \sum_{v_i \in \mathcal{N}_j^-} f^{{(p)}}_{ji}[k+1] - \sum_{v_l \in \mathcal{N}_j^+} f_{lj}[k+1]\label{perc_imb_1} \\
 & = & \sum_{v_i \in \mathcal{N}_j^-} ( f^{(i)}_{ji}[k] + f^{(j)}_{ji}[k] - f_{ji}[k] ) - \nonumber \\
 & & \; \; - \sum_{v_l \in \mathcal{N}_j^+} ( f^{(j)}_{lj}[k] + f^{(l)}_{lj}[k] - f_{lj}[k] ) \; . \nonumber
\end{eqnarray}

\noindent
Since
\begin{equation}\label{perc_imb_2}
\sum_{v_i \in \mathcal{N}_j^-} f^{(j)}_{ji}[k] = \sum_{v_l \in \mathcal{N}_j^+} f^{(j)}_{lj}[k] ,
\end{equation}
(\ref{perc_imb_1}) becomes
\begin{eqnarray}
x^{(p)}_j[k+1] & = & \sum_{v_i \in \mathcal{N}_j^-} ( f^{(i)}_{ji}[k] - f_{ji}[k] ) - \nonumber \\
 & & \; \; - \sum_{v_l \in \mathcal{N}_j^+} ( f^{(l)}_{lj}[k] - f_{lj}[k] )\label{perc_imb_3} \; . 
\end{eqnarray}

\noindent
Also, since $f^{(i)}_{ji}[k] \geq f_{ji}[k]$ and $f^{(l)}_{lj}[k] \leq f_{lj}[k]$, $\forall \ (v_j, v_i), \ (v_l, v_j) \in \mathcal{E}$, we conclude $x^{(p)}_j[k+1] \geq 0, \; \forall v_j \in \mathcal{V}$.
%
%
%
%
As a result we conclude that, for scenario (a), the nonnegative {\em perceived} weight imbalance of node $v_j$ at iteration $k$ remains nonnegative at iteration $k+1$.

For scenario (b), let us assume (without loss of generality) that $f_{ji}[k+1]$, sent from node $v_i$ to node $v_j$ at Step~7 of the proposed algorithm, suffered a packet drop while all the other transmissions were successful. We have that 
{\small \begin{eqnarray}
x^{(p)}_j[k+1] & = & \sum_{v_{i'} \in \mathcal{N}_j^-} f^{(p)}_{ji'}[k+1] - \nonumber \\ 
 & & \; \; - \sum_{v_l \in \mathcal{N}_j^+} f_{lj}[k+1] \label{perc_imb_6} \\
 & = & f_{ji}^{(j)}[k] + \sum_{v_{i'} \in \mathcal{N}_j^- - \{ v_i \} } ( f_{ji'}[k+1] ) + f^{(j)}_{ji}[k] - \nonumber \\
 & & \; \; - \sum_{v_l \in \mathcal{N}_j^+} ( f^{(j)}_{lj}[k] + f^{(l)}_{lj}[k] - f_{lj}[k] ) \; , \nonumber
\end{eqnarray}
}

\noindent
which, in a similar manner, leads to the conclusion that $x^{(p)}_j[k+1] \geq 0, \; \forall v_j \in \mathcal{V}$.
%
%
%
%
Thus, for scenario (b), we conclude that if only the transmission from node $v_i$ to node $v_j$ suffered a packet drop, the nonnegative {\em perceived} weight imbalance of node $v_j$ at iteration $k$ remains nonnegative at iteration $k+1$.

The remaining scenarios, where multiple transmissions suffer packet drops during the same iteration $k$, as well as the remaining cases, where all neighbors of node $v_j$ have negative or zero perceived weight imbalance, can be argued in a similar manner.

As a result we have that during an iteration $k$ of Algorithm~\ref{constralg1_packet}, nodes with nonnegative {\em perceived} weight imbalance can never reach negative perceived weight imbalance at iteration $k+1$.
From Remark~\ref{perceivedSmaller}, since $x^{(p)}_j[k] \leq x_j[k]$, $\forall \ k \geq 0$, we have that also nodes with nonnegative weight imbalance can never reach negative weight imbalance, thus establishing the proof of the proposition.
\end{proof}

\begin{prop}
\label{PROP3_packDr}
Consider the problem formulation described in Section~\ref{ProbStatementUpperLower_delays}. During the execution of Algorithm~\ref{constralg1_packet}, it holds that
$$
0 \leq \varepsilon[k+1] \leq \varepsilon[k] \; , \; \; \forall k \geq 0 \; ,
$$
where $\varepsilon[k] \geq 0$ is the total imbalance of the network at iteration~$k$ (see Definition~\ref{defn:totalim}).

\end{prop}

\begin{proof} 
From the third statement of Proposition~\ref{constrPROP1}, we have $\varepsilon[k+1] = 2 \sum_{v_j \in \mathcal{V}^-[k+1]} \vert x_j[k+1] \vert$ and $\varepsilon[k] = 2 \sum_{v_j \in \mathcal{V}^-[k]} \vert x_j[k] \vert$, whereas from Proposition~\ref{PROP2_packDr}, we have $\mathcal{V}^-[k+1] \subseteq \mathcal{V}^-[k]$.

Consider a node $v_j \in \mathcal{V}^-[k]$ with weight imbalance $x_j[k]<0$ (obviously we have that also $x^{(p)}_j[k]<0$ since $x^{(p)}_j[k] \leq x_j[k]$, $\forall \ k \geq 0$ from Remark~\ref{perceivedSmaller}). \\
\noindent
We analyze below the following two cases:
\begin{enumerate}
\item all neighbors of node $v_j$ have negative or zero perceived weight imbalance,
\item at least one neighbor of node $v_j$ has positive perceived weight imbalance.
\end{enumerate}
In both cases, node $v_j$ will not make any weight changes on its edges. 
In the first case, the weight imbalance of node $v_j$ will not change (i.e., $x_j[k+1] = x_j[k] < 0$). 
In the second case, we have (i) $x^{(p)}_i[k] \geq 0$ for some $v_i \in  \mathcal{N}_j^-$, or (ii) $x^{(p)}_l[k] \geq 0$ for some $v_l \in  \mathcal{N}_j^+$.

For (i) we have that during the iteration $k$ of Algorithm~\ref{constralg1_packet}, the incoming edge weights of $v_j$ might change by its in-neighbors (i.e., the weight of an incoming edge $(v_j,v_i)$ might be increased to be equal to $f_{ji}[k+1] = f_{ji}^{(i)}[k]$ for some $v_i \in \mathcal{N}_j^-$).
In this case (regardless if we have a packet drop during the transmission of $f_{ji}[k+1]$ from $v_i$ to $v_j$) we have $\varepsilon[k+1] \leq \varepsilon[k]$ (using the third statement in Proposition~\ref{constrPROP1}).
For (ii) we have that during iteration $k$ of Algorithm~\ref{constralg1_packet}, the out-neighbor of $v_j$ might transmit the new outgoing edge weights to node $v_j$ (i.e., $v_j$ might receive the new $f_{lj}^{(l)}[k]$ from some $v_l \in \mathcal{N}_j^+$).
In this case, if $f_{lj}^{(l)}[k]$ suffers a packet drop, the weight imbalance of $v_j$ will not change (i.e., $x_j[k+1] = x_j[k] < 0$). 
If $f_{lj}^{(l)}[k]$ is transmitted successfully then the weight imbalance of node $v_j$ will satisfy $x_j[k+1] \geq x_j[k]$.
As a result, for both cases, we have $\varepsilon[k+1] \leq \varepsilon[k]$ (using the third statement in Proposition~\ref{constrPROP1}).
\end{proof}

\begin{prop}
\label{PROP4_packDr}
Consider the problem formulation described in Section~\ref{ProbStatementUpperLower_delays} where the Integer Circulation Conditions in Section~\ref{CircConditions} are satisfied. 
Algorithm~\ref{constralg1_packet} balances the weights in the graph in a finite number of steps, with probability one (i.e., $\exists \ k_0$ so that almost surely $\forall k \geq k_0$, $f_{ji}[k_0] = f_{ji}[k]$, $\forall (v_j,v_i) \in \mathcal{E}$ and $x_j[k] = x_j[k_0] =0$, $\forall v_j \in \mathcal{V}$). 
\end{prop}

\begin{proof} 
By contradiction, suppose Algorithm~\ref{constralg1_packet} runs for an infinite number of iterations and its total imbalance remains positive (i.e., $\varepsilon[k]>0$ for all $k$). 
During the execution of the proposed distributed balancing algorithm, packets containing information are dropped with probability $q_{lj} <1$ for each communication link $(v_l,v_j) \in \mathcal{E}$ (we assume independence between packet drops at different time steps and different links and link directions).
During transmissions on link $(v_l, v_j)$, we have that at each transmission, a packet goes through with probability $1 - q_{lj} > 0$.
Thus, if we consider $k_{lj}$ consecutive uses of link $(v_l, v_j)$, the probability that at least one packet will go through is $1 - q_{lj}^{k_{lj}}$, which will be arbitrarily close to $1$ for a sufficiently large $k_{lj}$.
\noindent
Specifically, for any (arbitrarily small) $\epsilon > 0$, we can choose 
$$
k_{lj} = \left \lceil \frac{\log \epsilon}{\log q_{lj}} \right \rceil, 
$$
to ensure that each transmission goes through by $k_{lj}$ steps with probability $1 - \epsilon$. 

Suppose now that Algorithm~\ref{constralg1_packet} runs for an infinite number of iterations (where infinite successful packet transmissions occurred on each link $(v_l, v_j)$, for a sufficiently large $k_{lj}$) and, by contradiction, its total imbalance remains positive (i.e., $\varepsilon[k]>0$ for all $k$). 
This means that always (at each $k$) there will exist at least one node with positive weight imbalance and thus the proof of this Proposition becomes identical to the proof of Proposition~\ref{constrPROP4}. 

As a result we have that if the Integer Circulation Conditions in Section~\ref{CircConditions} hold, the total imbalance $\varepsilon[k]$ decreases after a finite number of iterations, and Algorithm~\ref{constralg1_packet} results in a weight-balanced digraph after a finite number of iterations. 
\end{proof}

\section{Simulation Study}
\label{resultsupperloweralgorithm_delays}

In this section, we present simulation results and comparisons for the proposed distributed algorithms. 
Specifically, we present detailed numerical results for a random graph of size $n = 20$ and for the average of $1000$ random digraphs of $20$ and $50$ nodes each.
We illustrate the behavior of the proposed distributed algorithm for the following three different scenarios: 
(i) the scenario where Algorithm~\ref{constralg1_delays} operates in a randomly created graph of $20$ nodes where for every communication link $(v_j,v_i) \in \mathcal{E}$ there are bounded transmission delays $0 < \tau_{lj} < \overline{\tau}$ where $\overline{\tau} = 10$ (independently between different links and link directions) and each node $v_j$ transmits the desired weight change $c^{(j)}_{lj}[k]$ ($c^{(j)}_{ji}[k]$) on each outgoing (incoming) edge $(v_l,v_j)\in \mathcal{E}$ ($(v_j,v_i)\in \mathcal{E}$) to each $v_l \in \mathcal{N}_j^+$ ($v_i \in \mathcal{N}_j^-$), at each iteration $k$,
(ii) the scenario where Algorithm~\ref{constralg1_event} operates in a randomly created graph of $20$ nodes  where for every communication link $(v_j,v_i) \in \mathcal{E}$ there are bounded transmission delays $0 < \tau_{lj} < \overline{\tau}$ where $\overline{\tau} = 10$ (independently between different links and link directions) and each node $v_j$ transmits only once the desired weight change $c^{(j)}_{lj}[k]$ ($c^{(j)}_{ji}[k]$) on each outgoing (incoming) edge $(v_l,v_j)\in \mathcal{E}$ ($(v_j,v_i)\in \mathcal{E}$) to each $v_l \in \mathcal{N}_j^+$ ($v_i \in \mathcal{N}_j^-$), 
(iii) the scenario where Algorithm~\ref{constralg1_packet} operates in a randomly created graph of $20$ nodes where for every communication link $(v_j,v_i) \in \mathcal{E}$ there are packet drops with equal probability $q$ (where $0 \leq q < 1$) (independently between different links and link directions) and each node $v_j$ transmits the new weight $f_{lj}[k+1]$ ($f_{ji}^{(j)}[k]$) on each outgoing (incoming) edge $(v_l,v_j)\in \mathcal{E}$ ($(v_j,v_i)\in \mathcal{E}$) to each $v_l \in \mathcal{N}_j^+$ ($v_i \in \mathcal{N}_j^-$), at each iteration $k$.
Note that the the integer circulation conditions (presented in Section~\ref{CircConditions}) hold of all three different scenarios.

In Fig.~\ref{working20_del} we show the operation of Algorithm~\ref{constralg1_delays} in a randomly created graph of $20$ nodes where for every communication link $(v_j,v_i) \in \mathcal{E}$ there are bounded transmission delays $0 < \tau_{lj} < \overline{\tau}$ where $\overline{\tau} = 10$ (independently between different links and link directions) and each node $v_j$ transmits the desired weight change $c^{(j)}_{lj}[k]$ ($c^{(j)}_{ji}[k]$) on each outgoing (incoming) edge $(v_l,v_j)\in \mathcal{E}$ ($(v_j,v_i)\in \mathcal{E}$) to each $v_l \in \mathcal{N}_j^+$ ($v_i \in \mathcal{N}_j^-$), at each iteration $k$.
In the first case, we plot the \textit{absolute imbalance} $\varepsilon = \sum_{j=1}^{n} \vert x_j \vert$, $\forall v_j \in \mathcal{V}$ (blue line) and the perceived total imbalance $\varepsilon^{(p)} = \sum_{j=1}^{n} \vert x^{(p)}_j \vert$ (red line) against the number of iterations $k$.
In the second case the \textit{nodes balances} $x_j[k]$ (as defined in Definition~\ref{DEFnodebalance}) as a function of the number of iterations $k$ for the distributed algorithm. 
Here, the plot suggests that the absolute imbalance $\varepsilon$ becomes equal to zero after a finite number of iterations, which means that Algorithm~\ref{constralg1_delays} is able to obtain a set of integer weights that balance the corresponding digraph after a finite number of iterations in the presence of bounded transmission delays $0 < \tau_{lj} < \overline{\tau}$, where $\overline{\tau} = 10$, on each link $(v_l,v_j)\in \mathcal{E}$.

\begin{figure}[h]
\begin{center}
\includegraphics[width=0.70\columnwidth]{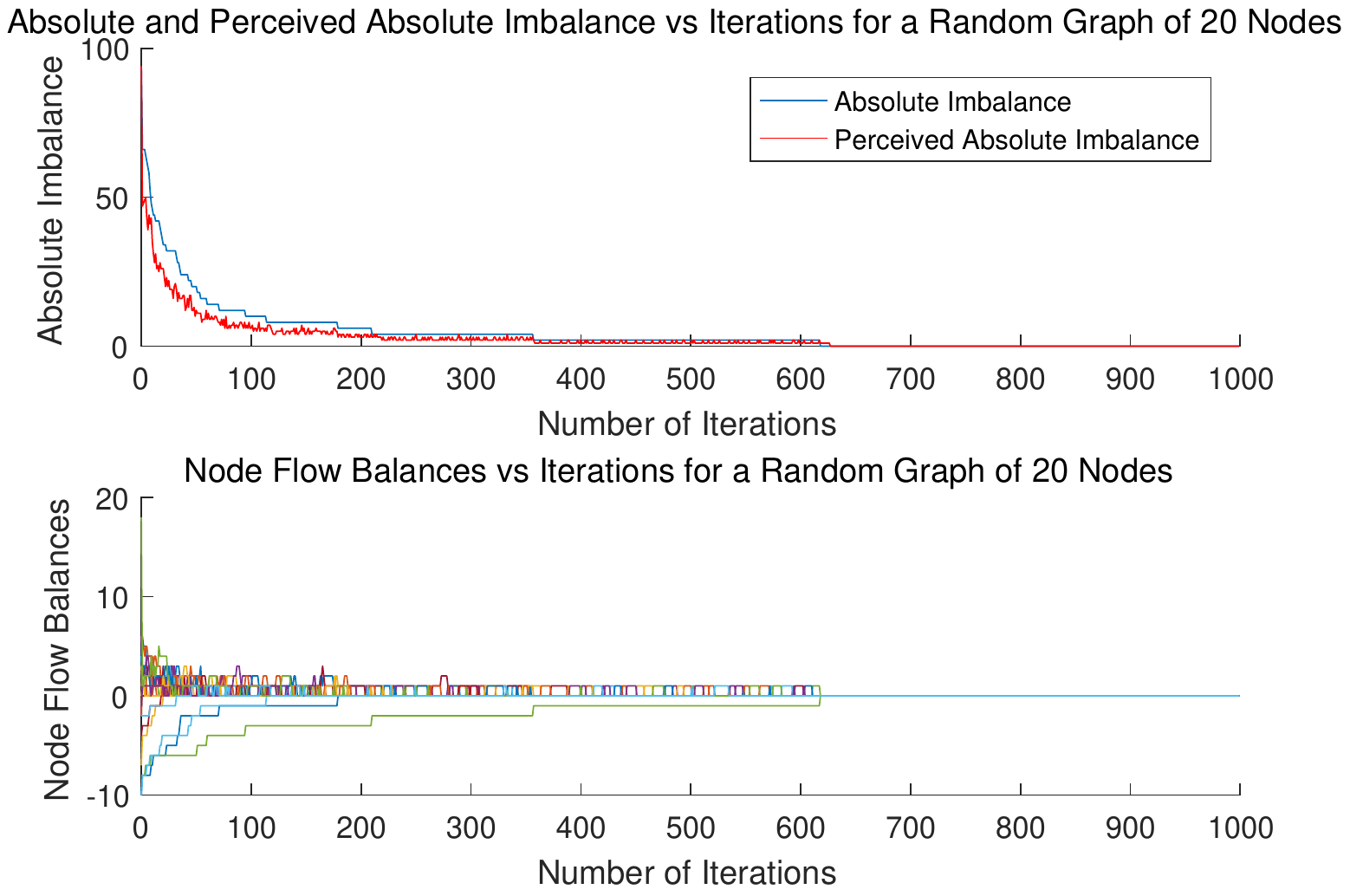}
\caption{Execution of Algorithm~\ref{constralg1_delays} for the case when the integer circulation conditions hold for a random graph of $20$ nodes with transmission delays $0 < \tau_{lj} < \overline{\tau}$ where $\overline{\tau} = 10$. \emph{Top figure:} Total (absolute) imbalance $\varepsilon[k]$ (blue line) and perceived total imbalance $\varepsilon^{(p)}[k]$ (red line) plotted against number of iterations. \emph{Bottom figure:} Node weight imbalances $x_j[k]$ plotted against number of iterations.}
\label{working20_del}
\end{center}
\end{figure}

In Fig.~\ref{working20_eventTr} we show the operation of Algorithm~\ref{constralg1_event} for the same case as Fig.~\ref{working20_del}.
Here the plot suggests that Algorithm~\ref{constralg1_event} is able to obtain a set of integer weights that balance the corresponding digraph after a finite number of iterations in the presence of bounded transmission delays $0 < \tau_{lj} < \overline{\tau}$, where $\overline{\tau} = 10$, on each link $(v_l,v_j)\in \mathcal{E}$, for the case where each node $v_j$ transmits only once the desired weight change.

\begin{figure}[h]
\begin{center}
\includegraphics[width=0.70\columnwidth]{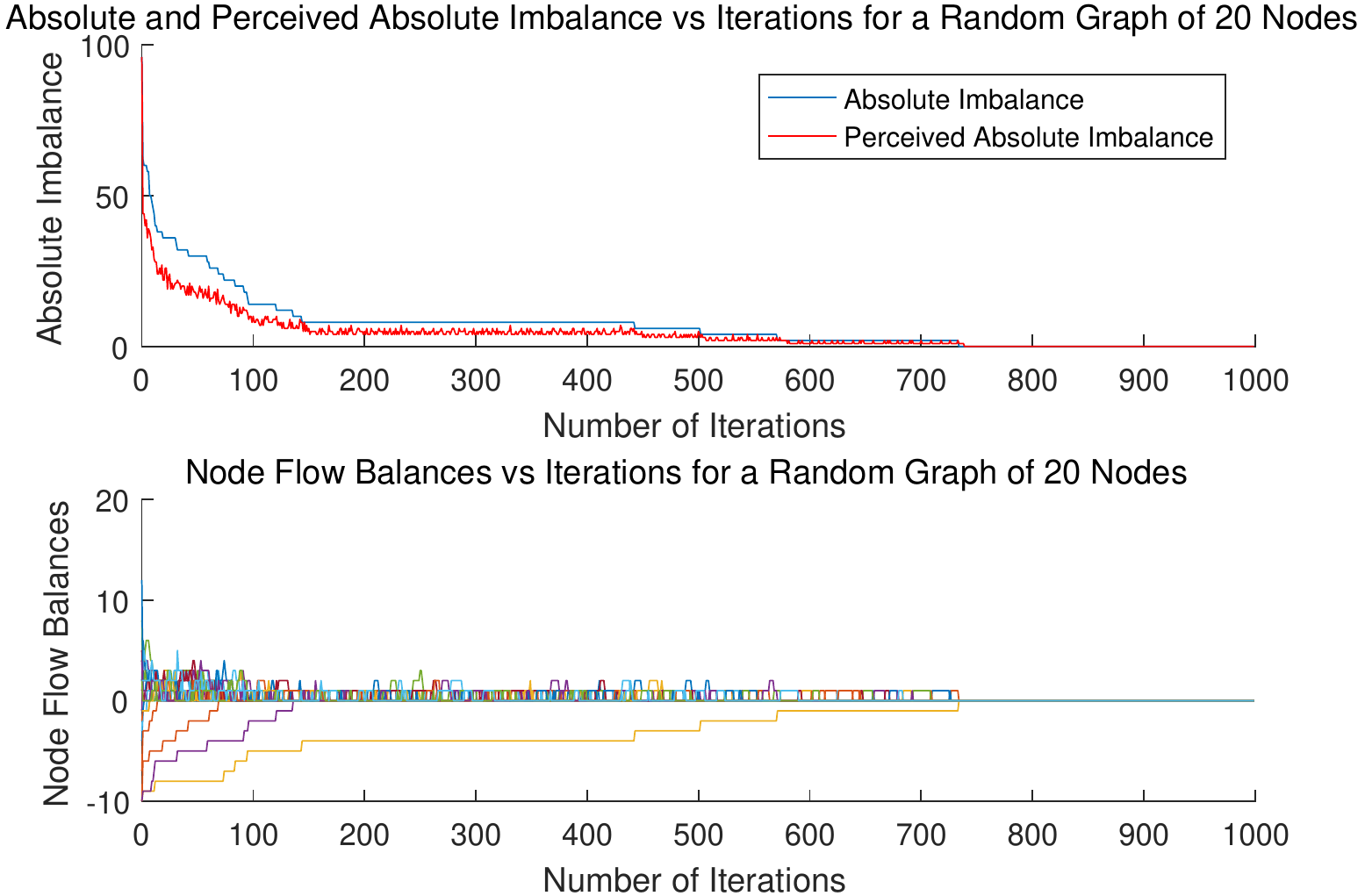}
\caption{Execution of Algorithm~\ref{constralg1_event} for the case when the integer circulation conditions hold for a random graph of $20$ nodes with transmission delays $0 < \tau_{lj} < \overline{\tau}$ where $\overline{\tau} = 10$. \emph{Top figure:} Total (absolute) imbalance $\varepsilon[k]$ (blue line) and perceived total imbalance $\varepsilon^{(p)}[k]$ (red line) plotted against number of iterations. \emph{Bottom figure:} Node weight imbalances $x_j[k]$ plotted against number of iterations.}
\label{working20_eventTr}
\end{center}
\end{figure}

In Fig.~\ref{working20_packetDr} we show the operation of Algorithm~\ref{constralg1_packet} for the same cases as Figs.~\ref{working20_del} and \ref{working20_eventTr}. 
The plot suggests that Algorithm~\ref{constralg1_packet} is able to obtain a set of integer weights that balance the corresponding digraph after a finite number of iterations in the presence of packet dropping links with probability $q = 0.8$, on each link $(v_l,v_j)\in \mathcal{E}$.

\begin{figure}[h]
\begin{center}
\includegraphics[width=0.70\columnwidth]{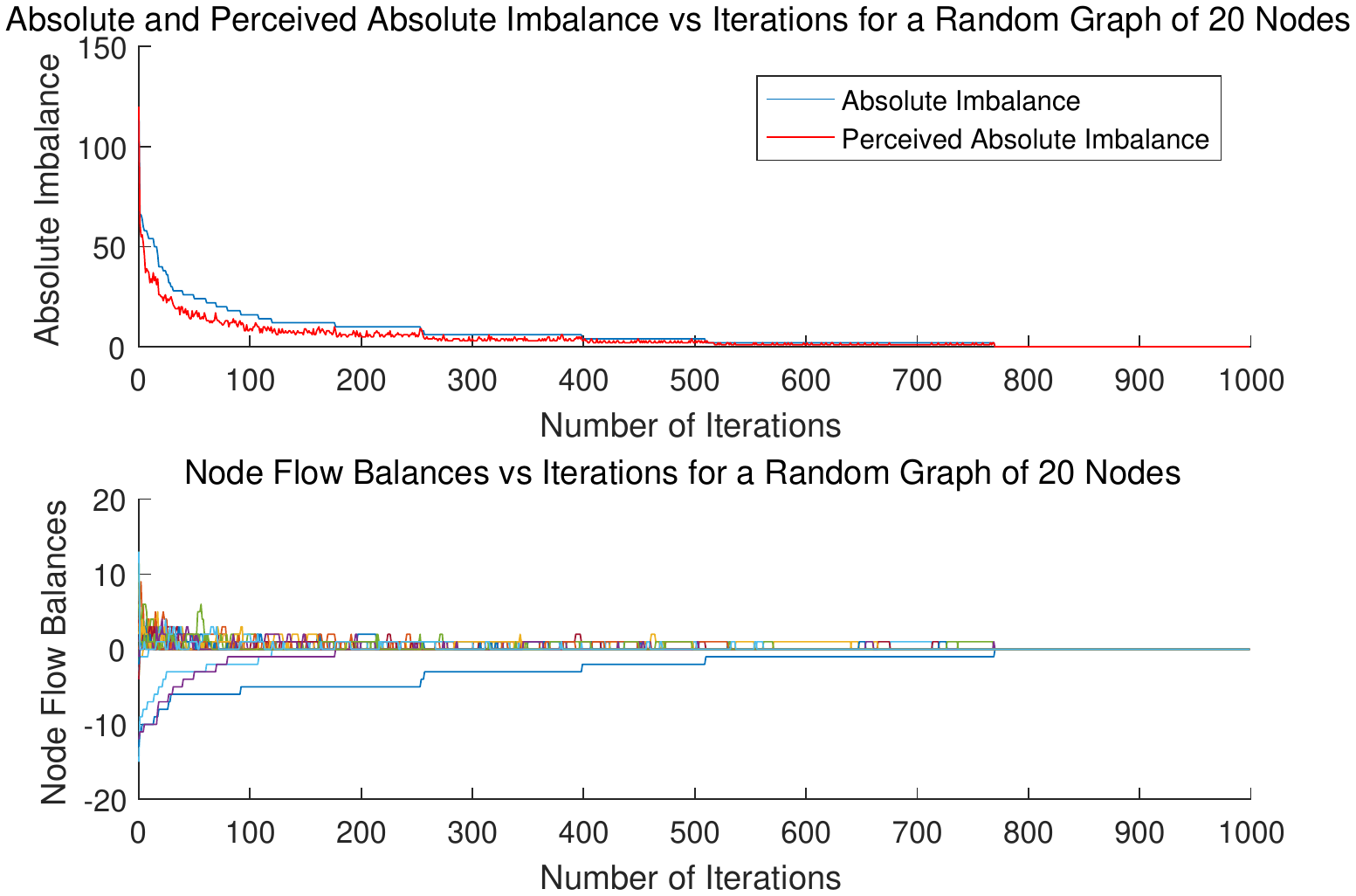}
\caption{Execution of Algorithm~\ref{constralg1_packet} for the case when the integer circulation conditions hold for a random graph of $20$ nodes with packet drop probability $q_{ji} = 0.8$. \emph{Top figure:} Total (absolute) imbalance $\varepsilon[k]$ (blue line) and Perceived Total Imbalance $\varepsilon^{(p)}[k]$ (red line) plotted against number of iterations. \emph{Bottom figure:} Node weight imbalances $x_j[k]$ plotted against number of iterations.}
\label{working20_packetDr}
\end{center}
\end{figure}


\section{Chapter Summary}
\label{SummaryUpperLower_delays}

In this chapter, we introduced and analyzed a novel distributed algorithm which achieves integer weight balancing in a multi-component system under lower and upper constraints on the edge weights in the presence of time delays over the communication links.  
We analyzed its functionality, established its correctness and showed that it achieves integer weight balancing after a finite number of steps. 
We also added extensions to handle the cases of packet drops over the communication links and event-triggered operation where we showed that in both scenarios, the proposed algorithm converges (with probability one) to a set of weights that form a balanced graph after a finite number of iterations
We also demonstrated the operation, performance, and advantages of the proposed algorithm via various simulations.

\clearpage

\lhead{\emph{Conclusions and Future Directions}}

\chapter{Conclusions \\ and Future Directions}
\label{conclusions}

\section{Conclusions}

In this thesis, we presented distributed algorithms for weight balancing over a static directed graph.
In Chapter~\ref{centrvsdistr} we presented a distributed algorithm which deals with the problem of balancing a weighted digraph. 
The proposed distributed algorithm operates by having each node compute its weight imbalance and then increase the integer weights of its outgoing edges so that it becomes weight balanced. Specifically, the outgoing edges are assigned, if possible, equal integer weights; otherwise, if this is not possible, they are assigned integer weights such that the maximum difference among them is equal to unity. We showed that our distributed algorithm results in a weight balanced digraph after a finite number of iterations (bounded in the worst case by $O(n^7)$) and we also carried out numerical simulations to illustrate the operation and potential advantages of the proposed distributed algorithm.

In Chapter~\ref{distrdelpacket}, we presented a novel distributed algorithm which deals with the problem of balancing a weighted digraph in the presence of time delays (bounded by a maximum value $\overline{\tau}$) and packet drops over the communication links. This algorithm operates by having each node compute its delayed weight imbalance according to the latest received weight values from its in-neighbors. Then, if it has positive (delayed) imbalance, it increases by $1$ the integer weights of its outgoing edges one at a time, following a fixed priority order (in a round robin fashion) until it becomes weight balanced. This means that the outgoing edges are assigned, if possible, equal integer weights; otherwise, if this is not possible, they are assigned integer weights such that the maximum difference among them is equal to one. We showed that our distributed algorithm converges to a set of weights $f_{lj}=f_{lj}^*, \forall (v_l,v_j) \in \mathcal{E}$, after $O(n^6\overline{\tau})$ iterations (where $f^*_{lj}$ is a set of weights that form a weight balanced digraph) after a finite number of steps bounded by $O(n^6)$ under no delays ($\overline{\tau}=0$). Then, we extended this result for the case where communication links could also result in possible packet drops (i.e., unbounded delays) in the corresponding communication network where we showed that the proposed distributed algorithm converges, with probability one, to a set of weights $f_{lj} = f_{lj}^*, \ \forall (v_l,v_j) \in \mathcal{E}$, after a finite number of iterations despite the presence of packet drops occurring with probability $q_{lj}$, where $f^*_{lj}$ is the set of weights that form a weight balanced digraph and are obtained after a finite number of steps bounded by $O(n^6)$ under no packet drops ($q_{lj}=0$). Following these developments, we presented an event-triggered version of the proposed distributed algorithm where each agent autonomously decides when communication and control updates should occur so that the resulting network executions still result to a set of weights $f_{lj} = f_{lj}^*, \ \forall (v_l,v_j) \in \mathcal{E}$, after a finite number of steps bounded by $O(n^6\overline{\tau})$ iterations (where the set of weights $f^*_{lj}$ is the set of weights obtained by the nominal algorithm that runs with no even-triggering and no delays). We also carried out numerical simulations to show the operation and potential advantages of the proposed distributed algorithm.

In Chapter~\ref{constraintsbalancing}, we presented a novel distributed algorithm which deals with the problem of balancing a weighted digraph in the presence of upper and lower weight constraints over the communication links. Our distributed algorithm operates by having each node compute its weight imbalance according to the weight values from its out- and in-neighbors. Then, if it has positive imbalance it attempts to add $+1$ (or subtract $-1$) to its outgoing (or incoming) integer weights one at a time, according to a predetermined (cyclic) order, in a round robin fashion, until its weight imbalance becomes zero. Each node transmits the amount of change it calculated on each outgoing (or incoming) edge while it receives the amount of change calculated by its out- and in-neighbors; it then assigns integer weights on its incoming and outgoing edges with respect to the corresponding upper and lower weight constraints. We showed that our distributed algorithm results in a weight balanced digraph after a finite number of iterations and carried out numerical simulations to illustrate the operation and potential advantages of the proposed distributed algorithm.

Finally, in Chapter~\ref{constraintsbalancing_delays}, we presented a novel distributed algorithm which deals with the problem of balancing a weighted digraph within the allowable edge weight intervals in the presence of time delays and packet drops over the communication links. 
Our distributed algorithm operates by having each node compute its perceived weight imbalance according to the latest received weight values from its in-neighbors. 
When the communication links are subject to time delays, a node has positive perceived imbalance, it calculates the desired change amount for each incoming and outgoing links by adding $+1$ (or subtracting $1$) to its outgoing (or incoming) integer weights one at a time, according to a predetermined (cyclic) order until its weight imbalance becomes zero.  
The node subsequently transmits the amount of change it calculated on each outgoing (or incoming) edge while it receives the amount of change calculated by its out- and in-neighbors, and assigns integer weights on its incoming and outgoing edges with respect to the corresponding upper and lower weight constraints. 
We showed that our distributed algorithm results in a weight balanced digraph after a finite number of iterations. 
The operation of the proposed algorithm was extended for the case where we have event-driven actuators,
enabling a more efficient use of the available resources. 
Specifically, we presented an event-triggered operation of the proposed distributed algorithm where each agent autonomously decides when communication and control updates should occur so that the resulting executions still result to a weight balanced digraph after a finite number of iterations. 
Finally, the operation of the proposed algorithm was extended for the case where communication links could also result in possible packet drops (i.e., unbounded delays) in the corresponding communication network. 
In this case, we showed that the proposed distributed algorithm converges, with probability one, to a weight balanced digraph after a finite number of iterations. 
We also carried out numerical simulations to illustrate the operation and potential advantages of the proposed distributed algorithm.

\section{Future Directions}

In this thesis we have that weight-balanced graphs/matrices play an important role in the analysis and convergence of distributed coordination algorithms. 
The algorithms introduced in this thesis can also be used for other applications. 
A distributed algorithm which deals with the problem of balancing a weighted digraph, introduced and analyzed in Chapter~\ref{centrvsdistr}, can be used for the case where each agent in the network wants to calculate a common quantized value equal to the exact average of the initial values (i.e., the nodes need to reach quantized consensus). 
Specifically, by assuming that each node has two initial values (the quantized measurement along with the value $1$) we can implement a ``mass summation'' algorithm in which every node sums the incoming values and then directly transmits them to an out-neighbor, chosen according to the predetermined priority order. 
This iteration will allow each agent to obtain two integer values, the ratio of which is equal to the average of the initial values of the nodes.
The extension towards quantized average consensus has applications in capacity and memory constrained sensor networks, load balancing in processor networks, and others.

In Chapter~\ref{constraintsbalancing} we presented a novel distributed algorithm which deals with the problem of balancing a weighted digraph in the presence of upper and lower weight constraints over the communication links and we extended its operation, in Chapter~\ref{constraintsbalancing_delays}, for the cases where we have time delays and packet drops over the communication links. 
The presence of constraints over the communication links means that the differences of the resulting link weights (that form a weight balanced digraph) are also constrained and depend on the upper and lower weight intervals and the graph structure. 
Given a weight balanced digraph, one can use a finite-time algorithm, based on max-consensus, to obtain a doubly stochastic matrix which find applications in distributed averaging. 
In this matrix the link weight differences are also constrained (since they depend on the weight balanced digraph) and affect the asymptotic convergence rate towards average consensus. 
Thus, the design of a distributed algorithm which calculates a set of weights that form a weight balanced digraph in the presence of upper and lower weight constraints and minimizes the link weight differences is an important open problem, which will allow us to perform distributed averaging under the maximum possible rate of convergence (i.e., the convergence rate under which every node reaches, asymptotically, average consensus will be the maximum possible).

The operation of the proposed distributed algorithm, presented in Chapters~\ref{constraintsbalancing} and \ref{constraintsbalancing_delays}, can also be extended to the case where nodes with negative weight imbalance also attempt to change the integer weights in both its incoming edges and its outgoing edges with respect to the corresponding upper and lower weight constraints. 
Depending on the graph structure, this could improve the speed under which we are able to obtain a set of weights that form a weight balanced digraph, which means that once the nodes have reached a weight balanced digraph, they can move on to some other distributed computation (e.g., average consensus).
Furthermore, the proposed algorithm operation can be extended to handle more realistic scenarios in which the edges are able to take values in multiple spaces. 
This extension is highly important since, firstly, it can handle the case where each edge may suffer damages and will be unable to obtain a certain range of values, and secondly, a possible solution of this problem will lead to the definition of a new circulation theorem, thus improving and extending the one presented in Section~\ref{CircConditions}, possibly leading to improved versions of algorithms which deal with the standard and maximum flow problem, auction problem, and energy minimization problem.
Finally, the proposed algorithm operation relies on bi-directional communication (i.e., the communication topology is captured by the undirected graph that corresponds to the network digraph). 
This assumption may not be valid for applications which require directed communication since transmitting and receiving information requires energy, which is typically a sparse commodity in many networked applications, such as sensor networks and mobile ad hoc communication networks.
Thus, the extension of the operation of the proposed algorithm for the case when the communication topology matches exactly the physical topology (i.e., the communication topology is captured by a directed graph) is an open problem which will extend the range of applications towards more realistic scenarios.

\clearpage



\addtocontents{toc}{\vspace{2em}} 

%



\addtocontents{toc}{\vspace{2em}}  
\backmatter

\label{Bibliography}
\lhead{\emph{Bibliography}}  
\bibliographystyle{IEEEtran}  
\bibliography{Bibliography}  

\end{document}